\documentclass[11pt]{amsart}
\usepackage{amsmath, amsthm, amssymb, amsfonts, enumerate}
\usepackage[colorlinks=true,linkcolor=blue,urlcolor=blue]{hyperref}

\usepackage[normalem]{ulem}
\usepackage{color}
\usepackage{geometry}
\usepackage{soul}
\usepackage{graphicx}
\usepackage{bm}
\usepackage{caption}
\usepackage{subcaption}
\usepackage[authoryear]{natbib}

\allowdisplaybreaks[1]

\newtheorem{theorem}{Theorem}[section]
\newtheorem{remark}[theorem]{Remark}
\newtheorem{assumption}[theorem]{Assumption}
\newtheorem{lemma}[theorem]{Lemma}
\newtheorem{proposition}[theorem]{Proposition}
\newtheorem{corollary}[theorem]{Corollary}

\DeclareMathOperator{\interior}{int}

\def \cC{{\mathcal C}}
\def \cD{{\mathcal D}}
\def \cE{{\mathcal E}}
\def \cF{{\mathcal F}}

\def \cI{{\mathcal I}}
\def \cK{{\mathcal K}}
\def \cL{{\mathcal L}}
\def \cO{{\mathcal O}}

\def \cR{{\mathcal R}}

\def \E{\mathsf{E}}
\def \P{\mathsf{P}}
\def \R{\mathbb{R}}

\def \N{\mathbb{N}}

\def \tc{\tilde{c}}
\def \eps{\varepsilon}

\def \ss{\mathring{\sigma}}
\def \d{{\mathrm d}}
\def \cDp{{\cK}}
\def \cCp{{\mathcal A}}

\DeclareMathOperator*{\esssup}{ess\,sup}
\def \mathds {{\bf}}
\newcommand{\rom}[1]{\uppercase\expandafter{\romannumeral #1\relax}}

\definecolor{brightmaroon}{rgb}{0.7, 0.23, 0.2}

\title[The American put stochastic interest rate]{The American put with finite-time maturity \\ and stochastic interest rate}
\author[C.~Cai]{Cheng Cai}
\author[T.~De Angelis]{Tiziano De Angelis}
\author[J.~Palczewski]{Jan Palczewski}
\address{C.~Cai: School of Mathematics, University of Leeds, Woodhouse Lane, LS2 9JT Leeds, UK.}
\email{\href{mailto:mmcca@leeds.ac.uk}{mmcca@leeds.ac.uk}}
\address{T.~De Angelis: School of Management and Economics, Dept.\ ESOMAS, University of Turin, C.so Unione Sovietica 218bis, 10134, Turin, ITALY. Collegio Carlo Alberto, P.za Arbarello 8, 10122, Turin, ITALY.}
\email{\href{mailto:tiziano.deangelis@unito.it}{tiziano.deangelis@unito.it}}
\address{J.~Palczewski: School of Mathematics, University of Leeds, Woodhouse Lane, LS2 9JT Leeds, UK.}
\email{\href{mailto:j.palczewski@leeds.ac.uk}{j.palczewski@leeds.ac.uk}}
\date{\today}

\numberwithin{equation}{section}

\begin{document}

\begin{abstract}
In this paper we study pricing of American put options on a non dividend-paying stock in the Black and Scholes market with a stochastic interest rate and finite-time maturity. We prove that the option value is a $C^1$ function of the initial time, interest rate and stock price. By means of It\^o calculus we rigorously derive the option value's early exercise premium formula and the associated hedging portfolio. We prove the existence of an optimal exercise boundary splitting the state space into {\em continuation} and {\em stopping} region. The boundary has a parametrisation as a jointly continuous function of time and stock price, and it is the unique solution to an integral equation which we compute numerically. Our results hold for a large class of interest rate models including CIR and Vasicek models. We show a numerical study of the option price and the optimal exercise boundary for Vasicek model.
\end{abstract}

\maketitle

\section{Introduction}
Pricing of American options is a classical problem in mathematical finance which has attracted continuous attention since the initial work of \citet{mckean1965free}. Its study has also become a benchmark for methodological developments of optimal stopping theory and the associated free boundary problems. In this paper we contribute to this strand of research by studying the American put option on a Black and Scholes market with a stochastic interest rate and finite-time maturity. The stock price and the interest rate are driven by (possibly) correlated Brownian motions and we make minimal assumptions about the dynamics of the interest rate under the pricing measure: the coefficients are time independent and Lipschitz continuous. CIR model, which does not satisfy these conditions, is also included in our analysis.

It is well known \citep{bensoussan1984theory, karatzas1988pricing} that the American put option price is given by the \emph{value function} of a related optimal stopping problem. In our model, this optimal stopping problem is 3-dimensional with 2-dimensional diffusive dynamics (stock price and interest rate) and time. The stopping set, i.e., the set of points $(t,r,x)$ in which it is optimal to exercise the option, is separated from the continuation set, where it is optimal to hold (or sell) the option, by a single surface called the \emph{stopping boundary}. The value function is a classical solution to a PDE in the interior of the continuation set, i.e., it is twice continuously differentiable in $(r,x)$ and once continuously differentiable in $t$, whereas it coincides with the put payoff in the stopping set. 

One of our technical contributions is to establish by means of probabilistic methods that the value function is globally once continuously differentiable in all variables. Then, the continuity of the gradient of the value function permits the application of a generalisation of It\^{o}'s formula (due to \cite{cai2021b}) and a rigorous derivation of a hedging portfolio. The hedging portfolio invests in three instruments: the money market (savings) account, the zero-coupon bond with maturity equal to the maturity of the option and the stock. We show that the usual Delta hedging strategy is optimal: the positions in the bond and the stock are given by relevant partial derivatives of the value function. As a further consequence of the generalised It\^o's formula we also derive the decomposition of the American option price as the sum of the price of a European put option with the same maturity and the same exercise price, and an \emph{early exercise premium}. This is known in the literature as the {\em early exercise premium formula}, which corresponds to Doob's decomposition of supermartingales into a martingale and a non-increasing process (applied here to the Snell envelope of the optimal stopping problem).

Our second contribution concerns the continuity properties of the stopping boundary in our model, which have not been established in the literature. We are able to demonstrate that the stopping boundary, when parametrised as a function of $(t,x)$, is continuous. Apart from being of interest in its own right, this enables a characterisation of the stopping boundary as the unique continuous solution of an integral equation arising from the early exercise premium decomposition. 
When a stopping boundary is known, efficient numerical methods are at disposal for computation of the option price. One can use Monte Carlo methods based on the early exercise decomposition or classical PDE methods for Cauchy problems (in contrast to the original problem with a free boundary). 

American option pricing with stochastic interest rates has already attracted a lot of attention in the literature, mainly focusing on approximations and numerical methods. Lattice (tree) based methods are employed by \cite{appolloni2015robust} to price options in Black and Scholes model with CIR interest rate dynamics and by \cite{battauz2019american} in a model with Vasicek interest rates. \cite{geske1984american}'s approximation of discretely exercised American options prices is adapted by \cite{ho1997valuation} and \cite{Chung2000} to a class of stochastic interest rate models that lead to log-normally distributed bond prices. An alternative approximation is provided by \cite{menkveld2000pricing}. A framework for option pricing with Heath, Jarrow, Morton's bond market model \citep{heath1992bond} is developed by \cite{amin1992pricing} with a binomial-tree-based implementation of pricing of foreign exchange options performed in \cite{amin1995discrete}. 

\cite{detemple2002valuation} study the pricing of American options in a general diffusive model with a $d$-dimensional Brownian motion. They formulate assumptions under which there is a single exercise surface but without proving its continuity. In a Black and Scholes market model with Vasicek interest rates they show that this exercise boundary solves an integral equation of the same form as in this paper. The uniqueness of solutions to this integral equation is not discussed and their numerical method for computing the solution is different to ours.

Hedging underlies the success of mathematical finance in derivatives markets. A rigorous theory that links hedging of American options with solutions of optimal stopping problems was initiated by \cite{bensoussan1984theory} using PDE methods and extended by \cite{karatzas1988pricing} to more general models and payoffs thanks to the martingale theory of optimal stopping. A hedging strategy for an American option consists of an investment portfolio and a non-decreasing cumulative consumption process which increases only when the state-time process is in the stopping set. In the Black and Scholes model with constant interest rate the classical Delta hedge is known to replicate the option \citep[Thm.\ 7.9, Ch.\ 2]{karatzas1998methods}. This paper seems to be the first to rigorously derive the hedging strategy for American put options on a market with a stochastic interest rate. This is accomplished thanks to the $C^1$-regularity of the value function that we are able to prove and which did not appear in previous works. 

A characterisation of an optimal stopping boundary as solution to a (system of) integral equations has been known since the earliest works (see \cite{van1976optimal}). In more recent works \cite{kim1990analytic, jacka1991optimal, carr1992alternative, myneni1992pricing}, the stopping boundary for the classical Black and Scholes market with constant interest rate is shown to be the unique solution to an uncountable {\em system} of integral equations arising from the early exercise premium decomposition of the option price. A break-through came with the work of \cite{peskir2005american} where he shows that the stopping boundary is the unique continuous solution of a single integral equation. His key observation is that the integral equation only needs to be satisfied for stock prices at the boundary while earlier results required that it does so for all stock prices at and below the boundary. \cite{peskir2005american}'s integral equation opens doors to side-stepping the computation of the value function in the process of determining the optimal exercise strategy; see numerical methods designed in \cite{Little2000, Kim2013}. Our paper extends \cite{peskir2005american}'s results to the market with a stochastic interest rate and the optimal boundary being a two-dimensional surface. It is also the continuity of the boundary that allows us to establish the uniqueness of solutions to the integral equation. A closely related paper that furthermore motivated our numerical approach is \cite{Kitarbayev2018} where the authors solve an integral equation for Black and Scholes market with stochastic volatility.

The regularity of the value function in one-dimensional optimal stopping problems is often phrased as {\em smooth-fit} and plays a major role in determining explicit solutions. In a Black and Scholes model with constant interest rate, smooth-fit for American options with finite-time maturity is understood as continuous differentiability of the value function with respect to the stock price, for each fixed value of the time variable (see \cite{jacka1991optimal} and subsequent works). That is a `directional' derivative and continuity is only considered with respect to one variable. Sobolev space regularity is studied in \cite{jaillet1990variational} for American options on multiple assets and {\em deterministic}, time-dependent discount rate under the assumption of uniform ellipticity of the associated second order differential operator. By Sobolev embedding it is possible to determine continuous differentiability of the value function with respect to the initial values of all the assets but not with respect to time. Continuous differentiability with respect to time and stock price for the value of the American put with finite-time maturity and constant interest rate is obtained in \cite{de2018global} along with other complementary findings about continuous differentiability of the value function for a large class of optimal stopping problems. In this paper, we refine the arguments from \cite{de2018global} and remove global integrability conditions that may not hold in our set-up.

The early exercise premium formula for American options was studied in great generality, in non-Markovian problems beyond the setting of the American put option by \cite{rutkowski1994early} with methods from martingale theory. The nature of the methods employed in \cite{rutkowski1994early} to derive his main results is such that the emphasis is removed from the optimal boundary, which in fact only appears in specific examples (see Sec.\ 3 of that paper) as a time-dependent function. Here instead we derive the early exercise premium formula starting from the analysis of the optimal boundary (and its regularity) as a function of time and one stochastic factor from our two-factor model.

Some of the ideas in this paper find wider applicability in optimal stopping theory. The generalisation of It\^o's formula that we use to find the hedging portfolio has natural applications to other optimal stopping problems as discussed extensively in Section 3 of \cite{cai2021b}. Localisation of the arguments from \cite{de2018global} to prove continuous differentiability of the value function does not rely much on the specific structure of our problem and suggests a general recipe to address the issue. Finally, our ideas for the continuity of the optimal stopping surface have been expanded upon to cover more general settings in \cite{cai2021a}.

The paper is structured as follows. Section \ref{sec:formulation} introduces the market model, main assumptions and notation. The main contributions are discussed in Section \ref{sec:main} while their proofs are delayed until after Section \ref{sec:numerics}. A numerical study with interest rates following Vasicek model is presented in Section \ref{sec:numerics} along with a sensitivity analysis.  Monotonicity and Lipschitz continuity of the value function is proved in Section \ref{sec:lipschitz}. Existence of the stopping surface and its regularity (in the sense of diffusions) are shown in Section \ref{sec:properties_boundary}. In Section \ref{sec:cont} we prove that the value function is continuously differentiable on the whole domain. Auxiliary estimates needed for admissibility of the hedging strategy are provided in Section \ref{sec:hedge}. Three appendices contain further details.

\section{Problem formulation}\label{sec:formulation}

Let $(\Omega, \cF, \P)$ be a complete probability space carrying two correlated Brownian motions $(B_t)_{t\ge0}$ and $(W_t)_{t\ge0}$ with $\E(W_tB_t)=\rho\, t$ for all $t\ge0$ and a fixed $\rho\in(-1,1)$  (here $\E(\,\cdot\,)$ is the expectation under $\P$). We denote by $(\cF_t)_{t\ge0}$ the filtration generated by $(B,W)$ augmented with the $\P$-null sets.  On this probability space we consider a financial market with one risky asset $(X_t)_{t \ge 0}$ and a bond. The asset and the risk-free (short) rate $(r_t)_{t \ge 0}$ take values, respectively, in intervals $\R_+ := (0, \infty)$ and $\cI\subseteq \R$, and follow the dynamics
\begin{align}
\label{eq:X} &d X_t=r_t X_t dt+\sigma X_tdB_t,\qquad X_0=x,\\[+4pt]
\label{eq:r} &d r_t=\alpha(r_t)dt+\beta(r_t)dW_t,\qquad r_0=r,
\end{align}
with $\alpha,\beta:\cI\to\R$ specified below. The probability measure $\P$ is a risk neutral measure for this market. We denote by $T > 0$ a fixed finite trading horizon.

Throughout the paper we assume $\sigma>0$ and $\cI=(\underline r,\overline r)$ (with $\cI$ possibly unbounded). The right boundary $\overline r$ is unattainable in a finite time (it is a natural or entrance-not-exit boundary). The left boundary $\underline r$ is either unattainable or reflecting. It will become clear later that the exact behaviour of the interest rate process at this boundary is irrelevant for the majority of results and their proofs. For the dynamics of the interest rate our benchmark example is the CIR model, but, with a relatively small additional effort, our results cover other stochastic interest rate models, e.g., Vasicek model. Therefore, we make the following standing assumption:
\begin{assumption}\label{ass:coef}
The coefficients $\alpha$ and $\beta$ in \eqref{eq:r} meet one of the conditions below:
\begin{itemize}
\item[(i)] \emph{(CIR model)} For $\kappa,\theta,\gamma>0$ we have $\alpha(r)=\kappa(\theta-r)$ and $\beta(r)=\gamma\sqrt{r}$.
\item[(ii)] $\alpha$ and $\beta$ are globally Lipschitz and continuously differentiable on bounded subsets of $\cI$ with $\beta(r)>0$ for all $r\in\cI$, and $\overline r > 0 \ge \underline r$. 
For any compact set $\cK\subset\cI$, and any $p\in[1,p']$ for some $p' > 2$ and $T>0$, there is $C_1>0$ (depending on $T$, $p$ and $\cK$) such that
\begin{align}\label{eq:integr}
\sup_{r \in \cK} \E\left[\sup_{0\le s\le T }e^{-p\int_0^s r_udu}\Big| r_0 = r\right]\le C_1.
\end{align}
\end{itemize}
\end{assumption}
The assumption that $\overline r > 0$ cannot be relaxed without trivialising the pricing problem. A strictly positive lower boundary $\underline r$ could, however, be of interest. For the clarity of presentation, it is omitted but it can be studied with similar methods as those developed in this paper.

The above assumptions are sufficient to guarantee that \eqref{eq:r} admits a unique strong solution defined on $\cI$. In the case of CIR model, we also have $\kappa\theta>0$ which implies that the spot rate is non-negative (but not necessarily strictly positive), see e.g.~\cite[Sec.~6.3.1]{jeanblanc2009mathematical}, so the left boundary $\underline r = 0$ is reflecting (also non-attainable if $\kappa \theta > \sigma^2/2$). Hence, the bound \eqref{eq:integr} is satisfied with the constant $C_1=1$. The linear growth of $\alpha$ and $\beta$ in \eqref{eq:r} guarantees that for each $p\ge 2$ there is $C_2>0$ only depending on $T$ and $p$, such that \citep[Ch.~2, Sec.~5, Thm.~9]{krylov}
\begin{align}\label{eq:subg}
\E\big[\sup_{0\le s\le T}|r_s|^p\ \big|\, r_0 = r\big]\le C_2(1+|r|^p),\qquad\text{for $r\in\cI$}.
\end{align}

Under Assumption \ref{ass:coef}, the solution of \eqref{eq:X} may be expressed as
\begin{align}\label{eq:X2}
X_t=x\exp\left(\sigma B_t+\int_0^t\big(r_s-\tfrac{\sigma^2}{2}\big)ds\right), \qquad\text{for $t\ge0$,}
\end{align}
so that $X$ depends on both initial values $r$ and $x$. On the contrary, the dynamics of the interest rate does not depend on the initial asset value. The coupling between the processes $(r_t)_{t\ge0}$ and $(X_t)_{t\ge0}$ stems from formula \eqref{eq:X2} and the correlation between the Brownian motions. To keep track of the dependence of the processes on their initial values, in what follows we often use the notation $(r^r_t,X^{r,x}_t)_{t\ge 0}$ for the process started at $r^r_0=r$ and $X^{r,x}_0=x$. Also we may sometimes use the notation $\P_{t,r,x}(\,\cdot\,)=\P(\,\cdot\,|r_t=r,X_t=x)$, $\P_{r,x}=\P_{0,r,x}$, and $\P_r (\,\cdot\,)=\P(\,\cdot\,|r_0=r)$.

According to the classical theory \cite[Ch.\ 2, Thm.\ 5.8]{karatzas1998methods} the rational price of an American put option with maturity time $T$, strike price $K>0$, written on the asset $X$ and evaluated at time $t\in[0,T]$ is given by
\[
p_t=\esssup_{t\le \tau\le T}\E\left[e^{-\int_t^\tau r_sds}\Big(K-X_\tau\Big)^+\Big|\cF_t\right],\quad 0\le t\le T,
\]
where the essential supremum is over $(\cF_t)$-stopping times in $[t,T]$ and the function $(\,\cdot\,)^+$ denotes the positive part.
In our Markovian set-up, $p_t=v(t,r_t,X_t)$ for a Borel-measurable function $v:[0,T]\times[\bar r,\underline r]\times \R_+$ (see \cite[Ch.\ 3]{shiryaev2007optimal}). Using that the process $(r_t,X_t)_{t\ge 0}$ is time-homogeneous and strong Markov, we can express $v$ as
\begin{align}\label{eq:v}
v(t,r,x)=\sup_{0\le \tau\le T-t}\E_{r,x}\left[e^{-\int_0^\tau r_s ds}\Big(K-X_\tau\Big)^+\right],
\end{align}
where $r\in\cI$ and $x\in\R_+$ are, respectively, the values of the spot rate and of the asset at time $t$. The above is an optimal stopping problem with Markovian structure and a 3-dimensional state space. 

Since the process 
\begin{align}\label{eq:pproc}
t\mapsto e^{-\int_0^t r_sds}\Big(K-X_t\Big)^+
\end{align} 
is non-negative and continuous, and thanks to the integrability condition \eqref{eq:integr}, we can rely on standard optimal stopping theory (see, e.g., Appendix D in \cite{karatzas1998methods}) to conclude that the smallest optimal stopping time for \eqref{eq:v} is $\P_{r,x}$-a.s.~given by
\begin{align}\label{eq:OST}
\tau_*:=\inf\{s\ge 0\,:\,v(t+s,r_s,X_s)=(K-X_s)^+\},
\end{align}
where we note that $\tau_* \le T-t$ since $v(T, r, x) = (K-x)^+$. Clearly $\tau_*=\tau_*(t,r,x)$ depends on the initial value $(t,r,x)$ of the 3-dimensional state process $(t+s,r_s,X_s)_{s\ge 0}$. 

The form \eqref{eq:OST} of $\tau^*$ gives rise to the so-called continuation set $\cC$ and its complement, the stopping set $\cD$, that is 
\begin{align}
\label{C}&\cC:=\{(t,r,x)\in[0,T]\times\cI\times\R_+\,:\,v(t,r,x)>(K-x)^+\},\\[+4pt]
\label{D}&\cD:=\{(t,r,x)\in[0,T]\times\cI\times\R_+\,:\,v(t,r,x)=(K-x)^+\}.
\end{align}
Upon observing the spot rate and the asset value, at each time the option holder must decide whether to hold the option or to exercise it. She should \emph{wait} (possibly trading the option on the market) if $(t,r_t,X_t)\in\cC$ since the option value is strictly larger than the payoff of immediate exercise. On the contrary, if $(t,r_t,X_t)\in\cD$ the option should be immediately \emph{exercised}. Notice that
\[
\{T\}\times\cI\times\R_+\subseteq\cD.
\]

\begin{remark}
Setting
\[
D_s:=\exp\big(-\int_0^sr_u du\big),\quad V_s:=v(t+s,r_s,X_s)\quad\text{and}\quad Y_s:=D_sV_s
\] 
(i.e., $Y$ is the discounted option value process), we have that \citep[Appendix D]{karatzas1998methods}
\begin{align}
\label{eq:supmg}&\text{$(Y_s)_{s\in[0,T-t]}$ is a right-continuous $\P_{r,x}$-supermartingale,}\\[+4pt]
\label{eq:mg}&\text{$(Y_{s\wedge\tau_*})_{s\in[0,T-t]}$ is a right-continuous $\P_{r,x}$-martingale}.
\end{align}
We will soon show (Proposition \ref{prop:vlip}) that $v$ is a continuous function, so that $Y$ is a continuous process.
\end{remark}

\noindent\textbf{Notation.}
We set 
\begin{align}\label{eq:O}
\cO:=[0,T)\times\cI\times\R_+,
\end{align} 
and denote by $\partial \cC$ the boundary of $\cC$ in $\cO$, i.e., $\partial \cC := (\overline{\cC} \cap \cO) \setminus \cC$.

For future frequent use we denote by $\cL$ the infinitesimal generator of $(r_t,X_t)_{t\ge0}$, which, for any $f \in C^2(\cI \times \R)$ reads
\begin{align}\label{eq:LrX}
\cL f:=\frac{\sigma^2x^2}{2}f_{xx}+\frac{\beta^2(r)}{2}f_{rr}+\rho\sigma x\beta(r)f_{r x}+rx f_x+\alpha(r)f_{r},
\end{align}
where $f_r$, $f_x$ and $f_{rr}$, $f_{rx}$, $f_{xx}$ denote, respectively, the first and second order partial derivatives of $f$.

\section{Main results}\label{sec:main}

In this section we provide the main results of the paper. In Sections \ref{sec:mr1} and \ref{sec:mr2}, under the sole Assumption \ref{ass:coef}, we establish continuous differentiability of the value function $v(t,r,x)$ (jointly in all variables), along with its monotonicity in $(t,r,x)$ and convexity in $x$. We also prove the existence and monotonicity of an optimal exercise boundary and present two possible parametrisations of it. Then, in Sections \ref{sec:mr3}---\ref{sec:mr6}, under a mild additional assumption on $\alpha$ and $\beta$ (Assumption \ref{ass:coef_H}) we derive continuity of the optimal exercise boundary (as a function of two variables) and an integral equation that uniquely determines it (also under Assumption \ref{ass:rbar}). Finally, we obtain the early exercise premium formula for the option price and the hedging portfolio that replicates the option's payoff at all times. 

\subsection{Optimal stopping boundary}\label{sec:mr1}
In the classical Black-Scholes model with constant interest rate, the stopping set is determined by a boundary: it is optimal to exercise the option the first time when the stock price drops below this boundary. A similar characterisation of the stopping region $\cD$ can be derived in our model with the difference that the stopping boundary is a surface. To this end, we research monotonicity properties of the value function.

\begin{proposition}\label{prop:vmonot}
The value function $v$ is finite for all $(t,r,x)\in\cO$ and it satisfies the following conditions:
\begin{itemize}
\item[(i)] $t\mapsto v(t,r,x)$ is non increasing for all $(r,x)\in\cI\times\R_+$,
\item[(ii)] $r\mapsto v(t,r,x)$ is non increasing for all $(t,x)\in[0,T]\times\R_+$,
\item[(iii)] $x\mapsto v(t,r,x)$ is convex and non increasing for all $(t,r)\in[0,T]\times\cI$.
\end{itemize} 
\end{proposition}
\begin{proof}
See Section \ref{sec:lipschitz}.
\end{proof}

The monotonicity in $t$ and $x$ and the convexity in $x$ is the same as in the classical Black-Scholes model and the proof is very similar. The dependence on $r$ has financial explanation: larger interest rate implies stronger discounting of future cashflows and, hence, lower present value.

\begin{remark}\label{rem:inf-hor}
In the case $T=+\infty$ (perpetual option) the discounted payoff process \eqref{eq:pproc} is still uniformly integrable and continuous. This implies that, letting $v_\infty$ denote the value of the perpetual option, the stopping time 
\[
\tau_\infty=\inf\{t\ge0\,:\,v_\infty(r_t,X_t)=(K-X_t)^+\}
\] 
is optimal by standard theory and \eqref{eq:supmg}--\eqref{eq:mg} continue to hold in this setting (see, e.g., \cite[Ch.~3, Thm.~3]{shiryaev2007optimal}). 
In particular, it can be shown that $r\mapsto v_\infty(r,x)$ is non-increasing and $x\mapsto v_\infty(r,x)$ is convex and non-increasing.
\end{remark}

From the general optimal stopping theory we expect that the value function $v$ be continuous. Indeed, this fact is proved from first principles in our Proposition \ref{prop:vlip} in Section \ref{sec:lipschitz} (without relying on the form of the stopping set). The continuity of $v$ means that the continuation set $\cC$ is open and the stopping set $\cD$ is closed. In view of the monotonicity properties established in Proposition \ref{prop:vmonot}, we can show that there is a surface splitting $\cC$ and $\cD$. 

In models with constant interest rate, an optimal boundary is often defined as function of time which provides a threshold for the process $(X_t)$. A parametrisation of the stopping surface as a function $b(t,r)$ of time and interest rate is also available in our setting. For the sake of mathematical tractability we prefer to work with the parametrisation $c(t,x)$ in terms of time and stock price. Due to technical reasons that will become clearer in Section \ref{sec:properties_boundary}, we are able to prove the continuity of $(t,x)\mapsto c(t,x)$ jointly in both variables $(t,x)$, but not the joint continuity of $b$ in $(t,r)$. However, $b$ is more convenient for numerical computations in Section \ref{sec:numerics} as it admits values in a bounded interval $[0, K]$. The connection between $b$ and $c$ is established in Proposition \ref{prop:D}.

\begin{proposition}
\label{prop:boundary-c}
There exists a function $c(t,x)$ on $[0,T]\times \R_+$ such that
\begin{align}
\cD&=\{(t,r,x)\in\cO\,:\,r\ge c(t,x)\}\cup \big(\{T\}\times \cI \times \R_+\big),\label{eq:D_c}\\
\cC&=\{(t,r,x)\in\cO\,:\,r < c(t,x)\}.\label{eq:C_c}
\end{align}
The function $c(t,x)$ has following properties:
\begin{itemize}
\item[(i)]
For any $(t_0, x_0) \in [0, T) \times \R_+$,  the mapping $t \mapsto c(t, x_0)$ is right-continuous and non-increasing and the mapping $x \mapsto c(t_0, x)$ is left-continuous and non-decreasing.
\item[(ii)]
$c(t,x)=\overline{r}$ for $(t,x)\in [0,T)\times [K,\infty)$.
\item[(iii)]
$c(t,x)\ge 0$ for $(t,x)\in [0,T)\times \R_+$; when the risk-free rate $(r_t)_{t\ge 0}$ is non-negative (i.e., $\underbar{r} = 0$), we have $\lim_{x\downarrow 0}c(t,x)=0$ for $t\in [0,T)$.
\end{itemize}
\end{proposition}
\begin{proof}
See Section \ref{sec:properties_boundary}.
\end{proof}

Notice that (ii) and (iii) above imply that it is never optimal to exercise the option out of the money or if the interest rate is negative. This is in line with classical financial wisdom.

The following proposition whose simple proof is omitted gives details of the reparametrisation of the stopping boundary as a function $b(t,r)$ of time and interest rate.
\begin{proposition}\label{prop:D}
Define
\[
b(t,r):=\inf\{x\in\R_+: c(t,x)>r\},\qquad (t, r) \in [0, T) \times \cI.
\]
The mappings $t \mapsto b(t, r_0)$ and $r \mapsto b(t_0,r)$ are right-continuous and non-decreasing for any $(t_0, r_0) \in [0, T) \times \cI$. For any $t \in [0, T)$ we have $K>b(t,r)>0$ when $r>0$, and $b(t,r)=0$ when $r<0$. Furthermore,
\begin{align*}
\cD&=\{(t,r,x)\in\cO\,:\,x\le b(t,r)\}\cup\big(\{T\}\times \cI \times \R_+\big),\\
\cC&=\{(t,r,x)\in\cO\,:\,x> b(t,r)\}.
\end{align*}
\end{proposition}

\subsection{Smoothness of the value function}\label{sec:mr2}
It is well-known that $v$ satisfies (in the classical sense)
\begin{equation}
\label{eqn:pde}
\begin{aligned}
&v_t(t,r,x)+(\cL-r) v(t,r,x)=0,\qquad  (t,r,x) \in \cC,\\
& v(t,r,x) = (K-x)^+,\qquad (t,r,x) \in \cD,
\end{aligned} 
\end{equation}
where $\cL$ is the generator of $(r,X)$ defined in \eqref{eq:LrX}. Hence, standard arguments assert that $v$ is $C^{1,2}$ in $\cC \cap \interior(\cD)$. Classical optimal stopping theory identifies the boundary of the set $\cC$ by imposing the so-called smooth-fit condition. In the American put problem with constant interest rate this corresponds to proving that $x\mapsto v^\circ_x(t,x)$ is continuous for each $t\in[0,T)$ fixed, with $v^\circ$ denoting the value function associated to the option price. In our setting we prove a stronger result and show continuous differentiability of $v$ across the stopping boundary $\partial \cC$, i.e., the global continuity of the gradient of $v$ (as a function of all variables) in $\cO$. We use ideas similar to those in \cite{de2018global} but we must refine arguments therein and use estimates with `local' nature since we are not able to directly check their assumptions. In particular, global differentiability of the flow $r\mapsto(r^r_s,X^{r,x}_s)$ and related integrability conditions (see Eqs.\ (4.4)--(4.7) and Theorem 10 in \cite{de2018global}) are not easily verifiable when, for example, the interest rate follows the CIR dynamics.

\begin{theorem}
\label{prop:v-diff}
We have $v\in C^1(\cO)$.
\end{theorem}
\begin{proof}
See Section \ref{sec:cont}.
\end{proof}

It is worth noticing that the proof of the above result combines a number of steps that may be of independent interest. In particular, we prove local Lipschitz continuity of $v$ (Proposition \ref{prop:vlip}) and the regularity of the stopping boundary in the sense of diffusions. The latter gives the continuity of optimal stopping times $\tau_*$ as functions of the initial state, which plays a crucial role in the proof of the theorem.

\subsection{Continuity of the stopping boundary and Dynkin's formula}\label{sec:mr3}
Preliminary right/left-continuity properties of the stopping boundary $(t,x) \mapsto c(t,x)$ illustrated above follow from its monotonicity and the closedness of the stopping set $\cD$ (see Proposition \ref{prop:boundary-c}). However, thanks to the $C^1$ regularity of the value function $v$, we can also prove joint continuity of the stopping boundary in both variables.
For this we require local H\"older continuity of the derivatives of the coefficients in the dynamics of the short rate $r$. 
\begin{assumption}\label{ass:coef_H}
The functions $\alpha$ and $\beta$ in \eqref{eq:r} have first and second order derivatives, respectively, H\"{o}lder continuous on any compact subset of $\cI$.
\end{assumption}
Note that this assumption is satisfied by CIR model. It strengthens Assumption \ref{ass:coef}(ii) by requiring that the derivatives are not only locally continuous but also locally H\"older continuous. This technical requirement is satisfied by many popular short rate models. The joint continuity of optimal stopping boundaries depending on multiple variables has not been proved with probabilistic techniques before, so the next result is of independent mathematical interest.

\begin{proposition}\label{prop:conti-c}
Under Assumption \ref{ass:coef_H}, the function $c:[0,T)\times \R_+ \to [0,\infty)$ is continuous.
\end{proposition}
\begin{proof}
See Section \ref{sec:integral}.
\end{proof}

Summarising, we have $v\in C^1(\cO)\cap C^{1,2}(\cC)\cap C^{1,2}(\cD)$, and the optimal stopping boundary $c$ is continuous. This is not sufficient to apply the change of variable formula developed in \cite{peskir2007change} which is often used in optimal stopping literature to establish It\^{o}'s formula for the value function. Indeed, since \cite{peskir2007change} deals with functions that are not necessarily $C^1$, it requires that $t \mapsto c(t,X_t)$ be a semi-martingale, so that the local time on the stopping boundary is well-defined. While we were unable to prove it for our optimal boundary, we can instead take advantage of the continuous differentiability of our value function and use a generalisation of It\^{o}'s formula from \cite{cai2021b} which only requires the monotonicity of the boundary. Notice that, interestingly, we need not control the second order spatial derivatives near $\partial\cC$ in order to apply results from \cite{cai2021b}. We do however need to ensure that both boundary points of the set $\cI$ are non-attainable, because we have not proven that the derivatives $v_t(t,\underline r,x)$, $v_r(t,\underline r,x)$ and $v_x(t,\underline r,x)$, understood as the limit as $r \to \underline r$, are well-defined.
\begin{assumption}\label{ass:rbar}
The lower boundary point $\underline r$ is non-attainable by the process $(r_t)$. In particular, under Assumptions \ref{ass:coef}-(i) we require $k\theta>\sigma^2/2$.
\end{assumption}

\begin{proposition}
\label{prop:int-eq}
Under Assumption \ref{ass:rbar}, for any $(t,r,x)\in \cO$ and any stopping time $\tau\in[0, T-t]$, the value function satisfies the following Dynkin's formula:
\begin{equation}
\label{eq:int-1}
\begin{aligned}
v(t,r,x)
&=\E_{r,x}\left[\int_{0}^{\tau}e^{-\int_{0}^{u}r_v dv}Kr_u \mathds{1}_{\{r_u>c(t+u,X_u)\}}du+e^{-\int_{0}^{\tau}r_v dv}v(t+\tau,r_\tau,X_\tau)\right].
\end{aligned}
\end{equation}
\end{proposition}
\begin{proof}
See Section \ref{sec:integral}.
\end{proof}

In the proof of the above proposition, we show that the discounted value function satisfies for any stopping time $\tau \in [0, T-t]$
\begin{equation}
\label{eq:ito-v}
\begin{aligned}
&e^{-\int_{0}^{\tau}r_v dv}v(t+\tau,r_\tau,X_\tau)\\
&=v(t,r,x)- \int_0^\tau e^{-\int_{0}^{s}r_v dv}Kr_s \mathds{1}_{\{r_s>c(t+s,X_s)\}}ds
+ \int_0^\tau e^{-\int_{0}^{s}r_v dv}\sigma X_sv_x(t+s,r_s,X_s)dB_s\\
&\quad+ \int_0^\tau e^{-\int_{0}^{s}r_v dv}\beta(r_s)v_r(t+s,r_s,X_s)dW_s.
\end{aligned}
\end{equation}
This representation will play a fundamental role in deriving a hedging strategy for the American put option in Section \ref{sec:mr6}.

\begin{remark}\label{rem:stop_indicator}
By arguments in Appendix \ref{app:reg}, in particular Remark \ref{rem:lambda}, the distribution of $(r_u, X_u)$, for any $u > 0$, is absolutely continuous with respect to Lebesgue measure on any compact set. Hence, for any $(t,r,x)\in \cO$,
\begin{equation}\label{eqn:stop_indicator}
\E_{r,x}\left[\int_{0}^{T-t} \big|\mathds{1}_{\{ (t+u, r_u, X_u) \in \cD\}} - \mathds{1}_{\{r_u>c(t+u,X_u)\}} \big| du \right] = 0.
\end{equation}
When we apply this result to \eqref{eq:int-1}, we obtain
\[
v(t,r,x) =\E_{r,x}\left[\int_{0}^{\tau}e^{-\int_{0}^{u}r_v dv}Kr_u \mathds{1}_{\{ (t+u, r_u, X_u) \in \cD\}}du+e^{-\int_{0}^{\tau}r_v dv}v(t+\tau,r_\tau,X_\tau)\right].
\]
\end{remark}

\subsection{Early exercise premium}\label{sec:mr4}

Inserting $\tau = T-t$ in \eqref{eq:int-1}, we obtain a decomposition of the American option price into a sum of the European option price $v_e$ and  an \emph{early exercise premium} $v_p$ (see \cite{rutkowski1994early} for a derivation of this formula only using general martingale theory):
\begin{equation}\label{eqn:decomposition}
v(t,r,x) = v_p(t,r,x;T,b)+v_e(t,r,x;T),
\end{equation}
where
\begin{equation}\label{eqn:v_e}
\begin{aligned}
v_e(t,r,x;T) &= \E_{r, x}\left[e^{-\int_{0}^{T-t}r_v dv}(K - X_{T-t})^+\right],\\
v_p(t,r,x;T,b) 
&= \E_{r, x}\left[\int_{0}^{T-t}\!\!e^{-\int_{0}^{u}r_v dv}Kr_u \mathds{1}_{\{r_u>c(t+u,X_u)\}}du\right]\\
&= \E_{r, x}\left[\int_{0}^{T-t}\!\!e^{-\int_{0}^{u}r_v dv}Kr_u \mathds{1}_{\{X_u<b(t+u,r_u)\}}du\right].\\
\end{aligned}
\end{equation}
The last equality follows from $r > c(t,x) \Leftrightarrow x < b(t,r)$ by construction of $b$ as the generalised inverse of $c$.  By Remark \ref{rem:stop_indicator} the early exercise premium also reads
\[
v_p(t,r,x;T,b) 
= \E_{r, x}\left[\int_{0}^{T-t}\!\!e^{-\int_{0}^{u}r_v dv}Kr_u \mathds{1}_{\{(t+u,r_u,X_u)\in\cD\}}du\right].
\]

\subsection{Integral equation for the stopping boundary}\label{sec:mr5}

Proposition \ref{prop:int-eq} provides a characterisation of the optimal stopping boundary $c(t,x)$. Indeed, for any $(t,x) \in [0, T) \times \R_+$ such that $c(t,x) \in \cI$, inserting $\tau=T-t$ and $r=c(t,x)$ in \eqref{eq:int-1} yields an integral equation for $c$:
\begin{equation}
\label{eq:int-2}
\begin{aligned}
(K-x)^+
=\E_{c(t,x),x}\left[\int_{0}^{T-t}e^{-\int_{0}^{u}r_v dv}Kr_u \mathds{1}_{\{r_u>c(t+u,X_u)\}}du+e^{-\int_{0}^{T-t}r_v dv}(K - X_{T-t})^+\right].
\end{aligned}
\end{equation}
The condition that $c(t,x) \in \cI$ is necessary as $c$ can take values $\underline{r}$ and $\overline{r}$ which do not belong do the state space $\cI$, and the interest rate process $r$ may not be started from there. Notice also that $c(t,x) \notin \cI$ when $x \ge K$ so the left-hand side of \eqref{eq:int-2} can be replaced by $(K-x)$. In line with well-known results for American options with constant interest rate \citep{peskir2005american}, it also turns out that $c$ is the unique solution of the integral equation.

\begin{proposition}\label{prop:int-eq-1}
Under Assumptions \ref{ass:coef_H} and \ref{ass:rbar}, the function $c$ is the unique function $\phi:[0, T) \times \R_+ \to [0, \overline{r}]$ such that:
\begin{enumerate}
 \item is continuous, non-decreasing in $x$ and non-increasing in $t$, with $\phi(t,x)=\bar r$ for $x\ge K$,
 \item $\phi$ satisfies \eqref{eq:int-2} (with $c$ therein replaced by $\phi$) for all $(t,x) \in [0, T) \times \R_+$ for which $\phi(t,x) \in \cI$, and $\{(t,x) \in [0, T) \times \R_+:\ \phi(t,x)\in\cI\}\ne\varnothing$.
\end{enumerate}
\end{proposition}

The integral equation \eqref{eq:int-2} has an analogue for the function $b(t,r)$ from Proposition \ref{prop:D}. Indeed, for $b(t,r) > 0$, taking $x=b(t,r)$ and $\tau=T-t$ in Proposition \ref{prop:int-eq} and using $v(t,r,b(t,r))=K-b(t,r)$ we see that $b$ solves the integral equation:
\begin{equation}
\label{eq:numer-1}
\begin{aligned}
K-b(t,r)
&=
\E_{r, b(t,r)}\left[\int_{0}^{T-t}\!\!e^{-\int_{0}^{u}r_v dv}Kr_u \mathds{1}_{\{X_u<b(t+u,r_u)\}}du\right]\\
&\hspace{13pt}+\E_{r, b(t,r)}\left[e^{-\int_{0}^{T-t}r_v dv}(K - X_{T-t})^+\right],
\end{aligned}
\end{equation} 
where we use $\{X_u<b(t+u,r_u)\}=\{r_u > c(t+u,X_u)\}$ which follows from $x>b(t,r)\Leftrightarrow r<c(t,x)$ by construction of $b$ as the generalised inverse of $c$. This parametrisation of the integral equation extends the one obtained in the classical American put problem with constant interest rate to our two-factor set-up. Once again we can prove uniqueness of the solution to the integral equation but without requiring continuity of $b$, which is a non-standard result for this type of equations.

\begin{corollary}\label{cor:int-eq}
Under the assumptions of Proposition \ref{prop:int-eq-1},
the function $b$ is the unique function $\psi:[0, T) \times \cI \to [0,K)$ such that:
\begin{enumerate}
\item $t\mapsto \psi(t,r)$ and $r\mapsto \psi(t,r)$ are right-continuous and non-decreasing,
\item the generalised left-continuous inverse $\phi(t,x) := \inf\{r \in \cI: \psi(t,r) \ge x \}$ is continuous in $(t,x)$, non-decreasing in $x$ and non-increasing in $t$, 
\item $\psi$ satisfies \eqref{eq:numer-1} with ($b$ therein replaced by $\psi$) for all $(t,r) \in [0, T) \times \cI$ such that $\psi(t,r) > 0$ and $\{(t,x) \in [0, T) \times \cI:\ \psi(t,r) > 0 \} \ne \varnothing$.
\end{enumerate}
\end{corollary}
\noindent Notice that $\phi(t,x)=\bar r$ for $x\ge K$ follows immediately from $\psi(t,r) < K$.

Integral equations \eqref{eq:int-2} and \eqref{eq:numer-1} offer a method to compute the optimal stopping boundary without using the value function $v$. We will demonstrate it in Section \ref{sec:numerics} where we design a numerical method for solving such integral equations. Knowing the stopping boundary $b$, the decomposition \eqref{eqn:decomposition} can be employed to obtain an efficient numerical estimate of the option value. This offers an alternative to numerical solution of the variational inequality for the value function $v$, and, subsequently, extraction of the optimal exercise boundary.

\subsection{Hedging portfolio}\label{sec:mr6}
Thanks to the change of variable formula \eqref{eq:ito-v} we are also able to rigorously construct a hedging portfolio that (super)replicates the option payoff at all times. This is based on the classical delta-hedging ideas in the Black and Scholes model but its rigorous mathematical derivation requires smoothness of the option price function which was not previously established in the literature.

Consider a market comprising three instruments: the money market account $M_t:=e^{\int_{0}^{t}r_udu}$, the risky stock with the dynamics \eqref{eq:X}, and a zero-coupon bond with maturity $T$. We will construct a hedging portfolio for the American option on this market. We remark that the zero-coupon bond can be replaced by any other financial instrument whose dynamics depends on the Brownian motion $W$ driving the interest rate, see \cite{karatzas1988pricing}.

The risk-neutral price of the zero-coupon bond at time $t\in[0,T]$ is given by
\begin{equation}\label{eqn:bond}
P(t,r):=\E_{r}\left[e^{-\int_{0}^{T-t}r_udu}\right],\quad P(T,r)=1.
\end{equation} 
By standard arguments based on pathwise continuity of the flow $(t,r)\mapsto r^r_t(\omega)$, one can easily show that $P$ is continuous on $[0,T]\times\cI$. Then, under Assumption \ref{ass:coef}, the classical PDE theory \citep[Thm.\ 9, Ch.\ 4, Sec.\ 3]{Friedman} guarantees that $P$ is the unique classical solution of the boundary value problem
\begin{align*}
&(\partial_t+ \cL_{r}-r) u(t, r)=0,\qquad (t,r) \in [0, T]\times (a,b),\\
&u(t,r)=P(t,r),\qquad\qquad\quad\:\:\, t\in[0,T),\, r\in\{a,b\}\\
&u(T,r)=1, \qquad\qquad\qquad\quad\:\:\, r\in [a,b],
\end{align*}
where $\cL_{r} = \alpha(r) \partial_r + \beta(r)^2/2 \partial_{rr}$ and $(a,b)\subset \cI$ is an arbitrary bounded interval. In particular, by arbitrariness of $(a,b)$ we have $P\in C^{1,2}([0,T)\times\cI)$ and 
\[
(\partial_t+ \cL_{r}-r) P(t, r)=0,\qquad (t,r) \in [0, T]\times \cI.
\]  
Then, using It\^{o}'s formula, the discounted bond price dynamics reads
\begin{equation}
\label{eq:zero-coup}
de^{-\int_{0}^{s}r_udu}P(s,r_s)=P_r(s,r_s)\beta(r_s)dW_s.
\end{equation} 

Denote by $\phi^{(1)},\phi^{(2)}, \phi^{(3)}$ the holdings in the stock, the bond and the money account, respectively. Let $C$ be a non-decreasing continuous process starting from $0$ modelling consumption. The value of a self-financing portfolio starting at time $0$ from $v(0, r, x)$ is
\begin{equation}\label{eqn:h_2}
\Pi_{s} = v(0,r,x) + \int_0^s \phi^{(1)}_u d X_u + \int_0^s \phi^{(2)}_u dP(u,r_u) + \int_0^s \phi^{(3)}_u dM_u - C_s,\quad s \in [0, T].
\end{equation}
The portfolio is \emph{admissible} if all integrals above are semimartingales. Taking the money-market account as a numer\'{a}ire, we obtain from equations \eqref{eqn:h_2} and
\begin{equation}\label{eqn:h_1}
\Pi_{s}:=\phi^{(1)}_s X_s+ \phi^{(2)}_s P(s,R_s)+\phi^{(3)}_s M_s,\quad s \in [0, T],
\end{equation}
that the dynamics of the discounted portfolio value reads
\begin{align}
de^{-\int_{0}^{s}r_udu}\Pi_{s}&=\phi^{(1)}_s de^{-\int_{0}^{s}r_udu}X_s+ \phi^{(2)}_s de^{-\int_{0}^{s}r_udu}P(s,r_s) - e^{-\int_{0}^{s}r_udu} dC_s \nonumber\\
&=e^{-\int_{0}^{s}r_udu}\phi^{(1)}_s\sigma X_s dB_s+ e^{-\int_{0}^{s}r_udu}\phi^{(2)}_s\beta(r_s)P_r(s,r_s) dW_s - e^{-\int_{0}^{s}r_udu} dC_s.\label{eqn:disc_port}
\end{align}
This means that a self-financing portfolio is uniquely determined by the processes $\phi^{(1)}, \phi^{(2)}$ and $C$.

Comparing \eqref{eqn:disc_port} with \eqref{eq:ito-v}, a candidate for the hedging strategy is given by
\begin{equation}\label{eqn:hedge}
\phi^{(1)}_s=v_x(s,r_s,X_s),\quad \phi^{(2)}_s=\frac{v_r(s,r_s,X_s)}{P_r(s,r_s)}, 
\quad 
C_s = \int_0^s K r_u \mathds{1}_{\{r_u>c(u, X_u)\}} du.
\end{equation}
We can indeed prove that such portfolio strategy is admissible and replicates the option's payoff.
\begin{proposition}\label{prop:hedge}
Under Assumption \ref{ass:rbar} the portfolio $(\phi^{(1)}, \phi^{(2)}, C)$ is admissible and replicates the payoff of the American put option.
\end{proposition}
\begin{proof}
See Section \ref{sec:hedge}.
\end{proof}
From \eqref{eqn:disc_port} and \eqref{eqn:hedge} one can immediately see that if the option holder exercises the option according to the optimal rule \eqref{eq:OST}, no consumption is available to the seller.


\section{Numerical analysis}\label{sec:numerics}
In the numerical analysis, we assume that the interest rate $r$ follows Vasicek model. In particular, this means that $\cI=\R$ and
\begin{equation}
\label{eq:numer-R}
dr_t=\kappa(\theta-r_t)dt+\beta dW_t,
\end{equation} 
whose explicit solution is given by
\begin{equation}
\label{eq:numer-R-ex}
r_s=r_te^{-(s-t)\kappa}+\theta(1-e^{-(s-t)\kappa})+\beta e^{-s\kappa}\int_{t}^{s}e^{\kappa u}dW_u,\quad s\ge t \ge 0.
\end{equation}

We first derive a numerical method for computing the optimal stopping boundary using the integral equation from \eqref{eq:numer-1}. Once the boundary is obtained, we use it to also compute the value function via \eqref{eqn:decomposition}. Section \ref{subsec:num} contains an analysis of the effect of parameters on the stopping boundary and the value function.

\subsection{Computational approach}
With an abuse of notation, we denote by $P(t, T) = P(t, r, T)$ the time-$t$ price of a zero-coupon bond with maturity $T$ (c.f. \eqref{eqn:bond}); the dependence on the initial state $r$ is indicated in the subscript of the expectation operator. By Proposition \ref{prop:D} we have $b(t,r) = 0$ for $r < 0$, i.e., it is never optimal to stop for negative values of the interest rate. To compute $b(t,r)$ for $r \ge 0$, recall the integral equation \eqref{eq:numer-1}: for $(t,r) \in [0, T) \times \cI$ such that $b(t,r) > 0$, we have
\begin{equation}\label{eqn:b_fixed}
K-b(t,r) = v_p(t,r,b(t,r);T, b)+v_e(t,r,b(t,r);T),
\end{equation}
where $v_e$ and $v_p$ are stated in \eqref{eqn:v_e}. With the last parameter $b$ of $v_p$, we emphasise the dependence on the function $b$:
\[
v_p(t,r,x;T, b) = \E_{r, x}\left[\int_{0}^{T-t}\!\!e^{-\int_{0}^{u}r_v dv}Kr_u \mathds{1}_{\{X_u<b(t+u,r_u)\}}du\right].
\]
In the numerical scheme below, we evaluate $v_p$ for consecutive approximations of $b$.

In Appendix \ref{app:numerics}, we derive the following formulas for $v_e$ and $v_p$ using well-known properties of the joint law of $(r_t, X_t)$:
\begin{align}
\label{eq:numer-2}
v_e(t,r,x;T)&=P(t,T)K \mathcal{N}(d_1)-x \mathcal{N}(d_2),\\
\label{eq:numer-3}
v_p(t,r,x;T, b)&=\int_{t}^{T} KP(t,u)\bigg[\int_{-\infty}^{\infty}\frac{1}{\sqrt{2\pi}}e^{-\frac{y^2}{2}}\Big(q(t,u) + y\sqrt{\gamma_2(t,u)}\Big) \mathcal{N}\big(\phi(t,u,y; b)\big)dy\bigg]\, du,
\end{align}
where $\mathcal{N}(\cdot)$ is the cumulative distribution function of the standard normal distribution. An explicit formula for $P(t, T)$ is given by \eqref{eq:numercom-7} and the other auxiliary quantities used above are stated in \eqref{eqn:explicit_pricing}.

Equation \eqref{eqn:b_fixed} defines the boundary $b$ as a fixed point of a non-linear mapping. To compute it, we follow an iterative scheme motivated by \cite{Kitarbayev2018}. We fix $-\infty < r_{min} < r_{max} < \infty$ and discretise the variables $(t, r)$ as follows:
\[
\{(t_i,r_j) \in [t,T]\times[r_{min},r_{max}]\}, \quad i=1,...,M, j=1,...,N.
\]
We specify an initial approximation $b^{(0)}$ of the boundary:
\[
b^{(0)}(t_i,r_j)=K,\quad \forall\, i,j.
\]
For each $n\ge 1$, we compute the boundary $b^{(n)}$ at points $(t_i,r_j)_{i,j}$ by solving the algebraic equation:
\begin{equation}
\label{eq:numer-4}
K-b^{(n)}(t_i,r_j)-v_e\left(t_i,r_j,b^{(n)}(t_i,r_j);T\right)=v_p\left(t_i,r_j,b^{(n-1)}(t_i,r_j);T,b^{(n-1)}\right).
\end{equation}
The right-hand side, which is difficult to compute, is independent of $b^{(n)}$, while the left-hand side is known in an explicit form. We stop iterations when, for a pre-determined $\eps > 0$,
\[
\max_{i,j}|b^{(n-1)}(t_i,r_j)-b^{(n)}(t_i,r_j)|<\eps.
\]

The numerical evaluation of $v_p\left(t_i,r_j,b^{(n-1)}(t_i,r_j);T,b^{(n-1)}\right)$ requires that the boundary $b^{(n-1)}$ be known for all points $(t, r)$ in the state space while we compute it only on the grid $(t_i, r_j)$. We, therefore, use Matlab interpolation function with the Modified Akima cubic Hermite polynomials (`makima') interpolation method. Integrals are computed using Matlab functions employing standard quadrature methods.

It should be remarked that the stopping boundary $b$ may have a singularity (jump) at $r=0$, which corresponds to a horizontal part of the parametrisation $c$ of the stopping surface: a jump occurs when $c^{-1}(\{0\}) \ne \varnothing$. Furthermore, $b(T-, r) := \lim_{t \uparrow T} b(t,r)$ satisfies $b(T-, r) = 0$ for $r < 0$ and $b(T-, r) \ge b(0, r) > 0$ for $r > 0$, see Proposition \ref{prop:D}. This hints at a potential numerical difficulty around $r=0$, particularly for times $t$ close to maturity.

\subsection{Sensitivity analysis}\label{subsec:num}

\begin{figure}[tp]
\centering
\includegraphics[width=11cm]{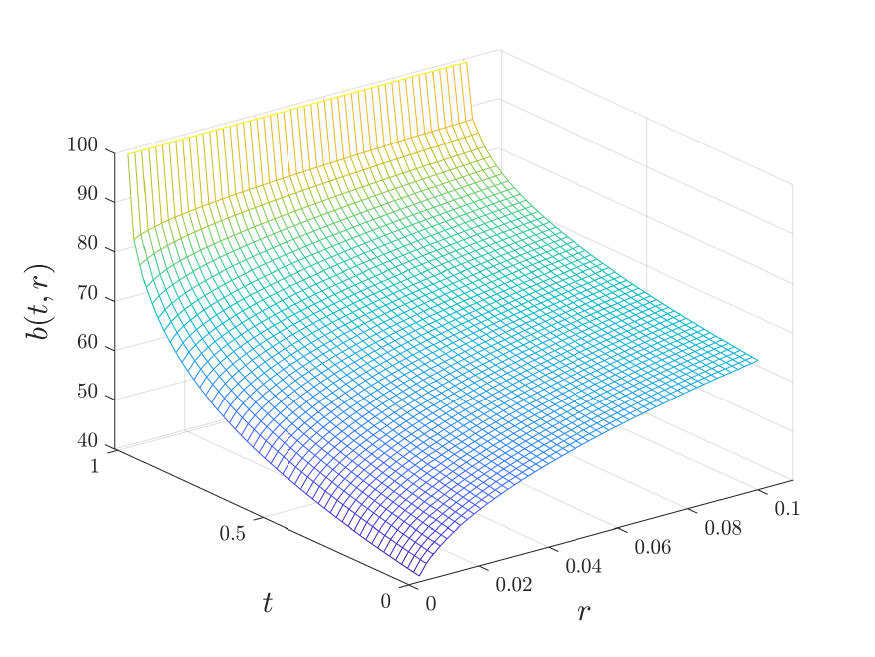}
\caption{Stopping boundary surface $b(t,r)$.}
\label{p0.1}
\end{figure} 
\begin{figure}[tbp]
\centering
\begin{subfigure}[b]{0.47\textwidth}
 \centering
 \includegraphics[width=0.9\textwidth]{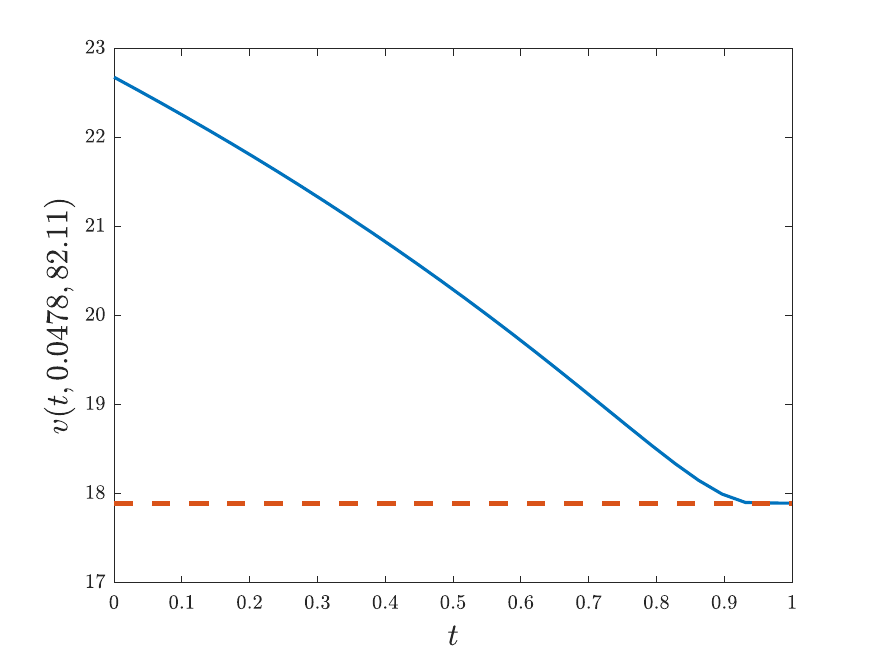}
 \subcaption{$t \mapsto v(t, 0.0478, 82.11)$}\label{p0.2a}
\end{subfigure}
\begin{subfigure}[b]{0.48\textwidth}
 \centering
 \includegraphics[width=0.9\textwidth]{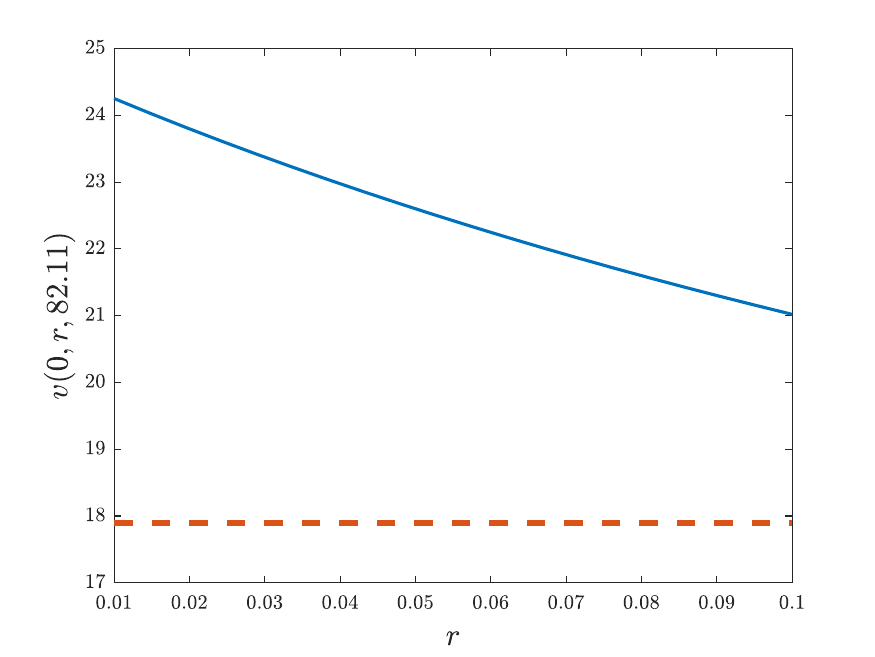}
 \subcaption{$r \mapsto v(0, r, 82.11)$}\label{p0.2b}
\end{subfigure}
\begin{subfigure}[b]{0.48\textwidth}
 \centering
 \includegraphics[width=0.9\textwidth]{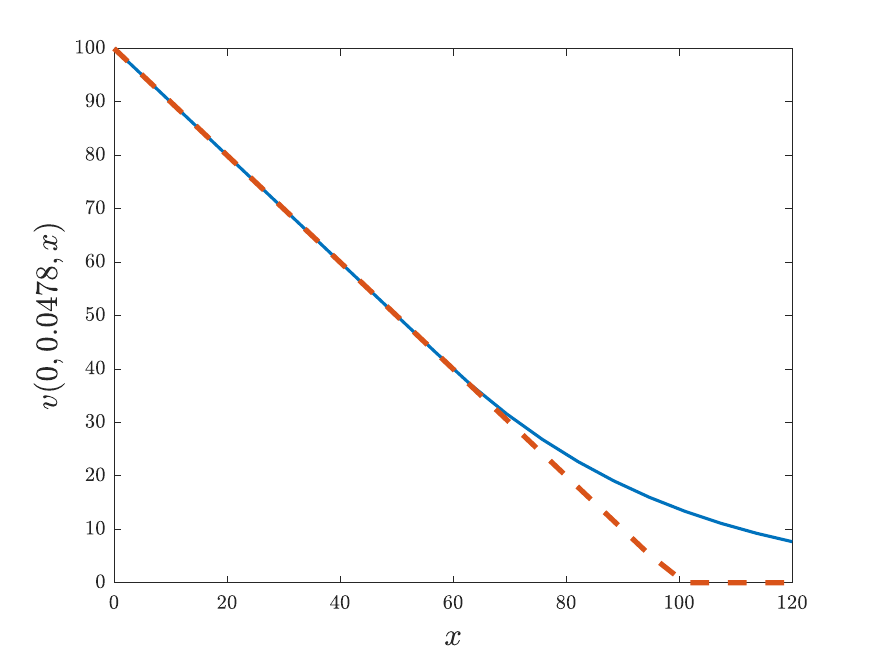}
 \subcaption{$x \mapsto v(0, 0.0478, x)$}\label{p0.2c}
\end{subfigure}
\caption{Sections of the value function $v(t,r,x)$ through the point $(0, 0.0478, 82.11)$. The dashed line displays the payoff $(K-x)^+$.}
\label{p0.2}
\end{figure} 

Unless stated otherwise, numerical results are presented for the parameter values
\begin{equation}
\label{eq:para-1}
T=1, K=100, \sigma=0.4, \kappa=0.3, \theta=0.05, \beta=0.01, \rho=0.5,
\end{equation}
and the convergence criterion with $\eps=0.01$. The magnitude of $\kappa, \theta$ and $\beta$ is based on empirical findings reported in the literature, c.f. \cite[Chapter 31]{hull2009options} and \cite{fergusson2015application}. Although main currencies have recently enjoyed much lower interest rates, our choice of $\theta$ means that the effects of random interest rate and its parameters on the market dynamics and optimal stopping boundary are more pronounced and graphs more transparent.

Figure \ref{p0.1} plots the stopping boundary $b(t,r)$ using parameters \eqref{eq:para-1}. The optimal stopping boundary increases as $t$ tends to the maturity $T$ and as the interest rate $r$ grows (c.f. Proposition \ref{prop:D}). This behaviour is consistent with the one of the optimal exercise boundary for the American put option in the Black-Scholes model with a constant interest rate \citep{peskir2005american}. Figure \ref{p0.2} illustrates the value function $v(t,r,x)$ via sections in directions of $t$, $r$ and $x$ rooted at the point $(0, 0.0478, 82.11)$, which illustrates the findings of Proposition \ref{prop:vmonot}. In Panel (\subref{p0.2a}), the value decreases to the value of the immediate exercise as the option is purchased deep in the money. \\

\emph{Effects of the interest rate}. The option price is significantly affected by the initial interest rate (Panel (\subref{p0.2b})) because the maturity of the option is long (1 year). The effect depends on the mean-reversion coefficient $\kappa$ and it increases when the mean reversion parameter decreases. Indeed, this tendency is clearly visible in Figure \ref{p1.7}. A large mean-reversion speed ($\kappa = 1$) means that the interest rate is quickly pulled towards $\theta = 0.05$, diminishing the effect of the initial value. Taking expectation on both sides of \eqref{eq:numer-R-ex} gives that the expected interest rate at the maturity $T=1$ is
\[
\E_r[r_1] = r e^{-\kappa} + \theta (1 - e^{-\kappa}),
\]
which, for $\kappa = 1$, means $\E_r[r_1] \approx 0.36\, r + 0.74\, \theta$. On the contrary, we obtain $\E_r[r_1] \approx 0.90\, r + 0.10\, \theta$
for $\kappa = 0.1$ and so the effect of the initial interest rate on the stopping boundary (Figure \ref{p1.7a}) and the value function (Figure \ref{p1.7b}) is more pronounced. The optimal strategy for $\kappa = 0.1$ prescribes to be more patient compared to larger values of $\kappa$ when the interest rate is near $0$ and act faster when the interest rate is close to $1$. Indeed, with a slow mean-reversion the interest rate stays close to the current value for longer, so the observed behaviour of the stopping boundary and the of value function is akin to that observed by a model with a constant interest rate \citep{broadie1996, peskir2005american}.

\begin{figure}[t]
\centering
\begin{subfigure}[b]{0.47\textwidth}
 \centering
 \includegraphics[width=0.9\textwidth]{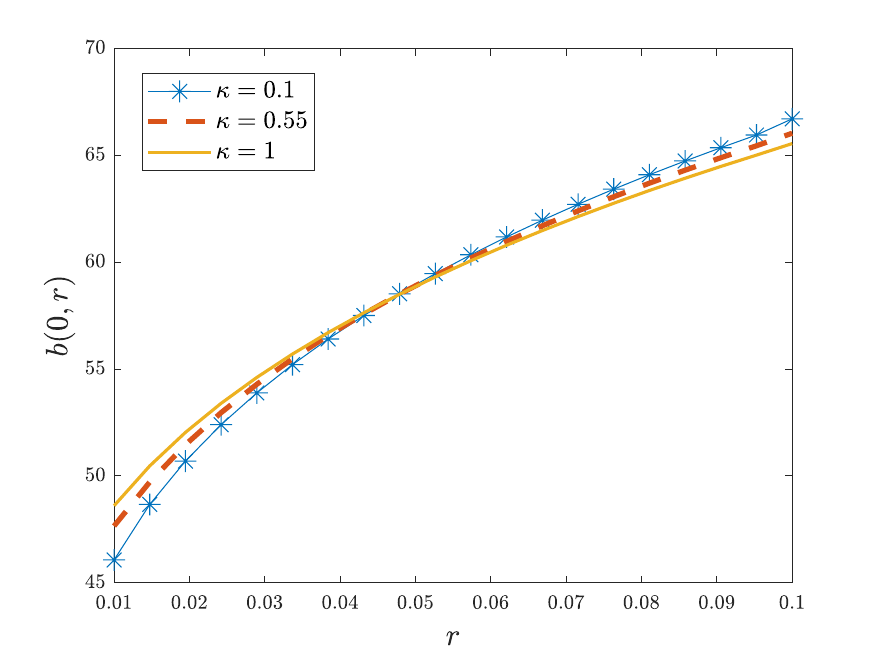}
 \subcaption{}\label{p1.7a}
\end{subfigure}
\begin{subfigure}[b]{0.47\textwidth}
 \centering
 \includegraphics[width=0.9\textwidth]{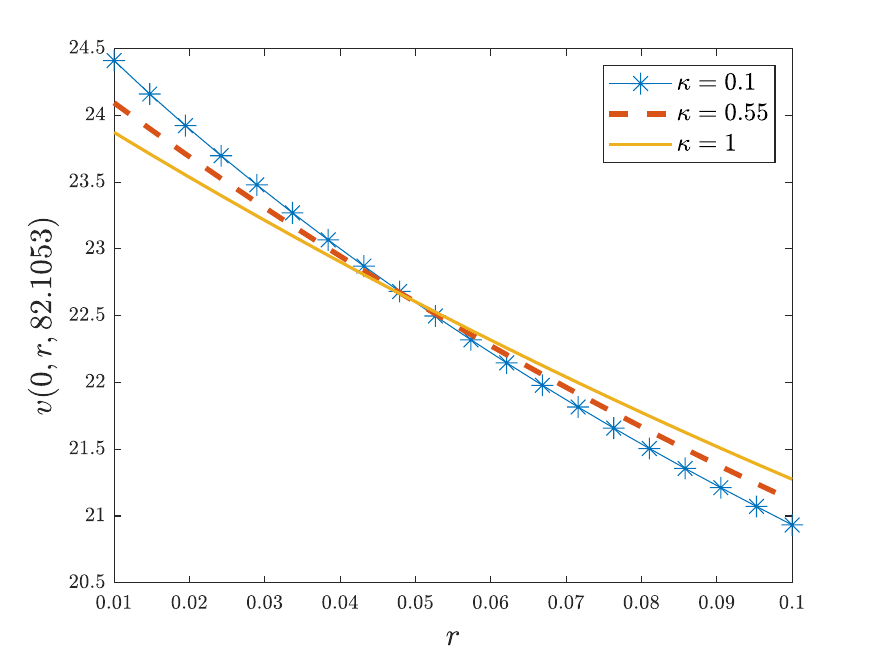}
 \subcaption{}\label{p1.7b}
\end{subfigure}
\caption{The $r$-sections of the stopping boundary (left panel) and the value function (right panel) for the mean-reversion parameter $\kappa \in \{0.1, 0.55, 1\}$.} 
\label{p1.7}
\end{figure} 

\begin{figure}[t]
\centering
\begin{subfigure}[b]{0.47\textwidth}
 \centering
 \includegraphics[width=0.9\textwidth]{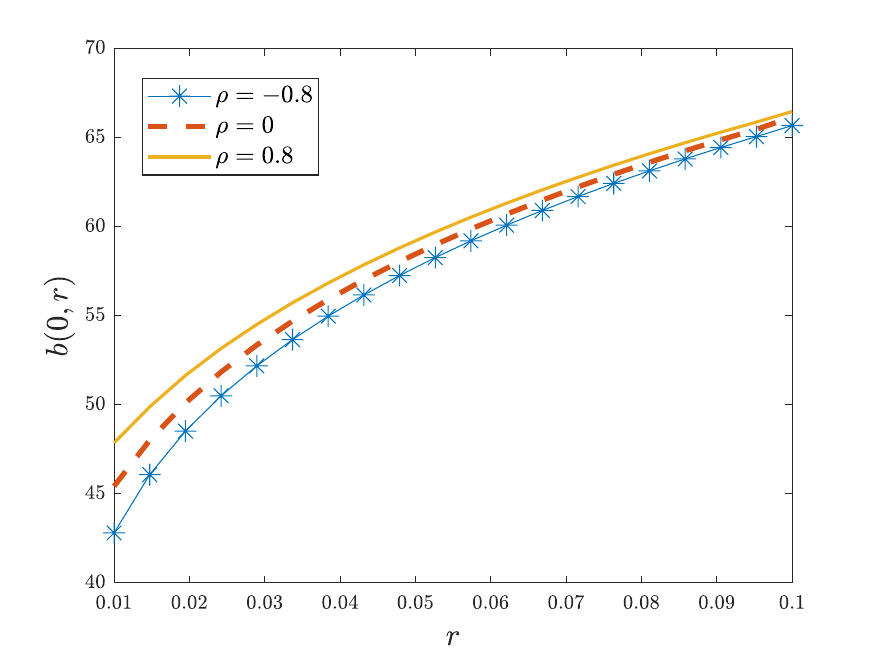}
 \subcaption{}\label{p1.5a}
\end{subfigure}
\begin{subfigure}[b]{0.47\textwidth}
 \centering
 \includegraphics[width=0.9\textwidth]{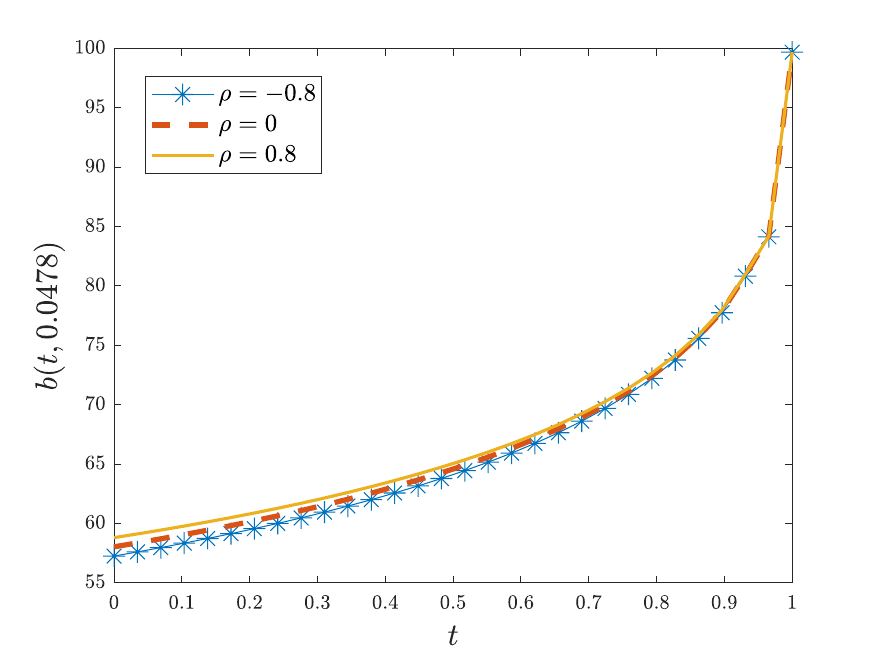}
 \subcaption{}\label{p1.5b}
\end{subfigure}
\caption{The $r$ and $t$-sections of the stopping boundary for the correlation coefficient $\rho \in \{-0.8, 0, 0.8\}$.} 
\label{p1.5}
\end{figure} 

{\em Effects of the correlation coefficient}. The sensitivity of the stopping boundary with respect to the correlation coefficient $\rho$ between Brownian motions driving the stock price and the interest rate is displayed in Figure \ref{p1.5}; the value function behaves accordingly and it is not displayed. High positive correlation $\rho=0.8$ implies that the interest rate and the stock price tend to move together. The increase in the interest rate pushes the stock price up and vice versa, resulting in a more unstable environment and an earlier optimal stopping. On the contrary, a strong negative correlation sees the stock price and the interest rate dampening the effect of each other's moves: an increase in the stock price brings a drop in the interest rate, therefore, making longer waiting (lower stopping boundary) more desirable due to effect on the drift of the stock price as well as on the discount factor. Naturally, this effect diminishes the closer one gets to the maturity of the option, see Figure \ref{p1.5b}.\\

{\em Effects of the volatility of stock and interest rate}. The effect of the diffusion coefficient of the spot rate $\beta$ on the stopping boundary and on the value function is negligible. We compared results for $\beta \in \{0.005, 0.01, 0.015\}$, the range of values observed in empirical literature mentioned above. We noticed variations in the value function of less than $0.1\%$ and in the stopping boundary of less than $1\%$. 

\begin{figure}[t]
\centering
\begin{subfigure}[b]{0.47\textwidth}
 \centering
 \includegraphics[width=0.9\textwidth]{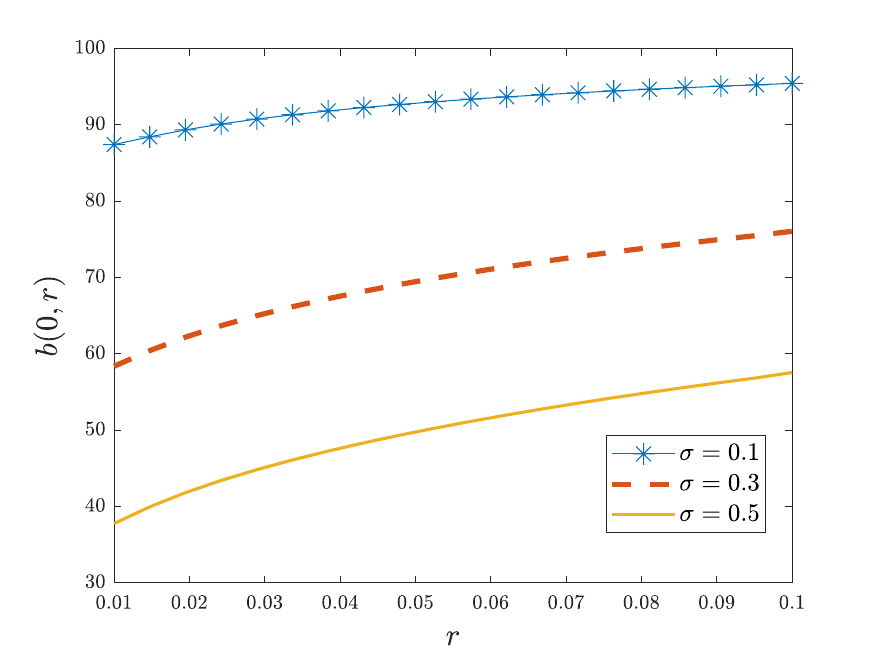}
 \subcaption{}\label{p1.3a}
\end{subfigure}
\begin{subfigure}[b]{0.47\textwidth}
 \centering
 \includegraphics[width=0.9\textwidth]{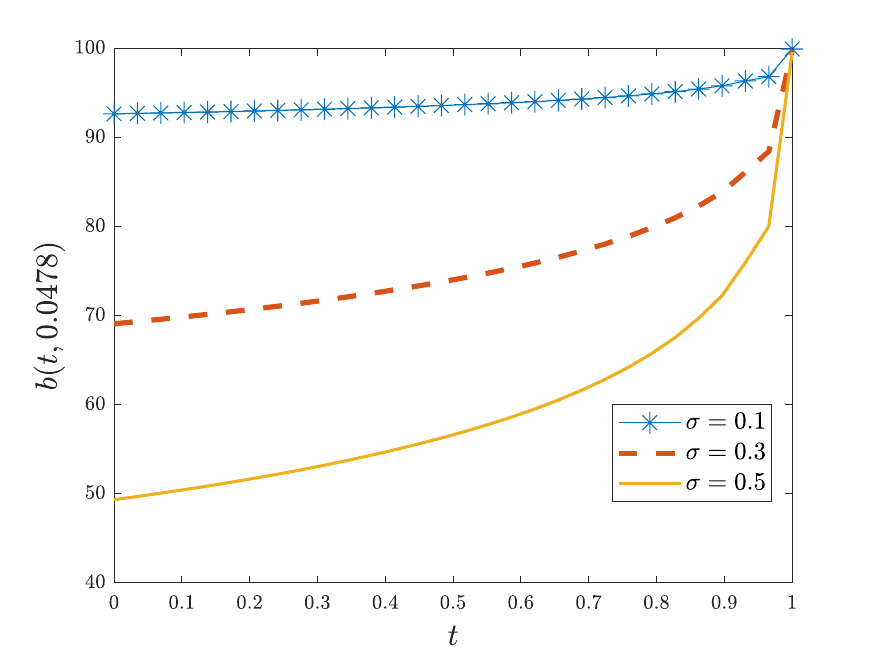}
 \subcaption{}\label{p1.3b}
\end{subfigure}
\begin{subfigure}[b]{0.47\textwidth}
 \centering
 \includegraphics[width=0.9\textwidth]{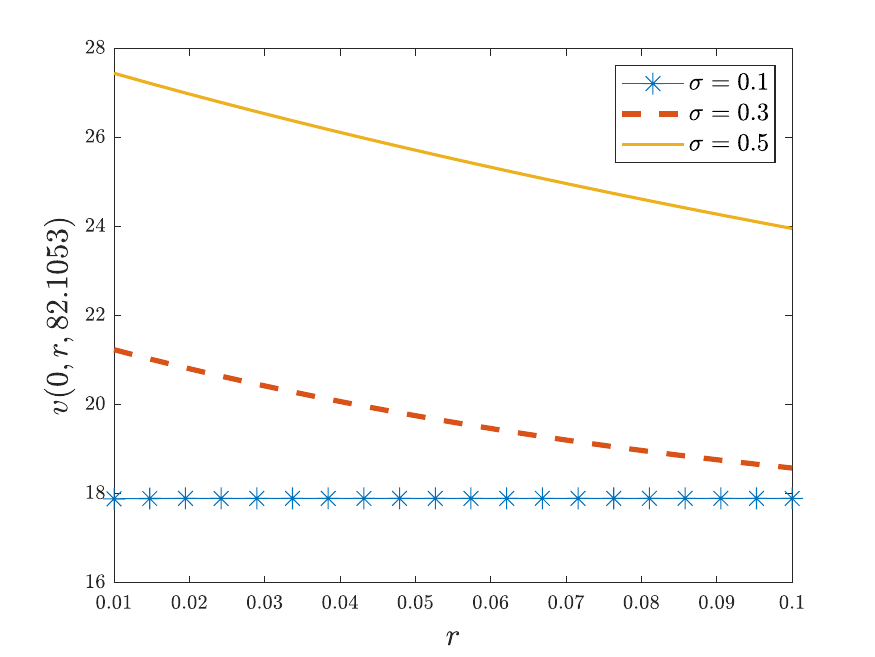}
 \subcaption{}\label{p1.3c}
\end{subfigure}
\begin{subfigure}[b]{0.47\textwidth}
 \centering
 \includegraphics[width=0.9\textwidth]{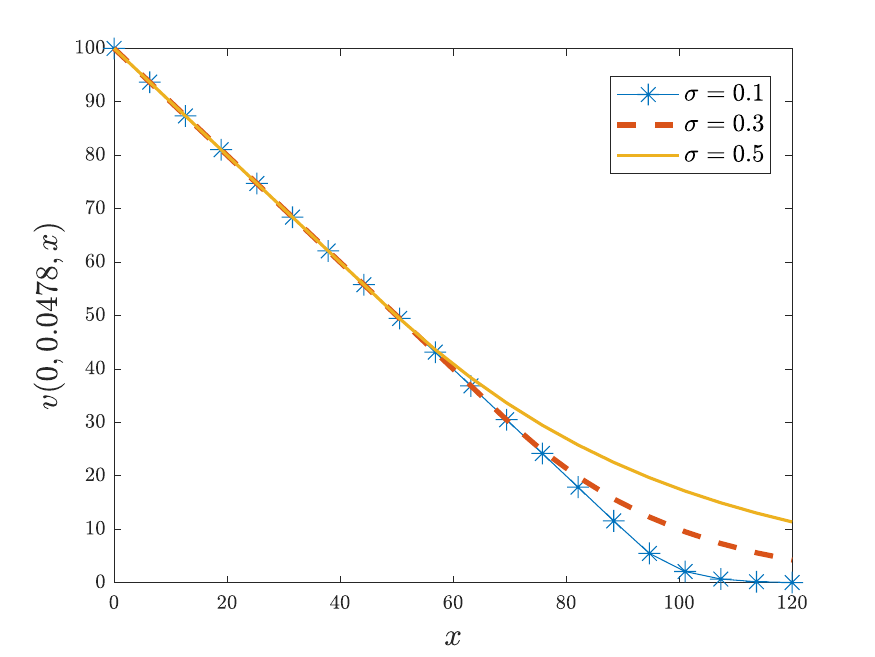}
 \subcaption{}\label{p1.3d}
\end{subfigure}
\caption{Effect of the volatility of the stock price $\sigma$. Panels (a) and (b) display the $r$ and $t$-sections of the stopping boundary $b(t,r)$ and Panels (c) and (d) show the $r$ and $x$-sections of the value function $v$ for $\sigma \in \{0.1, 0.3, 0.5\}$.} 
\label{p1.3}
\end{figure} 

In line with the financial intuition, the value of American Put option is increasing in $\sigma$, see Figure \ref{p1.3c} and \ref{p1.3d}. When $\sigma = 0.1$, the optimal stopping boundary is close to the exercise price $K$ (Figure \ref{p1.3a}), so the option is immediately exercised for the initial stock price $x = 82.1053$ presented on Panel (c), hence the flat graph. For other values of $\sigma$, the exercise boundary is below the initial stock price and the effect of the interest rate is clearly visible. The structure of results in Figure \ref{p1.3} is, as expected, in line with the findings for the American Put option in the Black-Scholes model with constant interest rate \citep{broadie1996, peskir2005american}. 

The remaining sections of the paper contain technical details and proofs.

\section{Monotonicity and Lipschitz continuity of the option value}\label{sec:lipschitz}

In this section we establish some initial regularity properties of the option value. We start with key monotonicity results and then prove Lipschitz continuity of the value function.

\begin{proof}[Proof of Proposition \ref{prop:vmonot}]
Finiteness of $v$ follows by \eqref{eq:integr} and boundedness of the put payoff. Monotonicity in $(i)$ is also a trivial consequence of the fact that the discounted put payoff is independent of time. For $(ii)$ we argue as follows: since $r\mapsto r^r_{t}$ is increasing $\P$-a.s.~for all $t\in[0,T]$ (by uniqueness of the trajectories) we get, for any $\eps>0$
\begin{align*}
v(t,r+\eps,x)=&\,\sup_{0\le \tau\le T-t}\E\left[\Big(K e^{-\int_0^\tau r^{r+\eps}_tdt} -xe^{\sigma B_\tau-\tfrac{\sigma^2}{2}\tau}\Big)^+\right]\\
\le &\,\sup_{0\le \tau\le T-t}\E\left[\Big(K e^{-\int_0^\tau r^r_tdt} -xe^{\sigma B_\tau-\tfrac{\sigma^2}{2}\tau}\Big)^+\right]=v(t,r,x)
\end{align*}
where we took the discounting inside the positive part and used \eqref{eq:X2}.

Finally, monotonicity in $(iii)$ is a simple consequence of monotonicity of \eqref{eq:X2} with respect to $x$ and the fact that $x\mapsto (K-x)^+$ is decreasing. Convexity also follows by standard arguments: fix $\lambda\in(0,1)$, take $x$ and $y$ in $\R_+$ and denote $x_\lambda:=\lambda x+(1-\lambda)y$. By the convexity of the put payoff, using that $X^{r,x_\lambda}=\lambda X^{r,x}+(1-\lambda) X^{r,y}$ and that $\sup(f+g)\le \sup f+\sup g$, it is not hard to verify that $v(t,r,x_\lambda)\le \lambda v(t,r,x)+(1-\lambda)v(t,r,y)$.
\end{proof}

\begin{proposition}{{\bf(Lipschitz continuity).}}\label{prop:vlip}
For any compact $\cK\subset \cO$ there exists a constant $L_\cK>0$ such that
\begin{align}\label{eq:vlip}
|v(t_1,r_1,x_1)-v(t_2,r_2,x_2)|\le L_\cK\big(|t_1-t_2|+|r_1-r_2|+|x_1-x_2|\big)
\end{align}
for all $(t_1,r_1,x_1)$ and $(t_2,r_2,x_2)$ in $\cK$.
\end{proposition}
\begin{proof}[Proof of Proposition \ref{prop:vlip}]
We look separately at Lipschitz continuity in the three variables. Arguments for $r$ and $x$ are quite standard while the main argument for the Lipschitz continuity in $t$ goes back to \cite[Thm.~3.6]{jaillet1990variational}. However, in our framework the interest rate is random and the coefficients of the underlying process are state dependent, which results in some additional difficulties.
\medskip

\emph{Continuity in $x$.} Fix $(t,r)\in[0,T)\times\cI$ and take $x_1\le x_2$ in $\R_+$. Let $\tau_1:=\tau_*(t,r,x_1)$ and note that it is admissible for $v(t,r,x_2)$. Using Proposition \ref{prop:vmonot}(iii), the explicit expression for $X^{r,x}$ in \eqref{eq:X2} and the Lipschitz property of the put payoff, we get
\begin{align*}
0\le v(t,r,x_1)-v(t,r,x_2)
&\le \E\left[e^{-\int_0^{\tau_1}r^r_sds}\Big((K-X^{r,x_1})^+-(K-X^{r,x_2})^+\Big)\right]\\[+4pt]
&\le \E\Big[e^{\sigma B_{\tau_1}-\tfrac{\sigma^2}{2}\tau_1}\Big](x_2-x_1)=(x_2-x_1),
\end{align*}
where in the last equality we used Doob's optional sampling theorem.
\medskip

\emph{Continuity in $r$.} Fix $(t,x)\in[0,T)\times\R_+$ and take $r_1\le r_2$ in $\cI$ such that $(t, r_1, x) \in \cK$. Denote, for simplicity, $r^1:=r^{r_1}$ and $r^2:=r^{r_2}$ and notice that $r^2_t\ge r^1_t$ for all $t\ge0$ $\P$-a.s. Set $\tau_1:=\tau_*(t,r_1,x)$. From Proposition \ref{prop:vmonot}(ii) and simple estimates we obtain
\begin{equation}\label{eqn:v_r1}
\begin{aligned}
0
\!\le\! v(t,r_1,x)\!-\!v(t,r_2,x)
&\!\le\! K\E\left[e^{-\int_0^{\tau_1}r^1_sds}\!-\!e^{-\int_0^{\tau_1}r^2_sds}\right]\!=\! K\E\left[e^{-\int_0^{\tau_1}r^1_sds}\Big(1\!-\!e^{-\int_0^{\tau_1}(r^2_s-r^1_s)ds}\Big)\right]\\
&\le K\E\Big[e^{-\int_0^{\tau_1}r^1_sds}\int_0^{\tau_1}(r^2_s-r^1_s)ds\Big].
\end{aligned}
\end{equation}
To complete the proof we consider separately cases (i) and (ii) in Assumption \ref{ass:coef}.  Let us start with $(i)$: using that $r^1_t\ge0$ for $t\ge 0$, and the explicit form of the SDE in the CIR model, we get
\begin{align*}
\E\Big[e^{-\int_0^{\tau_1}r^1_sds}\int_0^{\tau_1}(r^2_s-r^1_s)ds\Big]
&\le
\int_0^{T-t}\E\left[r^2_s-r^1_s\right]ds\\
&=
\int_0^{T-t}\E\Big[(r_2-r_1)+\int_0^s\kappa(r^1_u-r^2_u)du \Big]ds\le (T-t)(r_2-r_1),
\end{align*}
where we have used the integral equation for $(r_t)$ and that $r^2_t\ge r^1_t$.

If Assumption \ref{ass:coef}(ii) holds instead, we apply H\"older inequality:
\begin{align}
\E\Big[e^{-\int_0^{\tau_1}r^1_sds}\int_0^{\tau_1}(r^2_s-r^1_s)ds\Big]
&\le
\Big( \E\Big[e^{-2\int_0^{\tau_1}r^1_sds}\Big]\Big)^{\frac{1}{2}}\Big(\E\Big[\Big(\int_0^{T-t}(r^2_s-r^1_s)ds\Big)^2\Big] \Big)^{\frac{1}{2}} \label{eqn:r_bnd}\\
&\le C_1^{1/2} \Big((T-t) \int_0^{T-t}\E\left[(r^2_s-r^1_s)^2\right]ds\Big)^{\frac{1}{2}},\nonumber
\end{align}
where $C_1>0$ is the constant from \eqref{eq:integr} which depends on $\cK$. To conclude it is sufficient to use moment estimates for SDEs \citep[Ch.\ 2, Sec.\ 5, Thm.\ 9]{krylov} which guarantee that 
\begin{equation}\label{eqn:v_r2}
\E\left[\sup_{0\le s\le T}(r^2_s-r^1_s)^2\right]\le c'(r_2-r_1)^2 
\end{equation}
for some $c'>0$ only depending on $T$ and the coefficients in \eqref{eq:r}.
\vspace{+5pt}

\emph{Continuity in $t$.} For $t \in [0, T)$, define $r^{T-t}_u:=r_{u(T-t)}$ and $X^{T-t}_u:=X_{u(T-t)}$ for $u\in[0,1]$. The couple $(r^{T-t}_u,X^{T-t}_u)_{u \in [0,1]}$ is a strong solution to (see, e.g., \cite[Ch.~1, Prop.~8.6]{bass1998diffusions})
\begin{align*}
&dX^{T-t}_u=(T-t)r^{T-t}_uX^{T-t}_u du + \sigma X^{T-t}_u d\widetilde B_u,\qquad X^{T-t}_0=x,\\
&dr^{T-t}_u=(T-t)\alpha(r^{T-t}_u)du+\beta(r^{T-t}_u)d\widetilde W_u,\qquad r^{T-t}_0=r,
\end{align*}
where $(\widetilde B_{u},\widetilde W_{u})_{u\in[0,1]}:=(B_{u(T-t)},W_{u(T-t)})_{u\in[0,1]}$. Using these processes, we can rewrite \eqref{eq:v} as
\begin{align}\label{eq:vT}
v(t,r,x)=\sup_{0\le \theta\le 1}\E_{r,x}\left[\exp\Big\{\!-(T-t)\!\int_0^\theta r^{T-t}_u du\Big\}\Big(K-X^{T-t}_\theta\Big)^+\right],
\end{align}
where for any $(\cF_s)_{s\ge0}$-stopping time $\tau$ in $[0,T-t]$ the random variable $\theta:=\tau/(T-t)$ is an $(\cF_{u(T-t)})_{u\in[0,1]}$-stopping time. Since the process $(B_{u(T-t)},W_{u(T-t)})_{u\in[0,1]}$ is identical in law to $(\sqrt{T-t}B_u,\sqrt{T-t}W_u)_{u\in[0,1]}$, with a slight abuse of notation we can identify $(r^{T-t}_u,X^{T-t}_u)_{u\in[0,1]}$ with the unique strong solution of 
\begin{align}
\label{eq:XT}&dX^{T-t}_u=(T-t)r^{T-t}_uX^{T-t}_u du + \sqrt{T-t}\sigma X^{T-t}_u d B_u,\qquad X^{T-t}_0=x,\\
\label{eq:rT}&dr^{T-t}_u=(T-t)\alpha(r^{T-t}_u)du+\sqrt{T-t}\beta(r^{T-t}_u)d W_u,\qquad r^{T-t}_0=r,
\end{align}
and take stopping times $\theta\in[0,1]$ in \eqref{eq:vT} with respect to the filtration $(\cF_t)$ generated by $(B,W)$.
In what follows we denote by $\theta_*=\theta_*(t,r,x)$ an optimal stopping time for \eqref{eq:vT}.

Fix now $0\le t_1<t_2<T$ and set $r^1:=r^{T-t_1}$, $r^2:=r^{T-t_2}$. Let $\theta_1:=\theta_*(t_1,r,x)$ and for $i=1,2$ denote also
\[
R^i_u = (T-t_i) \int_0^u r^i_s ds
\quad \text{and}\quad 
\hat X^{T-t}_u = \exp\Big(\sqrt{T-t}\,\sigma B_u-(T-t)\frac{\sigma^2}{2}u\Big),
\]
so that $X^{T-t_i}_u = x e^{-R^i_u} \hat X^{T-t_i}_u$.  We remark that $\theta_1$ is also admissible for the problem in \eqref{eq:vT} and the underlying dynamics \eqref{eq:XT}--\eqref{eq:rT} with $t=t_2$, because it is an $(\cF_s)_{s \ge 0}$-stopping time in $[0,1]$. Indeed the advantage of \eqref{eq:vT} with \eqref{eq:XT}--\eqref{eq:rT} is that the class of admissible stopping times no longer depends on the initial time $t$.

Recalling Proposition \ref{prop:vmonot}(i) and using Lipschitz continuity of $x \mapsto (x)^+$ we have
\begin{equation}\label{eq:dt0}
\begin{aligned}
0\ge v(t_2,r,x)-v(t_1,r,x)
&\ge - \E_{r} \Big[ \Big| \Big(K e^{-R^2_{\theta_1}} - x \hat X^{T-t_2}_{\theta_1}\Big)^+ - \Big(K e^{-R^1_{\theta_1}} - x\hat X^{T-t_1}_{\theta_1}\Big)^+\Big| \Big]\\
&\ge -K\E_r\left[\Big|e^{-R^2_{\theta_1}} - e^{-R^1_{\theta_1}}\Big|\right]
-x \E\big[\big|\hat X^{T-t_1}_{\theta_1} - \hat X^{T-t_2}_{\theta_1} \big|\big].
\end{aligned}
\end{equation}

Let us consider the second term on the right hand side of \eqref{eq:dt0}. By the fundamental theorem of calculus and the explicit formula for $\hat X^{T-t}$
\begin{align}
\E\big[\big|\hat X^{T-t_1}_{\theta_1} - \hat X^{T-t_2}_{\theta_1} \big|\big]
&=
\E \Big[ \Big|\int_{t_1}^{t_2} \hat X^{T-t}_{\theta_1} \Big(\frac{\sigma^2}{2} \theta_1 - \frac{1}{2\sqrt{T-t}} \sigma B_{\theta_1} \Big) dt \Big|\Big]\nonumber\\
&\le
\int_{t_1}^{t_2} \E \Big[ \Big|\hat X^{T-t}_{\theta_1} \Big(\frac{\sigma^2}{2} \theta_1 - \frac{1}{2\sqrt{T-t}} \sigma B_{\theta_1} \Big) \Big|\Big] dt. \label{eqn:X_est}
\end{align}
For $t \in (t_1, t_2)$, define a measure $\tilde \P$ by $\frac{d\tilde{\P}}{d \P}:=\hat X^{T-t}_1$. Then $\tilde{B}_s=B_s-\sigma s\sqrt{T-t}$ is a Brownian motion under $\tilde{\P}$ and
\[
\begin{aligned}
&\, \E \Big[ \Big|\hat X^{T-t}_{\theta_1} \Big(\frac{\sigma^2}{2} \theta_1 - \frac{1}{2\sqrt{T-t}} \sigma B_{\theta_1} \Big) \Big|\Big]
=
\tilde{\E}\left[\Big|\frac{\theta_1}{2}\sigma^2-\frac{\sigma}{2\sqrt{T-t}}(\tilde{B}_{\theta_1}+\sqrt{T-t}\sigma\theta_1)\Big|\right]\\
&=
\tilde{\E}\left[\Big|\frac{\sigma}{2\sqrt{T-t}}\tilde{B}_{\theta_1}\Big|\right]
\le \left(\tilde{\E}\bigg[ \frac{\sigma^2\tilde{B}^2_{\theta_1}}{4(T-t)}\bigg]\right)^{1/2}
\le
\frac{\sigma}{2\sqrt{T-t}}
\le \frac{\sigma}{2\sqrt{T-t_2}} =: c_1,
\end{aligned}
\]
where we applied H\"{o}lder inequality and used that $\theta_1 \le 1$. Inserting the above estimate into \eqref{eqn:X_est} gives
\begin{equation}\label{eq:dt1}
\E\big[\big|\hat X^{T-t_1}_{\theta_1} - \hat X^{T-t_2}_{\theta_1} \big|\big] \le c_1 (t_2 - t_1).
\end{equation}

Next we address the first term on the right hand side of \eqref{eq:dt0}. This is performed separately in cases (i) and (ii) of Assumption \ref{ass:coef}. We start by considering case (ii), i.e., $\alpha$ and $\beta$ in \eqref{eq:rT} are Lipschitz continuous. Fundamental theorem of calculus and H\"older inequality give
\begin{equation}\label{eq:dt2}
\begin{aligned}
&\E_r\left[\Big|e^{-R^1_{\theta_1}} - e^{-R^2_{\theta_1}}\Big|\right]\\
&\le \E_r\left[\max_{i=1,2}\left\{e^{-(T-t_i)\int_0^{\theta_1}r^i_udu}\right\}\Big|(T-t_1)\!\int_0^{\theta_1}r^1_udu-(T-t_2)\!\int_0^{\theta_1}r^2_udu\Big|\right]\\
&\le \E_r\bigg[\max_{i=1,2}\left\{e^{-(T-t_i)\int_0^{\theta_1}r^i_udu}\right\}\Big((t_2-t_1)\Big|\int_0^{\theta_1}r^1_udu\Big|+(T-t_2)\Big|\int_0^{\theta_1}(r^2_u-r^1_u)du\Big|\bigg]\\
&\le
2 c_2 \bigg[
(t_2-t_1)\Big(\E_r\!\big[\sup_{0\le t\le 1}\big(r^1_t\big)^2\big]\Big)^{\frac{1}{2}}
+
(T-t_2) \Big(\E_r\left[\int_0^{1}(r^2_u-r^1_u)^2du\right]\Big)^{\frac{1}{2}} \bigg],
\end{aligned}
\end{equation}
where, using \eqref{eq:integr},
\[
c_2 := \sup_{(t,r,x) \in \cK} \bigg(\E_r\!\Big[\sup_{0\le s\le 1}e^{-2(T-t)\int_0^{s}r^{T-t}_udu}\Big]\bigg)^{\frac{1}{2}} < \infty.
\]
Thanks to \eqref{eq:subg}, $c_3 := \sup_{(t,r,x) \in \cK}\Big(\E_r\!\big[\sup_{0\le s\le 1}\big(r^{T-t}_s\big)^2\big]\Big)^{\frac{1}{2}} < \infty$, so it remains to estimate the last term of \eqref{eq:dt2}. By \cite[Ch.~2, Sec.~5, Thm.~9]{krylov} there is a constant $c_4$ depending only on $\cK$ and the Lipschitz constant for $\alpha$ and $\beta$ in \eqref{eq:rT} such that
\begin{align*}
&\E_r\Big[\sup_{0\le t\le 1}\big(r^1_t-r^2_t\big)^2\Big]\\
&\le 
c_4\, \E_r \Big[\int_0^1 \Big(|(T-t_1) \alpha(r^1_u) - (T-t_2) \alpha(r^1_u) |^2 + |\sqrt{T-t_1}\beta(r^1_u) - \sqrt{T-t_2} \beta(r^1_u)|^2\Big) \ du \Big]\\
&\le
c_4 (t_2-t_1)^2\, \E_r \Big[ \int_0^1 |\alpha(r^1_u)|^2 du \Big] + c_4\, \frac{(t_2-t_1)^2}{4(T-t_2)}\, \E_r \Big[\int_0^1 |\beta(r^1_u)|^2 du \Big],
\end{align*}
where for the second inequality we used that $0\le \sqrt{T-t_1} - \sqrt{T-t_2} \le (t_2-t_1)/[2\sqrt{T-t_2}]$. Notice that by  \eqref{eq:subg} and the linear growth of $\alpha$ and $\beta$
\[
c_5 := \sup_{(r,t,x) \in \cK} \E_r \Big[\int_0^1 |\alpha(r^{T-t}_u)|^2 + |\beta(r^{T-t}_u)|^2 du \Big]< \infty.
\]
Inserting the above estimates into \eqref{eq:dt2} we conclude that there is a constant $c_6$ such that for any $(t_1, r, x), (t_2, r, x) \in \cK$
\[
\E_r\left[\Big|e^{-R^1_{\theta_1}} - e^{-R^2_{\theta_1}}\Big|\right]
\le
c_6 |t_2 - t_1|.
\]
This and \eqref{eq:dt1} feed into \eqref{eq:dt0} so that 
\begin{equation}\label{eqn:v_t_lipsch}
0\ge v(t_2,r,x)-v(t_1,r,x)\ge -c|t_2-t_1|
\end{equation}
for a suitable $c>0$ that depends on $\cK$.

Finally, we must estimate the first term on the right hand side of \eqref{eq:dt0} under the assumption that $(r_t)_{t\ge0}$ follows the CIR dynamics (Assumption \ref{ass:coef}(i)). Let $\hat r^i_u:=r^i_u/(T-t_i)$ for $u\in[0,1]$ and $i=1,2$. The dynamics for $\hat r^i$ reads
\begin{align}\label{eq:R}
d\hat{r}^i_u=\kappa\big(\alpha-(T-t_i)\hat{r}^i_u\big)du+\beta\sqrt{\hat{r}^i_u}dW_u, \quad u\in[0,1].
\end{align}
Since $\kappa(\alpha-(T-t_1) \hat{r}) < \kappa(\alpha-(T-t_2) \hat{r})$ for $\hat r \ge 0$, and $\hat{r}^1_0=r/(T-t_1)\le r/(T-t_2)=\hat{r}^2_0$, comparison results for SDEs \citep[Prop.~5.2.18]{karatzas1998brownian} imply
\begin{equation}\label{eq:R1R2a}
\hat{r}^1_u\le \hat{r}^2_u\quad\text{for all $u\in[0,1]$, $\P$-a.s.}
\end{equation}
Using the integral version of \eqref{eq:R} and the martingale property of the stochastic integral, we obtain
 \begin{align*}
 \E_r\left[\hat{r}^2_u-\hat{r}^1_u\right]
 &= r\Big(\frac{1}{T-t_2}-\frac{1}{T-t_1}\Big)+\E_r\left[\int_0^u\left((T-t_1)\hat{r}^1_s-(T-t_2)\hat{r}^2_s\right)ds\right]\\
 &\le r \frac{t_2-t_1}{(T-t_1)(T-t_2)} +(t_2-t_1)\int_0^1\E_r\left[\hat{r}^1_s\right]ds
 +(T-t_2)\int_0^u\E_r\left[\hat{r}^1_s-\hat{r}^2_s\right]ds.
 \end{align*}
 Due to \eqref{eq:R1R2a}, the last term is non-positive, so
\begin{equation}\label{eq:R1R2b}
0\le \E_r\left[\hat{r}^2_u-\hat{r}^1_u\right]\le (t_2-t_1) \Big(\frac{r}{(T-t_1)(T-t_2)}+q_1\Big)\quad \text{for all $u\in[0,1]$}
\end{equation}
where 
\[
q_1:= \sup_{(t,r,x) \in \cK} \frac{1}{T-t}\int_0^1\E_r\left[r^{T-t}_u\right]du < \infty.
\]

We use the inequalities \eqref{eq:R1R2a}--\eqref{eq:R1R2b} and the property that $\hat{r}^i_u\ge 0$, for $i=1,2$, to obtain the following estimates
\begin{equation}\label{eq:dt5}
\begin{aligned}
\E_r\left[\Big|e^{-R^1_{\theta_1}} - e^{-R^2_{\theta_1}}\Big|\right]
&=\E_r\left[\Big|e^{-(T-t_1)^2\int_0^{\theta_1}\hat{r}^1_udu}-e^{-(T-t_2)^2\int_0^{\theta_1}\hat{r}^2_udu}\Big|\right]\\
&\le \E_r\left[\Big|e^{-(T-t_1)^2\int_0^{\theta_1}\hat{r}^1_udu}-e^{-(T-t_2)^2\int_0^{\theta_1}\hat{r}^1_udu}\Big|\right]\\
&\hspace{11pt}+\E_r\left[\Big|e^{-(T-t_2)^2\int_0^{\theta_1}\hat{r}^1_udu}-e^{-(T-t_2)^2\int_0^{\theta_1}\hat{r}^2_udu}\Big|\right]\\
&\le q_1\left((T-t_1)^2-(T-t_2)^2\right)+(T-t_2)^2\int_0^1\E_r\left[\hat{r}^2_u-\hat{r}^1_u\right]du\\
&\le (t_2-t_1) \Big( 2T q_1 + r \frac{T-t_2}{T-t_1}+ q_1 (T-t_2)^2\Big) \le c_7 (t_2-t_1),
\end{aligned}
\end{equation}
where the constant $c_7>0$ depends only on $\cK$ but not on a specific choice of $t_1, t_2, r, x$. Hence, as in the case of Assumption \ref{ass:coef}(ii), we obtain \eqref{eqn:v_t_lipsch}.
\end{proof}

\section{Properties of the free boundary}\label{sec:properties_boundary}

This section is devoted to establishing the existence of an optimal stopping boundary (free boundary) and some of its main properties. In particular we show the so-called `regularity' of the stopping boundary in the sense of diffusion theory which, together with the monotonicity, is instrumental in our proof of global $C^1$ regularity of the value function $v$.

\begin{proof}[Proof of Proposition \ref{prop:boundary-c}]
The payoff does not depend on $(r_t)$ and $v$ is non-increasing in $r$ by Proposition \ref{prop:vmonot}. Therefore, if $(t,r_1,x)\in \cD$ then $(t,r_2,x)\in\cD$ for for any $r_2>r_1$. This allows us to represent the stopping region $\cD$ via \eqref{eq:D_c} with
\begin{equation}
\label{eq:boundary-c}
c(t,x):=\inf\{r\in \cI: v(t,r,x)=(K-x)^+\},
\end{equation}
with the convention that $\inf \emptyset = \overline r$. It is convenient to prove (ii) first.
\vspace{+4pt}

(ii)
Fix $(t,r,x)\in [0,T)\times \cI \times [K,\infty)$. If we show that $\P_{r,x}(X_{\eps}<K)>0$ for some $\eps \in (0, T-t]$, then $v(t,r,x) > 0 = (K-x)^+$. This means that $(t,r,x) \in \cC$ and $c(t,x) = \overline r$. Recall that $\rho\in(-1,1)$ is the correlation coefficient between the Brownian motions $B$ and $W$ driving the SDEs for $X$ and $r$, respectively. Then we can write $B_t=\rho W_t+\sqrt{1-\rho^2}B^0_t$ for some other Brownian motion $B^0$ independent of $W$. Letting $(\cF^W_t)_{t\ge 0}$ be the filtration generated by $W$, using the explicit form of the dynamics of $X$ we have 
\begin{align}\label{eq:XK}
\P_{r,x}(X_{\eps}<K)&=\E_{r,x}\Big[\P_{r,x}(X_{\eps}<K|\cF^W_\eps)\Big]\notag\\
&=\E_{r,x}\Big[\P_{r}\Big(\exp\big(\sigma \sqrt{1-\rho^2}B^0_\eps\big)<(K/x)\exp\Big(-\sigma\rho W_\eps-\int_0^\eps r_t dt+\tfrac{\sigma^2}{2}\eps\Big)\Big|\cF^W_\eps\Big)\Big]\\
&=\E_{r,x}\Big[\Psi_x\Big(\sigma\rho W_\eps+\int_0^\eps r_t dt-\tfrac{\sigma^2}{2}\eps\Big)\Big],\notag
\end{align}
where
\[
\Psi_x(z):=\P\Big(\exp\big(\sigma \sqrt{1-\rho^2}B^0_\eps\big)<(K/x)e^{-z}\Big)
\]
and the final equality above holds by the independence of $B^0_\eps$ from $\cF^W_\eps$ and the fact that $(W_\eps,\int_0^\eps r_t dt)$ is $\cF^W_\eps$-measurable. Since $\rho\in(-1,1)$, then $\Psi_x(z)>0$ for any $z\in \R$ and we conclude that $\P_{r,x}(X_{\eps}<K)>0$.
\vspace{+4pt}

(i)
By the monotonicity of $v$ in $t$, we have $(t_1,r,x)\in \cD\implies(t_2,r,x)\in\cD$ for any $t_2>t_1$, hence $c(t,x)$ is non-increasing in $t$. 

Fix $0\le x_1<x_2<K$ and let $\tau_1:=\tau_*(t,r,x_1)$ be optimal for $v(t,r,x_1)$. Then, using that $X^{r,x_1}\le X^{r,x_2}$ and recalling \eqref{eq:X2}, we obtain
\begin{align*}
v(t,r,x_2)-v(t,r,x_1)\ge&\, \E\left[e^{-\int_0^{\tau_1}r_sds}\left(\big(K-X^{r,x_2}_{\tau_1}\big)^+-\big(K-X^{r,x_1}_{\tau_1}\big)^+\right)\right]\\
\ge&\, \E\left[e^{-\int_0^{\tau_1}r_sds}\left(X^{r,x_1}_{\tau_1}-X^{r,x_2}_{\tau_1}\right)\right]\\
=&\,x_1-x_2=(K-x_2)^+-(K-x_1)^+.
\end{align*}
Therefore, if $(t,r,x_1)\in\cC$ then $(t,r,x_2)\in\cC$, which implies that $c(t,x)$ is non-decreasing in $x$.

Fix arbitrary $(t,x)\in[0,T)\times \R_+$, let $t_{n}\downarrow t_0$ as $n\to\infty$, then $c(t_n,x)\uparrow c(t_0+,x)$ as $n\to\infty$, where the limit exists by the monotonicity of $t \mapsto c(t,x)$. Since $(t_n,c(t_n,x),x)\in\cD$, then also $(t_0,c(t_0+,x),x)\in\cD$ by the closedness of $\cD$, hence $c(t_0+,x)\ge c(t_0,r)$ which implies $c(t_0+,r)= c(t_0,r)$. Taking $x_n\uparrow x_0$, a similar argument yields $c(t,x_0-)=c(t,x_0)$.
\vspace{+4pt}

(iii)
Under the CIR model, the positivity follows by the definition of $c(t,x)$. Only under Assumption \ref{ass:coef} (ii) a proof is required. Assume that there exists $(t_0,\hat{x})\in [0,T)\times(0,K)$ such that $c(t_0,\hat{x})<0$. Let $0>r_2>r_0>r_1>c(t_0,\hat{x})$ and $0<x_0<\hat{x}$. Define a stopping time
\[
\tau_1=\inf\{s \ge 0: (s, r_s, X_s)\notin [0,T-t_0)\times(r_1,r_2) \times (0,\hat{x})\}. 
\]
By the monotonicity of $c(t,x)$, we have $(t_0,r_0,x_0)\in \cD$. Hence, $\tau_1$ is sub-optimal and
\begin{align}\label{eqn:Kx}
K-x_0=v(t_0,r_0,x_0) 
&\ge \E_{r_0,x_0}\left[e^{-\int_0^{\tau_1}r_s ds}\left(K-X_{\tau_1}\right)^+\right]
\ge K\E_{r_0,x_0}\left[e^{-\int_0^{\tau_1}r_s ds}\right]-x_0,
\end{align}
where the last inequality follows from the optional sampling theorem and the fact that $(K - X_{\tau_1})^+ \ge K- X_{\tau_1}$. Since $\P_{r_0,x_0}(\tau_1>0)=1$ and $r_s(\omega) < r_2 < 0$ for $s \in [0, \tau_1(\omega))$, we obtain
\[
K\E_{r_0,x_0}\left[e^{-\int_0^{\tau_1}r_s ds}\right]-x_0 > K - x_0,
\]
which, in conjunction with \eqref{eqn:Kx}, leads to a contradiction.

Finally, we show that $c(t,0+):=\lim_{x\downarrow 0}c(t,x)=\underline{r}$ for any $t\in [0,T)$ if $\underline{r} \ge 0$. Assume $c(t,0+)\ge \delta>\underline{r}$ for some $t \in [0,T)$. By the monotonicity of $c(t,x)$ and the openness of $\cC$ there is $\hat t \in (t, T)$ such that
\[
[0,\hat{t})\times (r_1, r_2) \times (0,\infty) \subset \cC,
\]
where $\underline{r}<r_1<r_2<\delta$. Fix $0\le t_0<\hat{t}$ and $r_0\in (r_1,r_2)$. Take an arbitrary $x_0 > 0$. Let
\[
\tau_2=\inf\{s \ge 0: (s, r_s)\notin [0,\hat{t}-t_0)\times (r_1,r_2)\}.
\] 
By construction $\P_{r_0,x_0}\big((t_0+s, r_s, X_s)\in\cC\,\text{for}\,s\le \tau_2\big)=1$, so $\tau_2\le \tau_*(t_0,r_0,x_0)$ $\P_{r_0, x_0}$-a.s. By the martingale property of the value function we obtain
\begin{equation}
\label{eq:range_c_1}
\begin{aligned}
K-x_0<v(t_0,r_0,x_0)
&=
\E_{r_0,x_0} \left[e^{-\int_0^{\tau_2}r_sds}v\left(t_0+\tau_2,r_{\tau_2},X_{\tau_2}\right)\right]\\
&\le K\E_{r_0, x_0} \left[e^{-r_1 \tau_2}\right]
= K\E_{r_0} \left[e^{-r_1 \tau_2}\right].
\end{aligned}
\end{equation}
A contradiction is obtained by taking the limit $x_0 \downarrow 0$, since $\E_{r_0} \left[e^{-r_1 \tau_2}\right]$ is independent of $X$ and strictly smaller than $1$ since $r_1 > \underline{r} \ge 0$.
\end{proof}

An important consequence of Proposition \ref{prop:D} is that for $\eps \in (0,x)$
\begin{align*}
(t,r,x)\in\cD\implies(t+\eps,r,x), (t,r+\eps,x), (t,r,x-\eps)\in\cD.  
\end{align*}
We immediately see that $\partial\cC$ enjoys the so-called \emph{cone property} \citep[Def.~4.2.18]{karatzas1998brownian}. Indeed, for any $(t_0,r_0,x_0)\in\partial\cC$, there is an orthant $\widehat C_0$ with vertex in $(t_0,r_0,x_0)$ (hence a cone with aperture $\pi/4$) that satisfies $\widehat C_0\cap\cO\subseteq \cD$. This will be used to establish regularity of the boundary $\partial \cC$ in the sense of diffusions, which, has important consequences for the smoothness of our value function $v$, as we shall see below. 

To this end, we introduce the hitting time to $\cD$, denoted $\sigma_\cD$, and the entry time to the interior of $\cD$, denoted $\ss_\cD$. That is, for $(t,r,x)\in\cO$ we set $\P_{r,x}$-a.s.
\begin{equation}\label{eqn:sigmacD}
\begin{aligned}
&\sigma_\cD:=\inf\{s>0\,:\,(t+s,r_s,X_s)\in\cD\},\\
&\ss_\cD:=\inf\{s\ge 0\,:\,(t+s,r_s,X_s)\in\interior(\cD)\}\wedge(T-t).
\end{aligned}
\end{equation}
Both $\sigma_\cD$ and $\ss_\cD$ are stopping times with respect to the filtration $(\cF_t)_{t \ge 0}$. We will often write $\sigma_\cD(t,r,x)$ and $\ss_\cD(t,r,x)$ to indicate the starting point of the process. 

\begin{proposition}[\textbf{Regularity of the boundary}]\label{prop:reg} 
For $(t_0,r_0,x_0)\in\partial\cC$, we have
\begin{align}\label{eq:reg}
\P_{t_0,r_0,x_0}(\sigma_\cD>0)=\P_{t_0,r_0,x_0}(\ss_\cD>0)=0.
\end{align}
\end{proposition}
The proof can be found in Appendix \ref{app:reg}. It rests on Gaussian bounds for the transition density of a diffusion and ideas from the proof of well-known analogous results for multi-dimensional Brownian motion, see e.g. \cite[Thm.~4.2.19]{karatzas1998brownian}. It is also worth recalling that $\partial\cC$ is the boundary of $\cC$ in $\cO$, so that it excludes $\{T\}\times\cI\times\R_+$.

\section{Continuous differentiability of the option value}\label{sec:cont}

We start by establishing the following continuity properties of processes $r$ and $X$.
\begin{lemma}\label{lem:runif}
Let $(r_n,x_n)_{n\ge1}$ be a sequence converging to $(r,x)\in\cI\times\R_+$ as $n\to\infty$. Then
\begin{align}
\label{eq:ucr}&\lim_{n\to +\infty}\sup_{0\le t\le T}\left|r^{r_n}_t-r^{r}_t\right|=0, \qquad\text{$\P$-a.s.}\\
\label{eq:ucx}&\lim_{n\to +\infty}\sup_{0\le t\le T}\left|X^{r_n,x_n}_t-X^{r,x}_t\right|=0, \qquad\text{$\P$-a.s.}
\end{align}
\end{lemma}
\begin{proof}[Proof of Lemma \ref{lem:runif}]
Assume first that $(r_n)_{n\ge 1}$ is a monotone sequence. Define $f^n_t:=r^{r_n}_t-r^{r}_t$. Then for a.e.~$\omega\in\Omega$, $t \mapsto f^n_t(\omega)$ is continuous and $f^n_t(\omega)$ converges to $0$ monotonically as $n \to \infty$ for all $t\in[0,T]$. Hence the convergence is uniform on $[0,T]$ thanks to Dini's theorem and \eqref{eq:ucr} holds. 

For an arbitrary sequence $(r_n)_{n\ge 1}$ define monotone sequences $\bar{r}_n = \sup_{k \ge n} r_k$ and $\underline{r}_n = \inf_{k \ge n} r_k$. Since $r^{\underline{r}_n}_t-r^{r}_t \le r^{r_n}_t - r^r_t \le r^{\bar{r}_n}_t-r^{r}_t$, we have
\[
0 \le \sup_{0\le t\le T}\left|r^{r_n}_t-r^{r}_t\right| \le \sup_{0\le t\le T}\left|r^{\underline{r}_n}_t-r^{r}_t\right| + \sup_{0\le t\le T}\left|r^{\bar{r}_n}_t-r^{r}_t\right|.
\]
By virtue of the first part of the proof, the terms on the right-hand side converge to $0$ as $n \to \infty$, which proves \eqref{eq:ucr}. The verification of \eqref{eq:ucx} is easy using the representation formula \eqref{eq:X2} for $X$ and \eqref{eq:ucr}.
\end{proof}

\begin{lemma}\label{cor:conv_tau}
Let $(t_n,r_n,x_n)_{n\ge 1}$ be a sequence in $\cC$ converging to $(t,r,x)\in \overline{\cC} \cap \cO$ as $n\to \infty$. Then
\[
\lim_{n\to\infty}\tau_*(t_n,r_n,x_n) = \tau_*(t,r,x),\qquad\text{$\P$-a.s.}
\]
\end{lemma}
\begin{proof}[Proof of Lemma \ref{cor:conv_tau}]
The proof relies on known facts from the theory of Markov processes, which we summarise in Appendix \ref{app:cont} for the reader's convenience, combined with Proposition \ref{prop:reg}. Proposition \ref{prop:reg} and Lemma \ref{lem:runif} imply that Assumptions \ref{ass:reg} and \ref{ass:cont} are satisfied for $\cDp = \cD \cap \cO$. It is also immediate to see that $\sigma_\cD = \sigma_{\cDp}$ $\P$-a.s. with $\sigma_{\cDp}$ defined in \eqref{eqn:sigma_cDp}.

The continuity of trajectories of $(r, X)$ means that the process cannot jump instantaneously to the stopping set $\cD$ when starting from $\cC$, so $\P_{\hat t,\hat r,\hat x}(\tau_*=\sigma_\cD)=1$ for any $(\hat t,\hat r,\hat x)\in\cC$. When $(\hat t,\hat r,\hat x) \in \partial \cC$, by construction we have $\tau_*(\hat t,\hat r,\hat x) = 0$, $\P$-a.s., and, using Proposition \ref{prop:reg}, $\sigma_\cD(\hat t,\hat r,\hat x) = 0$, $\P$-a.s. Recalling that $\overline{\cC} \cap \cO = \cC \cup \partial \cC$, the claim then follows from Proposition \ref{prop:rsigma}.
\end{proof}

Next we provide gradient estimates based on probabilistic arguments.
\begin{proposition}\label{prop:grad}
Let $\cK\subset\cO$ be a compact set with non-empty interior. There is $L = L(\cK) >0$ such that for any $(t,r,x)\in(\interior(\cK)\setminus\partial\cC)$ we have
\begin{align}
\label{eq:vx} v_x(t,r,x)&=-\E_{t,r,x}\left[\mathds{1}_{\{X_{\tau_*}\le K\}}e^{\sigma B_{\tau_*}-\frac{\sigma^2}{2}\tau_*}\right],\\
\label{eq:vt}0 \ge v_t(t,r,x)&\ge -L\;\E_{t,r,x}\left[e^{-\int_0^{\tau_\cK}r^r_sds}\mathds{1}_{\{\tau_\cK\le\tau_*\}}\right],
\end{align}
where $\tau_\cK:=\inf\{s\ge0\,:\,(t+s,r_s,X_s)\notin\interior(\cK)\}$.
\end{proposition}

\begin{remark}
Later on we obtain also a bound on the derivative $v_r$ of the value function with respect to the interest rate. We present it separately in \eqref{eq:v-diff-r-4.5} because, due to the square root appearing in the diffusion coefficient of the CIR dynamics, we need to use local approximations of the stochastic dynamics of $(r_t)_{t\ge 0}$. That procedure does not lead to a neat expression as in \eqref{eq:vx} and \eqref{eq:vt}. 
\end{remark}

\begin{proof}[Proof of Proposition \ref{prop:grad}] 

Fix $(t,r,x)\in(\interior(\cK)\setminus\partial\cC)$. Recall that $\cD\subset[0,T]\times\cI\times[0,K]$. If $(t,r,x)\in\text{int}(\cD)$ then \eqref{eq:vx} follows easily from $v(t,r,x)=K-x$ and $v_t(t,r,x) = 0$. Assume $(t,r,x)\in\cC$ and notice that $\tau_*=\sigma_\cD$, $\P_{t,r,x}$-a.s. We split the proof into two parts.

(\emph{Proof of \eqref{eq:vx}}) 
For all sufficiently small $\eps > 0$ we have $(t,r,x+\eps)\in \cC$. From now on, consider such $\eps$. To simplify notation let $\sigma_\cD:=\sigma_\cD(t,r,x)$. Using that $\sigma_\cD$ is admissible and sub-optimal for $v(t,r,x+\eps)$ we get
\begin{align*}
&v(t,r,x+\eps)-v(t,r,x)\\
&\ge \E\left[e^{-\int_0^{\sigma_\cD}r^r_sds}\left(\big(K-(x+\eps)X^{r,1}_{\sigma_\cD}\big)^+-\big(K-xX^{r,1}_{\sigma_\cD}\big)^+\right)\right]\nonumber\\
&\ge \E\left[e^{-\int_0^{\sigma_\cD}r^r_sds}\mathds{1}_{\{X^{r,x}_{\sigma_\cD}\le K\}}\left(xX^{r,1}_{\sigma_\cD}-(x+\eps)X^{r,1}_{\sigma_\cD}\right)\right]
=-\eps\E\Big[\mathds{1}_{\{X^{r,x}_{\sigma_\cD}\le K\}}e^{\sigma B_{\sigma_\cD}-\frac{\sigma^2}{2}\sigma_\cD}\Big]\nonumber.
\end{align*}
Dividing the above expression by $\eps$ and taking limits as $\eps\to0$ we get
\begin{align}\label{eq:vx1}
v_x(t,r,x)=\lim_{\eps\to0}\frac{1}{\eps}\left(v(t,r,x+\eps)-v(t,r,x)\right)\ge - \E\left[\mathds{1}_{\{X^{r,x}_{\sigma_\cD}\le K\}}e^{\sigma B_{\sigma_\cD}-\frac{\sigma^2}{2}\sigma_\cD}\right].
\end{align}

For the reverse inequality we use that $\sigma_\cD$ is admissible and sub-optimal for $v(t,r,x-\eps)$:
\begin{align*}
v(t,r,x)-v(t,r,x-\eps) &\le \E\left[e^{-\int_0^{\sigma_\cD}r^r_sds}\left(\big(K-x X^{r,1}_{\sigma_\cD}\big)^+-\big(K-(x-\eps) X^{r,1}_{\sigma_\cD}\big)^+\right)\right]\\
&\le-\eps\,\E\left[\mathds{1}_{\{X^{r,x-\eps}_{\sigma_\cD}\le K\}}e^{\sigma B_{\sigma_\cD}-\frac{\sigma^2}{2}\sigma_\cD}\right]
\le-\eps\,\E\left[\mathds{1}_{\{X^{r,x}_{\sigma_\cD}\le K\}}e^{\sigma B_{\sigma_\cD}-\frac{\sigma^2}{2}\sigma_\cD}\right],
\end{align*}
where in the last inequality we used that $X^{r,x-\eps}_s < X^{r,x}_s$, $s\ge 0$. Divide the above expression by $\eps$ and take limits as $\eps\to0$:
\begin{align}\label{eq:vx2}
v_x(t,r,x)=\lim_{\eps\to0}\frac{1}{\eps}\left(v(t,r,x)-v(t,r,x-\eps)\right)\le -\E\left[\mathds{1}_{\{X^{r,x}_{\sigma_\cD}\le K\}}e^{\sigma B_{\sigma_\cD}-\frac{\sigma^2}{2}\sigma_\cD}\right].
\end{align}
Now \eqref{eq:vx1} and \eqref{eq:vx2} imply \eqref{eq:vx}.
\vspace{+5pt}

(\emph{Proof of \eqref{eq:vt}}) The upper bound $v_t(t,r,x) \le 0$ follows from the monotonicity of $v$ in $t$ (Proposition \ref{prop:vmonot}). For all sufficiently small $\eps > 0$ we have $(t+\eps,r,x)\in \cK \cap \cC$ and $\tau_\cK := \tau_\cK(t,r,x) \le T-t-\eps$. From now on, consider such $\eps$. Denote $\sigma_\cD:=\sigma_\cD(t,r,x)$. Thanks to the choice of $\eps$, the stopping time $\eta:=\sigma_\cD\wedge\tau_\cK$ is admissible for $v(t+\eps,r,x)$. Using the (super)martingale property of $v$ (see \eqref{eq:supmg}--\eqref{eq:mg}) we get
\begin{align}\label{eqn:vta}
&v(t+\eps,r,x)-v(t,r,x)\nonumber\\
&\ge \E\left[e^{-\int_0^\eta r^r_sds}\left(v(t+\eps+\eta,r^r_\eta,X^{r,x}_\eta)-v(t+\eta,r^r_\eta,X^{r,x}_\eta)\right)\right]\\
&= \E\left[e^{-\int_0^{\tau_\cK} r^r_sds}\left(v(t+\eps+\tau_\cK,r^r_{\tau_\cK},X^{r,x}_{\tau_\cK})-v(t+{\tau_\cK},r^r_{\tau_\cK},X^{r,x}_{\tau_\cK})\right)\mathds{1}_{\{\tau_\cK<\sigma_\cD\}}\right],\nonumber
\end{align} 
where the equality follows from $v(t+\eps+\sigma_\cD,r^r_{\sigma_\cD},X^{r,x}_{\sigma_\cD})=v(t+{\sigma_\cD},r^r_{\sigma_\cD},X^{r,x}_{\sigma_\cD})=K - X^{r,x}_{\sigma_\cD}$ on $\{\tau_\cK\ge\sigma_\cD\}$ since $t\mapsto b(t,r)$ is non-decreasing (Proposition \ref{prop:D}). Let $\cK^\delta = \{ (t+s, r, x)\,:\, (t,r,x) \in \cK \text{ and } s \in [0, \delta]\}$. Fix a sufficiently small $\delta > 0$ so that this set is contained in $\cO$ and set $L$ equal to the Lipschitz constant for $v$ on $\cK^\delta$ (c.f. Proposition \ref{prop:vlip}). Since $(t+{\tau_\cK},r^r_{\tau_\cK},X^{r,x}_{\tau_\cK})\in\partial\cK$, we have $(t+\eps+{\tau_\cK},r^r_{\tau_\cK},X^{r,x}_{\tau_\cK}) \in \cK^\delta$ for any $\eps<\delta$. Using the Lipschitz continuity of $v$, we bound \eqref{eqn:vta} from below by 
\[
-\eps\,L\,\E\left[e^{-\int_0^{\tau_\cK} r^r_sds}\mathds{1}_{\{\tau_\cK<\sigma_\cD\}}\right].
\]
Dividing by $\eps$ and taking the limit $\eps \to 0$ completes the proof of \eqref{eq:vt}.
\end{proof}

We are now ready to prove that the value function is globally continuously differentiable on $\cO$.

\begin{proof}[Proof of Theorem \ref{prop:v-diff}]
It suffices to show that the value function has continuous partial derivatives across the stopping boundary, that is 
\begin{align}
&\lim_{n \to \infty} v_t(t_n,r_n,x_n)= \lim_{n \to \infty} v_r(t_n,r_n,x_n) = 0,\label{eq:v-diff-t}\\
&\lim_{n \to \infty} v_x(t_n,r_n,x_n)=-1\label{eq:v-diff-x}, 
\end{align}
for any sequence $(t_n,r_n,x_n)$ in $\cC$ converging to $(t_0,r_0,x_0)\in \partial\cC$ as $n\to \infty$. Fix such a sequence and denote $\tau_n = \tau_*(t_n,r_n,x_n)$.
\smallskip

\emph{Convergence of $v_x$.} Note that $\P_{t_n,r_n,x_n}(X_{\tau_n}= K,\, \tau_n < T\!-\!t_n) = 0$ (Proposition \ref{prop:boundary-c}) and
$\P_{t_n,r_n,x_n}(X_{\tau_n} =  K,\,\tau_n = T\!-\!t_n) \le  \P_{t_n,r_n,x_n}(X_{T-t_n} =  K) =0$ (the final equality can be shown by arguments as in \eqref{eq:XK}).
From Proposition \ref{prop:grad} we therefore have
\[
v_x(t_n,r_n,x_n)=-\E \left[\mathds{1}_{\{X^{r_n,x_n}_{\tau_n}< K\}}e^{\sigma B_{\tau_n}-\frac{\sigma^2}{2}\tau_n}\right].
\]
From Lemma \ref{cor:conv_tau}, we obtain $\lim_{n \to \infty} \tau_n = 0$ $\P$-a.s. We know from $(t_0, r_0, x_0) \in \partial \cC$ that $x_0 < K$. Lemma \ref{lem:runif} and the continuity of trajectories of $(r, X)$ imply the convergence $\mathds{1}_{\{X^{r_n,x_n}_{\tau_n}< K\}}\rightarrow \mathds{1}_{\{x_0< K\}}=1$ as $n \to \infty$. An application of the dominated convergence theorem completes the proof of \eqref{eq:v-diff-x}.

\emph{Convergence of $v_t$.} Let $\cK$ be a closed ball centered on $(t_0, r_0, x_0)$ and contained in $\cO$. With no loss of generality (by discarding a finite number of initial elements of the sequence) we assume that $(t_n,r_n,x_n)\in int(\cK)$ for all $n\ge 1$. Let 
\[
\tau^n_{\cK}:=\inf\{s\ge 0: (t_n+s,r_s^{r_n},X_{s}^{r_n,x_n})\notin \cK\}, \quad n\ge 0
\]
and notice, in particular, that $\P(\tau^0_{\cK} > 0) = 1$. The boundary $\partial \cK$ is regular for $\cO\setminus\cK$ and $(t, r, X)$ by the same reasoning as in the proof of Proposition \ref{prop:reg}. Repeating arguments from the proof of Lemma \ref{cor:conv_tau} shows that $\tau^n_{\cK}\to \tau^0_{\cK}$, $\P$-a.s. Fix $\eps\in (0,1)$. Since $\P(\tau^0_{\cK}>0)=1$, there exists $\delta>0$ such that $\P(\tau^0_{\cK}> \delta)\ge 1-\eps$. From inequality \eqref{eq:vt}, we get
\begin{equation}
\label{eq:v-diff-t-0}
\begin{aligned}
0 \ge v_t(t_n,r_n,x_n)&\ge -L\, \E \left[e^{-\int_0^{\tau^n_\cK}r^{r_n}_sds}\mathds{1}_{\{\tau^n_\cK\le\tau_n\}}\right]\\
&=-L\, \E \left[e^{-\int_0^{\tau^n_\cK}r^{r_n}_sds}\left(\mathds{1}_{\{\tau^n_\cK\le\tau_n\}\cap\{\tau^n_\cK\ge\delta\}}+\mathds{1}_{\{\tau^n_\cK\le\tau_n\}\cap\{\tau^n_\cK<\delta\}}\right)\right]\\
&\ge -L\, \E \left[e^{-\int_0^{\tau^n_\cK}r^{r_n}_sds}\left(\mathds{1}_{\{\tau_n \ge \delta\}}+\mathds{1}_{\{\tau^n_\cK<\delta\}}\right)\right].
\end{aligned}
\end{equation}
Using that $|r_{t\wedge\tau^n_{\cK}}|$ is bounded by some constant $r_\cK$ for every $n$, we have
\begin{equation}
\label{eq:v-diff-t-1}
0 \ge v_t(t_n,r_n,x_n)\ge -L e^{r_\cK T}\, \left(\P \left(\tau_n \ge \delta\right)+\P\left(\tau^n_\cK<\delta\right)\right).
\end{equation}
Lemma \ref{cor:conv_tau} guarantees that $\tau_n \to 0$ $\P$-a.s., so the first term converges to $0$ as $n \to \infty$ by the dominated convergence theorem. Fatou's lemma gives a bound for the second term:
\[
\limsup_{n \to \infty} \P\left(\tau^n_\cK<\delta\right) \le \E \Big[ \limsup_{n \to \infty} \mathds{1}_{\{\tau^n_\cK < \delta \}} \Big] \le \E \left[\mathds{1}_{\{\tau^0_\cK \le \delta \}} \right] \le \eps,
\]
where we used that $\limsup_n 1_{A_{n}}=1_{\limsup_n A_n}$ and the convergence of the stopping times. We obtain the convergence of $v_t$ in \eqref{eq:v-diff-t} by sending $\eps \to 0$.
\smallskip

\emph{Convergence of $v_r$.} Consider a sequence $(t_n,r_n,x_n)\in \cC$ converging to $(t_0,r_0,x_0) \in\partial \cC$. Since $\partial\cC$ is the boundary of $\cC$ in $\cO$, without loss of generality we can assume that
\[
\{(r_n,x_n)\} \subset \text{int}(\cK)\quad \text{with}\quad \cK:= [r_a,r_b]\times[x_a,x_b] \subset (\underline r,\overline{r})\times \R_+.
\]
Denote $\cK^T:=[t_a,t_b]\times \cK_0$, where $t_a = \inf_n t_n \ge 0$ and $t_b = \sup_n t_n < T$. 

We know that $v_r\le 0$ on $\cC$ (Proposition \ref{prop:vmonot}).  We will now develop a lower bound for $v_r$ on $\cC \cap \cK^T$, which will allow us to show that $v_r(t_n,r_n,x_n) \to 0$ as $n\to \infty$. Let $\tilde \cK \subset (\underline r,\overline{r})\times \R_+$ be compact and such that $\cK \subset \text{int}(\tilde \cK)$. Denote $\tilde \cK^T = [t_a, \tilde t_b] \times \tilde \cK$ for some $\tilde t_b \in (t_b, T)$. For $(t,r,x)\in\cC\cap\cK^T$ we define
\[
\tau_{\tilde \cK}(t,r,x):=\inf\{s\ge 0: (r^r_s, X^{r,x}_s)\notin \tilde \cK\}\wedge (T-t).
\]
By the monotonicity of $r \mapsto r^r_s$ and the explicit expression \eqref{eq:X2} for $X^{r,x}$ we have, for all $(r,x) \in \cK$,
\[
r^{r_a}_s \le r^r_s \le r^{r_b}_s, \quad\text{and}\quad X^{r_a, x}_s \le X^{r,x}_s \le X^{r_b, x}_s,  \quad \text{$\P$-a.s.}
\]
from which it is not hard to verify that $\tau_{\tilde \cK}(t,r,x) \ge \widehat{\tau}_{\cK} > 0$, $\P$-a.s., for all $(t,r,x)\in\cC\cap\cK^T$, where
\[
\widehat{\tau}_{\cK}:=\tau_{\tilde \cK}(t_b,r_a,x_a)\wedge \tau_{\tilde \cK}(t_b,r_a,x_b) \wedge \tau_{\tilde \cK}(t_b,r_b,x_a) \wedge \tau_{\tilde \cK}(t_b,r_b,x_b).
\]

Take $(t,r,x)\in \cC\cap \text{int}(\cK^T)$. There is $\overline \eps > 0$ such that $(t,r+\eps,x)\in \cC\cap \cK^T$ for all $\eps \in (0, \overline\eps]$. Denote by $\tau_*$ the optimal stopping time for $(t,r,x)$. For any $\eps \in (0, \overline\eps]$, we apply the (super)martingale properties of the value function \eqref{eq:supmg}-\eqref{eq:mg} with the stopping time $\tau_*\wedge \widehat{\tau}_{\cK}$:
\begin{equation}
\label{eq:v-diff-r-1}
\begin{aligned}
0&\ge v(t,r+\eps,x)-v(t,r,x)\\
&\ge  \E \Bigl[e^{-\int_0^{\tau_*\wedge \widehat{\tau}_{\cK}}r^{r+\eps}_s ds}v\big(t+(\tau_*\wedge \widehat{\tau}_{\cK}), r^{r+\eps}_{\tau_*\wedge \widehat{\tau}_{\cK}}, X^{r+\eps, x}_{\tau_*\wedge \widehat{\tau}_{\cK}}\big)\\
&\hspace{23pt}-e^{-\int_0^{\tau_*\wedge \widehat{\tau}_{\cK}}r^{r}_s ds} v\big(t+(\tau_*\wedge \widehat{\tau}_{\cK}), r^{r}_{\tau_*\wedge \widehat{\tau}_{\cK}}, X^{r, x}_{\tau_*\wedge \widehat{\tau}_{\cK}}\big)\Bigr]\\
&\ge  \E \left[\mathds{1}_{\{\widehat{\tau}_{\cK}\le \tau_*\}}\left(e^{-\int_0^{\widehat{\tau}_{\cK}}r^{r+\eps}_s ds}v(t+\widehat{\tau}_{\cK}, r^{r+\eps}_{\widehat{\tau}_{\cK}}, X^{r+\eps, x}_{\widehat{\tau}_{\cK}})-e^{-\int_0^{\widehat{\tau}_{\cK}}r^{r}_s ds}v(t+\widehat{\tau}_{\cK}, r^{r}_{\widehat{\tau}_{\cK}}, X^{r, x}_{\widehat{\tau}_{\cK}})\right)\right] \\
&\quad +\E \left[\mathds{1}_{\{\widehat{\tau}_{\cK} > \tau_*\}}\left(e^{-\int_0^{\tau_*}r^{r+\eps}_s ds}(K-X^{r+\eps, x}_{\tau_*})^+-e^{-\int_0^{\tau_*}r^{r}_s ds}(K-X^{r, x}_{\tau_*})^+\right)\right]\\
&=:E_1+E_2,
\end{aligned}
\end{equation}
where for the final inequality we used that $v(t+\tau_*, r^{r}_{\tau_*}, X^{r, x}_{\tau_*})=(K-X^{r,x}_{\tau_*})^+$, $\P$-a.s.
Recalling that $r^{r+\eps}_s \ge r^{r}_s$ and $v$ is non-negative we have
\begin{equation}
\label{eq:v-diff-r-2}
\begin{aligned}
E_1 &=\E \left[\mathds{1}_{\{\widehat{\tau}_{\cK}\le \tau_*\}} e^{-\int_0^{\widehat{\tau}_{\cK}}r^{r+\eps}_s ds}\left(v(t+\widehat{\tau}_{\cK}, r^{r+\eps}_{\widehat{\tau}_{\cK}}, X^{r+\eps, x}_{\widehat{\tau}_{\cK}})-v(t+\widehat{\tau}_{\cK}, r^{r}_{\widehat{\tau}_{\cK}}, X^{r, x}_{\widehat{\tau}_{\cK}})\right)\right]\\
&\hspace{12pt} - \E \left[\mathds{1}_{\{\widehat{\tau}_{\cK}\le \tau_*\}}\Big(e^{-\int_0^{\widehat{\tau}_{\cK}}r^{r}_s ds} - e^{-\int_0^{\widehat{\tau}_{\cK}}r^{r+\eps}_s ds} \Big) v(t+\widehat{\tau}_{\cK}, r^{r}_{\widehat{\tau}_{\cK}}, X^{r, x}_{\widehat{\tau}_{\cK}}) \right]\\
&\ge -L \, \E \left[\mathds{1}_{\{\widehat{\tau}_{\cK}\le \tau_*\}}e^{-\int_0^{\widehat{\tau}_{\cK}}r^{r+\eps}_s ds}\left(|r^{r+\eps}_{\widehat{\tau}_{\cK}}-r^{r}_{\widehat{\tau}_{\cK}}| + |X^{r+\eps,x}_{\widehat{\tau}_{\cK}}-X^{r,x}_{\widehat{\tau}_{\cK}}|\right)\right]\\
&\hspace{12pt} -  C \E \left[\mathds{1}_{\{\widehat{\tau}_{\cK}\le \tau_*\}}\Big(e^{-\int_0^{\widehat{\tau}_{\cK}}r^{r}_s ds} - e^{-\int_0^{\widehat{\tau}_{\cK}}r^{r+\eps}_s ds} \Big) \right],
\end{aligned}
\end{equation}
where the second inequality comes from the local Lipschitz property of the value function ($L>0$ is the constant from Proposition \ref{prop:vlip}), and the function $v$ is bounded from above by a constant $C$ depending on $\cK$, due to Assumption \ref{ass:coef}-(ii).

We shall now use the differentiability of the diffusion flow $(r^r_s)$ with respect to the parameter $r$ in the sense of \cite[Ch.~2, Sec.~8, Thm.~6]{krylov}. Apart from other assumptions, this requires that the coefficients be globally Lipschitz. As we only consider $(r, X)$ in a compact set $\cK$, we construct a two dimensional diffusion $(\widetilde{r}, \widetilde{X})$ whose coefficients coincide with the coefficients of $(r, X)$ on $\cK$, are globally Lipschitz, continuously differentiable and with a polynomial growth. The process $(\widetilde{r}_s, \widetilde{X}_s)$ is indistinguishable from $(r_s,X_s)$ on $\{ s \le \widehat\tau_\cK\}$, i.e., on the set where it is of interest for the estimation of $E_1$ and $E_2$, so for the sake of readability we will write $(r,X)$ in the estimates below (we use an analogous construction in Appendix \ref{app:reg}, where full details are available).

By \cite[Ch.~2, Sec.~8, Thm.~6]{krylov}, there is a measurable in $(s,\omega)$ process $(y^r_s(\omega))_{s\ge 0}$, depending on $r$, such that for any $q \ge 1$
\begin{equation}\label{eqn:diff}
\lim_{\eps \downarrow 0} \bigg\| \sup_{s \in [0, T]} \Big|\frac{r^{r+\eps}_s - r^r_s}{\eps} - y^r_s\Big| \bigg\|_q = 0 \quad\text{and}\quad
\lim_{\eps \downarrow 0} \Big\|\frac{r^{r+\eps}_\cdot - r^r_\cdot}{\eps} - y^r_\cdot\Big\|^*_q = 0,
\end{equation}
where $\|Z\|_{q} = (\E[|Z|^q])^{1/q}$ and $\|Y_\cdot\|^*_q = \big(\E \big[ \int_0^T |Y_s|^q ds \big] \big)^{1/q}$.

Fix $\frac{1}{p} + \frac{1}{q} + \frac{1}{w} = 1$ for some $p \in (1,2]$. Recalling that $r^{r+\eps}_s \ge r^r_s$ and using
H\"{o}lder inequality yields
\begin{equation}
\label{eq:v-diff-r-3}
\begin{aligned}
&\frac{1}{\eps} \E \left[\mathds{1}_{\{\widehat{\tau}_{\cK}\le \tau_*\}}e^{-\int_0^{\widehat{\tau}_{\cK}}r^{r+\eps}_s ds}|r^{r+\eps}_{\widehat{\tau}_{\cK}}-r^{r}_{\widehat{\tau}_{\cK}}| \right]\\
&\le
\E \left[\mathds{1}_{\{\widehat{\tau}_{\cK}\le \tau_*\}}e^{-\int_0^{\widehat{\tau}_{\cK}}r^{r+\eps}_s ds}\left(\Big|\frac{1}{\eps}\big(r^{r+\eps}_{\widehat{\tau}_{\cK}}-r^{r}_{\widehat{\tau}_{\cK}}) - y^r_{\widehat \tau_\cK}\Big| + |y^r_{\widehat \tau_\cK}| \right)\right]
\\
&\le
C_1^{1/p}
\P(\widehat{\tau}_{\cK}\le \tau_*)^\frac{1}{w}
\bigg( \Big\|\frac{1}{\eps}\big(r^{r+\eps}_{\widehat \tau_\cK} - r^{r}_{\widehat \tau_\cK}\big) - y^r_{\widehat \tau_\cK}\Big\|_q + \|y^r_{\widehat \tau_\cK}\|_q\bigg)\xrightarrow[\eps \downarrow 0]{}
C_1^{1/p} \P(\widehat{\tau}_{\cK}\le \tau_*)^\frac{1}{w}\ 
\|y^r_{\widehat \tau_\cK}\|_q,
\end{aligned}
\end{equation}
where we used the estimate \eqref{eq:integr} in the last inequality and \eqref{eqn:diff} to obtain the convergence.

To bound the last term on the right hand side of \eqref{eq:v-diff-r-2}, we observe that
\begin{align*}
\E \left[\mathds{1}_{\{\widehat{\tau}_{\cK}\le \tau_*\}}\Big(e^{-\int_0^{\widehat{\tau}_{\cK}}r^{r}_s ds} - e^{-\int_0^{\widehat{\tau}_{\cK}}r^{r+\eps}_s ds} \Big) \right]
&\le
\E \Big[\mathds{1}_{\{\widehat{\tau}_{\cK}\le \tau_*\}}e^{-\int_0^{\widehat{\tau}_{\cK}}r^{r}_s ds} \int_0^{\widehat{\tau}_{\cK}}(r^{r+\eps}_s - r^{r}_s) ds \Big].
\end{align*}
We then apply H\"{o}lder inequality and the second limit in \eqref{eqn:diff}:
\begin{equation}
\label{eq:v-diff-r-2a}
\begin{aligned}
&\frac{1}{\eps}\E \Big[\mathds{1}_{\{\widehat{\tau}_{\cK}\le \tau_*\}}e^{-\int_0^{\widehat{\tau}_{\cK}}r^{r}_s ds} \int_0^{\widehat{\tau}_{\cK}}(r^{r+\eps}_s - r^{r}_s) ds \Big]\\
&\le
\E \left[\mathds{1}_{\{\widehat{\tau}_{\cK}\le \tau_*\}} e^{-\int_0^{\widehat{\tau}_{\cK}}r^{r}_s ds} \int_0^{\widehat{\tau}_{\cK}} \Big|\frac{1}{\eps}(r^{r+\eps}_s - r^r_s) - y^r_s \Big| + |y^r_s| ds \right]\\
&\le
C_1^{1/p}\ \P(\widehat{\tau}_{\cK}\le \tau_*)^\frac{1}{w}
\bigg( \Big\|\frac{1}{\eps}(r^{r+\eps}_\cdot - r^r_\cdot) - y^r_\cdot \Big\|^*_q + \| y^r_\cdot \|^*_q  \bigg)\xrightarrow[\eps \downarrow 0]{}
C_1^{1/p}\ \P(\widehat{\tau}_{\cK}\le \tau_*)^\frac{1}{w}\ 
\| y^r_\cdot \|^*_q .
\end{aligned}
\end{equation}

By the explicit formula \eqref{eq:X2}, we have $X^{r,x}_t = e^{\int_0^t r^{r}_s ds} \hat X^x_t$, where $\hat X^x_t := x e^{\sigma B_t - \frac12 \sigma^2 t}$, and
\[
0 \le X^{r+\eps, x}_t-X^{r,x}_t \le e^{\int_0^t r^{r+\eps}_s ds} \hat X^x_t \int_0^t (r^{r+\eps}_s - r^r_s) ds.  
\]
We proceed similarly as in \eqref{eq:v-diff-r-2a} to obtain
\begin{equation}
\label{eq:v-diff-r-3a}
\begin{aligned}
&\frac{1}{\eps} \E \left[\mathds{1}_{\{\widehat{\tau}_{\cK}\le \tau_*\}}e^{-\int_0^{\widehat{\tau}_{\cK}}r^{r+\eps}_s ds} |X^{r+\eps,x}_{\widehat{\tau}_{\cK}}-X^{r,x}_{\widehat{\tau}_{\cK}}| \right]\\
&\le
\|\hat X^x_{\widehat{\tau}_{\cK}}\|_p\ \P(\widehat{\tau}_{\cK}\le \tau_*)^\frac{1}{w}
\bigg( \Big\|\frac{1}{\eps}(r^{r+\eps}_\cdot - r^r_\cdot) - y^r_\cdot \Big\|^*_q + \| y^r_\cdot \|^*_q  \bigg)\\
&\xrightarrow[\eps \downarrow 0]{}
\|\hat X^x_{\widehat{\tau}_{\cK}}\|_p\ \P(\widehat{\tau}_{\cK}\le \tau_*)^\frac{1}{w}\ 
\| y^r_\cdot \|^*_q .
\end{aligned}
\end{equation}

Similar arguments as above enable us to derive a lower bound for $E_2$:
\begin{equation}
\label{eq:v-diff-r-4}
\begin{aligned}
\frac{1}{\eps} E_2
&= 
\frac{1}{\eps} \E \left[\mathds{1}_{\{\widehat{\tau}_{\cK}> \tau_*\}}\left(\Big(Ke^{-\int_0^{\tau_*}r^{r+\eps}_s ds} -\hat X^{x}_{\tau_*}\Big)^+-\Big(Ke^{-\int_0^{\tau_*}r^{r}_s ds} -\hat X^{x}_{\tau_*}\Big)^+\right)\right]\\
&\ge -\frac{1}{\eps}\,K\, \E \left[\mathds{1}_{\{\widehat{\tau}_{\cK}> \tau_*\}}
\Big(e^{-\int_0^{\tau_{*}}r^{r}_s ds} - e^{-\int_0^{\tau_{*}}r^{r+\eps}_s ds} \Big)
\right]\\
&\ge
-\frac{1}{\eps}\,K\, \E \left[\mathds{1}_{\{\widehat{\tau}_{\cK}> \tau_*\}} e^{-\int_0^{\tau_{*}}r^{r}_s ds} 
\int_0^{\tau_*}(r^{r+\eps}_s - r^{r}_s) ds \right]\\
&\ge - K C_1^{1/p}\ \P(\widehat{\tau}_{\cK}> \tau_*)^\frac{1}{w}
\bigg( \Big\|\frac{1}{\eps}(r^{r+\eps}_\cdot - r^r_\cdot) - y^r_\cdot \Big\|^*_q + \Big(\E \Big[\int_0^{\tau_*} | y^r_s |^q ds \Big]\Big)^{1/q} \bigg)\\
&\xrightarrow[\eps \downarrow 0]{}
- K C_1^{1/p} \P(\widehat{\tau}_{\cK}> \tau_*)^\frac{1}{w}\ 
\Big(\E \Big[\int_0^{\tau_*} | y^r_s |^q ds \Big]\Big)^{1/q},
\end{aligned}
\end{equation}
where in the first inequality we used the Lipschitz property of $z \mapsto (z - \hat X^{x}_t(\omega))^+$ for any $\omega \in \Omega$.

Combining \eqref{eq:v-diff-r-3}--\eqref{eq:v-diff-r-4} gives a lower bound for $v_r$ on $\cC\cap\cK^T$:
\begin{equation}
\label{eq:v-diff-r-4.5}
\begin{aligned}
0 &\ge v_r (t,r,x)\\
 &\ge -
L\, \P(\widehat{\tau}_{\cK}\le \tau_*)^\frac{1}{w}\ \left(
C_1^{1/p} \|y^r_{\widehat \tau_\cK}\|_q + \|\hat X^x_{\widehat{\tau}_{\cK}}\|_p\ \|y^r_{\cdot}\|^*_q \right)
-K \, \P(\widehat{\tau}_{\cK}\le \tau_*)^\frac{1}{w}\
C_1^{1/p} \|y^r_{\cdot}\|^*_q 
\\
&\hspace{12pt}-
K C_1^{1/p} \P(\widehat{\tau}_{\cK}> \tau_*)^\frac{1}{w}\  \ 
\Big(\E \Big[\int_0^{\tau_*} | y^r_s |^q ds \Big]\Big)^{1/q}.
\end{aligned}
\end{equation}
By \cite[Ch.~2, Sec.~8, Thm.~8]{krylov} and standard diffusion estimates \citep[Ch.~2, Sec.~5, Cor.~10]{krylov} the norms of $y^r$ and $\hat X^x$ above are bounded uniformly for $(t, r,x)\in \cK^T \cap \cC$ (recall that $\tau_* = \tau_*(t,r,x)$). Now take $(t,r,x)=(t_n,r_n,x_n)$ in \eqref{eq:v-diff-r-4.5}. Since $\hat \tau_\cK > 0$ $\P$-a.s.\ and $\lim_{n \to \infty} \tau_*(t_n, r_n, x_n) = 0$ $\P$-a.s.~by Lemma \ref{cor:conv_tau}, the dominated convergence theorem gives that the first two terms of \eqref{eq:v-diff-r-4.5} tend to zero as $n \to \infty$ due to $\P(\widehat{\tau}_{\cK}\le \tau_*(t_n, r_n, x_n)) \to 0$ and the last term converges to zero because 
\[
\lim_{n \to 0} \E \Big[\int_0^{\tau_*(t_n, r_n, x_n)} | y^{r_n}_s |^q ds \Big] = 0
\]
and the mapping $r \mapsto y^r$ is continuous in the norm $\|\cdot\|^*_q$, see \cite[Ch.~2, Sec.~8, Thm.~6]{krylov}. This concludes the proof.
\end{proof}

\section{Continuity of the stopping boundary and the integral equation}\label{sec:integral}

\begin{proof}[Proof of Proposition \ref{prop:conti-c}]
Since $c(t,x)=\overline{r}$ on $[0,T)\times [K,\infty)$, it remains to prove the continuity at $(t_0,x_0) \in (0,T)\times (0,K]$. It is known from Proposition \ref{prop:boundary-c} that $t\mapsto c(t,x_0)$ is non-increasing and right-continuous at $t_0$, and $x\mapsto c(t_0,x)$ is non-decreasing and left-continuous at $x_0$.

We first show that $x\mapsto c(t_0,x)$ is right continuous at $x_0$. It is obvious for $x_0 = K$ since $c(t_0, x) = \overline{r}$ for $x \ge K$. We proceed with an argument for $x_0 < K$. Assume, by contradiction, that $c(t_0,x_0+)>c(t_0,x_0)$, so there exist $r_1,r_2$ such that $c(t_0,x_0+)>r_2>r_1>c(t_0,x_0)$. Let $R:=(r_1,r_2) \times (x_0,x_1)$ for some $x_1 \in (x_0, K)$ and $R_0:=(r_1,r_2) \times \{x_0\}$. From the monotonicity of $c(t,x)$, we have $\{t_0\} \times R \subset \cC$ and $\{t_0\} \times R_0 \subset \cD$. Let $u$ be a function defined on $\overline{R}$ and satisfying 
\begin{equation}
\label{eq:pde1}
\begin{aligned}
(\cL-r) u(r,x)&=-v_t(t_0,r,x),\qquad (r,x) \in R,\\
u(r,x)&=v(t_0,r,x),\qquad (r,x) \in \partial R.
\end{aligned}
\end{equation}
Thanks to \cite[Theorem 10, p.\ 72]{Friedman} we know that $(r,x) \mapsto v_t(t_0,r,x)$ is $C^1$ on $R$ with H\"older continuous derivatives. Since the coefficients of \eqref{eq:LrX} have H\"older continuous first derivatives, there is a unique classical solution $u(r,x)$ of the above PDE (which is of elliptic type) and $u\in C^3(R)\cap C(\overline{R})$ \citep[Theorems 19 and 20, p.\ 87]{Friedman}. From \eqref{eqn:pde}, the function $(r,x) \mapsto v(t_0,r,x)$ satisfies \eqref{eq:pde1}, so, by uniqueness, $u=v$ on $\overline{R}$ and $u\in C^1(\overline{R})$ by Theorem \ref{prop:v-diff}.

We differentiate the PDE in \eqref{eq:pde1} with respect to $r$ and obtain
\begin{align}
\label{eq:c-cts-1}
\frac{1}{2}\sigma^2x^2 u_{rxx}(r,x) &=-\cL_1 u_r(r,x)\!-\!\cL_2 u_x(r,x)\!-\!xu_x(r,x)\!-\!v_{tr}(t_0,r,x)\!+\!u(r,x),\quad (r,x) \in R,
\end{align}
where
\begin{align*}
\cL_1 f:=&\frac{1}{2}\beta^2(r)f_{rr}+\left(\beta(r)\beta'(r)+\alpha(r)\right)f_{r}+\left(\alpha'(r)-r\right)f\\
\cL_2 f:=&\rho\sigma\beta(r)x f_{rr}+\left(\rho\sigma\beta'(r)+rx\right)f_{r}.
\end{align*}
Let $\phi$ be a $C^{\infty}$ function with compact support on $(r_1,r_2)$ such that $\int^{r_2}_{r_1}\phi(r)dr=1$ and for $x\in(x_0,x_1)$ define
\[
F_{\phi}(x)=-\int^{r_2}_{r_1}u_{xx}(r,x)\phi'(r)dr.
\]
Multiply \eqref{eq:c-cts-1} by $\frac{2}{\sigma^2x^2}\phi(r)$ and integrate over $(r_1,r_2)$:
\begin{align*}
\int^{r_2}_{r_1}u_{rxx}(x,r)\phi(r)dr
&=
-\int^{r_2}_{r_1}\frac{2}{\sigma^2x^2}\phi(r)\big[ \cL_1 u_r(r,x) + \cL_2 u_x(r,x)\big] dr
-\int^{r_2}_{r_1}\frac{2}{\sigma^2x}\phi(r)u_x(r,x)dr\\
&\quad-\int^{r_2}_{r_1}\frac{2}{\sigma^2x^2}\phi(r)v_{tr}(t_0,r, x)dr+ \int^{r_2}_{r_1}\frac{2}{\sigma^2x^2}\phi(r)u(r,x)dr.
\end{align*}
Intergration by parts gives
\begin{equation}
\label{eq:c-cts-2}
\begin{aligned}
F_{\phi}(x)
&=
-\int^{r_2}_{r_1}\frac{2}{\sigma^2x^2} \big[ u_r(r,x)\cL^*_1 \phi(r)dr + u_x(r,x)\cL^*_2 \phi(r)\big] dr
-\int^{r_2}_{r_1}\frac{2}{\sigma^2x}\phi(r)u_x(r,x)dr\\
&\quad +\int^{r_2}_{r_1}\frac{2}{\sigma^2x^2}\phi'(r)v_{t}(t_0,r, x)dr
+ \int^{r_2}_{r_1}\frac{2}{\sigma^2x^2}\phi(r)u(r,x)dr,
\end{aligned}
\end{equation}
where $\cL^*_1$ and $\cL^*_2$ are adjoint operators to $\cL_1$ and $\cL_2$, respectively. The expression above involves only $u$ and its first derivatives, which are continuous by Theorem \ref{prop:v-diff}. We take the limit $x \to x_0$ in \eqref{eq:c-cts-2} and notice that
$u_r(r,x_0)=v_r(t_0, r,x_0)=v_t(t_0,r,x_0)=0$, $u_x(r,x_0)=v_x(t_0, r,x_0)=-1$ and $u(r,x_0)=K-x_0$. Thus,
\[
\lim_{x\downarrow x_0} F_{\phi}(x)
=
\int^{r_2}_{r_1}\frac{2}{\sigma^2x_0}\phi(r)dr
+ \int^{r_2}_{r_1}\frac{2}{\sigma^2x_0^2}\phi(r)(K-x_0)dr
=\frac{2K}{\sigma^2x_0^2} 
>0,
\]
where we also use that $\int_{r_1}^{r_2}\cL^*_2\phi(r)dr=0$.
Since $x \mapsto F_{\phi}(x)$ is continuous on $(x_0, x_1)$ and $\lim_{x\downarrow x_0} F_{\phi}(x)>0$, we have $F_{\phi}(x)>0$ on $(x_0,x_0 + \eps)$ for any sufficiently small $\eps > 0$. Using additionally that $u$ is $C^1(\overline{R})$, we perform the following integration
\begin{align*}
0<\int_{x_0}^{x_0+\eps}\int_{x_0}^{y}F_{\phi}(x)dxdy=&-\int^{r_2}_{r_1}\int_{x_0}^{x_0+\eps}\int_{x_0}^{y}u_{xx}(r,x)dxdy\,\phi'(r) dr\\
                                  = &-\int^{r_2}_{r_1}\int_{x_0}^{x_0+\eps}(u_{x}(r,y)+1)dy\,\phi'(r)dr\\
                                  =&-\int^{r_2}_{r_1}(u(r,x_0+\eps)-(K-x_0)+\eps)\phi'(r)dr\\
                                  =&\int^{r_2}_{r_1}u_{r}(r,x_0+\eps)\phi(r)dr,
\end{align*}
where we have used Fubini's theorem in the first equality, $u_x(r, x_0) = -1$ in the second equality, $u(r, x_0) = K - x_0$ in the third equality, and the integration by parts in the last equality. As the above inequality holds for an arbitrary smooth function $\phi$ with a compact support in $(r_1, r_2)$, we must have $u_{r}(r,x_0+\eps)=v_{r}(t_0, r,x_0+\eps)>0$ almost everywhere on $(r_1, r_2)$. This contradicts that $r\mapsto v(t_0,r,x_0+\eps)$ is a non-increasing function (see Proposition \ref{prop:vmonot}), hence $x\mapsto c(t,x)$ is continuous.

We turn our attention to the left-continuity of $t\mapsto c(t,x_0)$ at $t_0$ (the right-continuity has already been established in Proposition \ref{prop:boundary-c}). Assume, by contradiction, that the left-continuity fails at $t_0$. Since $t\mapsto c(t, x_0)$ is non-increasing, there exist $r_1,r_2$ such that 
$c(t_0-,x_0)>r_2>r_1>c(t_0,x_0)$. By the continuity of $x\mapsto c(t_0,x)$ at $x_0$ and the monotonicity of $c(t,x)$, there is $x_1 \in (x_0, K)$ such that $r_1>c(t_0,x_1)\ge c(t_0,x_0)$. Hence, for any sequence $t_n \uparrow t_0$, we have
\[
c(t_n,x_1)\ge c(t_n, x_0) \ge c(t_0-,x_0)>r_2>r_1>c(t_0,x_1)\ge c(t_0,x_0),
\]
so that
\begin{align*}
R:&=(t_1,t_0)\times (r_1,r_2)\times(x_0,x_1)\subset \cC, \\
R_{t_0}:&=\{t_0\}\times (r_1,r_2)\times(x_0,x_1)\subset \cD.
\end{align*}
Consider a PDE
\begin{equation}\label{eq:pde2}
\begin{aligned}
&w_t(t,r,x)+(\cL-r) w(t,r,x)=0,\qquad  (t,r,x) \in R,\\
& w(t,r,x)=v(t,r,x),\qquad (t,r,x) \in \partial_{p} R,
\end{aligned} 
\end{equation}
where $\partial_{p} R$ denotes the parabolic boundary of $R$. By \cite[Theorem 6, p. 65]{Friedman}, Equation \eqref{eq:pde2} admits a unique classical solution $w$, which coincides with $v$ on $\overline{R}$. This also implies that $w \in C^1(\overline{R})$ by Theorem \ref{prop:v-diff}.

Let $\phi_1$ be a $C^\infty$ function with compact support in $(x_0,x_1)$ and $\phi_2$ be a $C^\infty$ function with compact support in $(r_1,r_2)$ such that $\int_{x_0}^{x_1} \phi_1(x) dx = \int_{r_1}^{r_2}\phi_2(r) dr = 1$. Fixing $t = t_n\in(t_1,t_0)$ from the sequence $t_n\uparrow t_0$, we multiply \eqref{eq:pde2} by $\phi_1(x)\phi_2(r)$ and integrate over $(r_1,r_2)\times (x_0,x_1)$:
\begin{equation*}
\int_{r_1}^{r_2}\int_{x_0}^{x_1} \phi_1(x)\phi_2(r)\big\{w_t(t_n,r,x)+(\cL-r) w(t_n,r,x)\big\}dxdr=0.
\end{equation*} 
Integration by parts gives
\begin{equation}\label{eqn:s1}
\int_{r_1}^{r_2}\int_{x_0}^{x_1} \phi_1(x)\phi_2(r)w_t(t_n,r,x)dxdr +\int_{r_1}^{r_2}\int_{x_0}^{x_1} w(t_n,r,x)(\cL^*-r) \phi_1(x)\phi_2(r)dxdr=0,
\end{equation} 
where $\cL^*$ is the adjoint operator for $\cL$. When $n \to \infty$, the first integral vanishes since $w\in C^1(\overline{R_t})$ and $w_t=v_t=0$ on $R_{t_0}$. By the dominated convergence theorem, \eqref{eqn:s1} reads
\begin{align*}
0&=\int_{r_1}^{r_2}\int_{x_0}^{x_1} w(t_0,r,x)(\cL^*-r) \phi_1(x)\phi_2(r)dxdr
=\int_{r_1}^{r_2}\int_{x_0}^{x_1} \phi_1(x)\phi_2(r)(\cL-r) (K-x)dxdr\\
&=\int_{r_1}^{r_2}\int_{x_0}^{x_1} \phi_1(x)\phi_2(r)(-rK)dxdr=\int_{r_1}^{r_2}\phi_2(r)(-rK)dr
\end{align*} 
where we integrate by parts and use that $v(t,r,x) = (K-x)$ on $R_{t_0}$ for the second equality. We obtain a contradiction because the last integral is strictly negative.

Having established the continuity in $t$ and $x$ separately, the monotonicity of $c$ allows us to conclude the continuity of $(t,x)\mapsto c(t,x)$ at $(t_0,x_0)$ (see, e.g., \cite{kruse1969joint}).
\end{proof}

\begin{proof}[Proof of Proposition \ref{prop:int-eq}]
Let $\cK_n$ be an increasing sequence of compact subsets of $\cO$ such that $\cup_{n \in \N} \cK_n = \cO$ and define $\tau_n = \inf\{ t \in [0, T-t]:\ (t+s, r_s, X_s) \notin \cK_n \}\wedge(T-t-\frac{1}{n})$ for $n$ large enough so that $\frac1n \le T-t$. We apply a version of It\^{o} formula from \cite[Theorem 2.1]{cai2021b}, which we state in Appendix \ref{app:Ito} for the reader's convenience. We delay the verification of the assumptions required until the end of the proof. Using that
\begin{equation}\label{eqn:PDEv}
\begin{aligned}
&(\partial_t + \cL - r)v(t,r,x) = 0, & r < c(t, x),\\
&(\partial_t + \cL - r)v(t,r,x) = (\partial_t + \cL - r)(K-x) = -rK, & r > c(t, x),
\end{aligned}
\end{equation}
we obtain that the dynamics of the discounted value function on $[0, \tau_n]$ is given by
\begin{equation}
\begin{aligned}
&e^{-\int_{0}^{s\wedge{\tau_n}}r_v dv}v(t+s\wedge{\tau_n},r_{s\wedge{\tau_n}},X_{s\wedge{\tau_n}})\\
&=v(t,r,x) - \int_0^{s\wedge{\tau_n}}\!\!\!e^{-\int_{0}^{u}r_v dv}Kr_u \mathds{1}_{\{r_u>c(t+u,X_u)\}}du
+ \int_0^{s\wedge{\tau_n}}\!\!\!e^{-\int_{0}^{u}r_v dv}\sigma X_u v_x(t + u,r_u,X_u)dB_u\\
&\quad+\! \int_0^{s\wedge{\tau_n}}\!\!\!e^{-\int_{0}^{u}r_v dv}\beta(r_u)v_r(t + u,r_u,X_u)dW_u.
\end{aligned}
\end{equation}
Taking expectations and applying the optional sampling theorem we arrive at 
\begin{equation}\label{eqn:iq1}
v(t,r,x)
=\E_{r,x}\!\left[\int_{0}^{\tau\wedge\tau_n}\!\!\!e^{-\int_{0}^{u}r_v dv}Kr_u \mathds{1}_{\{r_u>c(t+u,X_u)\}}du\!+\!e^{-\int_{0}^{\tau\wedge \tau_n}r_v dv}v(t\!+\!(\tau\wedge\tau_n),r_{\tau\wedge\tau_n},X_{\tau\wedge\tau_n})\right].
\end{equation}
Using \eqref{eq:integr} and \eqref{eq:subg}, H\"{o}lder inequality implies
\[
\E_{r,x}\left[\int_{0}^{T-t}e^{-\int_{0}^{u}r_v dv} |Kr_u| du \right] < \infty.
\]
The majorant for the second term of \eqref{eqn:iq1} follows from Assumption \ref{ass:coef}
(details can be found in the proof of \eqref{eqn:i1} in Lemma \ref{lem:integr}). The dominated convergence theorem proves \eqref{eq:int-1}, since $\tau_n\uparrow T-t$ upon recalling that the boundary of $\cI\times\R_+$ is assumed non-attainable by the process $(r_t,X_t)$.

It remains to verify assumptions of \cite[Theorem 2.1]{cai2021b} as stated in Appendix \ref{app:Ito}. Identifying $X^1_t = r_t$ and $X^2_t = X_t$, we have,
\[
\beta^{1,1} (t,r,x) = \beta^2(r), \quad \beta^{1,2}(t,r,x) = \beta^{2,1}(t,r,x) = \sigma \rho \beta(r) x, \quad \beta^{2,2}(t,r,x) = \sigma^2 x^2.
\]
By Assumption \ref{ass:coef}, $\beta^{i,j}$ is Lipschitz for $i,j=1,2$ on every compact set in $\cO$. Indeed, it can be directly verified for the CIR process. In case (ii) of Assumption \ref{ass:coef} we use Lipschitz continuity of $\beta$. The marginal distribution of the process $(r_t, X_t)$ has density with respect to the Lebesgue measure (see Remark \ref{rem:lambda}, which makes use of Assumption \ref{ass:rbar}), so $(t, r_t, X_t) \notin \partial\cC$, $\P_{r,x}$-a.s. for any $t > 0$. Setting $\cC=\cE$, this verifies the first assumption in the theorem in Appendix \ref{app:Ito}. For the second assumption, using \eqref{eqn:PDEv}, we have
\[
\frac12 L(t, r,x) = - r x v_x(t,r,x) - \alpha(r) v_r(t,r,x) - v_t (t,r,x) + r v(t,r,x) - 1_{\{(t,r,x) \in \cD\}}\ rK.
\]
Since $v \in C^1(\cO)$ and the function $\alpha(r)$ is continuous (see Assumption \ref{ass:coef}), $L$ is continuous and bounded on $\cK_n \setminus \partial \cC$. We finally have that the third assumption in the theorem holds by Proposition \ref{prop:boundary-c}. 
\end{proof}

\begin{proof}[Proof of Proposition \ref{prop:int-eq-1}]
The proof follows ideas originally developed in \cite{peskir2005american}. Assume there exists another continuous function $\tc$ that satisfies conditions (1) and (2) in the statement of this proposition. Define a function
\begin{equation*}
\begin{aligned}
\tilde v(t,r,x) &=\E_{r,x}\!\left[\int_{0}^{T-t}\!\!e^{-\int_{0}^{u}r_v dv}Kr_u \mathds{1}_{\{r_u>\tc(t+u,X_u)\}}du\!+\!e^{-\int_{0}^{T-t}r_v dv}(K\! -\! X_{T-t})^+\right], \quad (t,r,x) \in \cO,\\
\tilde v(T,r,x) &= (K - x)^+, \quad (r,x) \in \cI \times \R_+.
\end{aligned}
\end{equation*}
It is not difficult to prove that $\tilde v$ is continuous by the continuity of $\tilde c$ and of the flow $(s,r,x)\mapsto (r^r_s,X^{r,x}_s)$.
By the Markov property of $(r,X)$, one can also check that
\begin{align*}
\tilde{V}_{s}:=\int_{0}^{s}e^{-\int_{0}^{u}r_v dv}Kr_u \mathds{1}_{\{r_u>\tc(t+u,X_u)\}}du+e^{-\int_{0}^{s}r_v dv} \tilde v(t+s,r_s,X_s),
\qquad s \in [0, T-t],
\end{align*}
is a continuous $\P_{r, x}$-martingale. Hence, for any $(t,r,x) \in \cO$ and any stopping time $\tau \le T-t$, the optional sampling theorem yields
\begin{equation}\label{eqn:dynkin_u}
\tilde v(t,r,x)=\E_{r,x}\big[\tilde{V}_{\tau}\big]
=\E_{r,x}\left[\int_{0}^{\tau}e^{-\int_{0}^{u}r_v dv}K r_u \mathds{1}_{\{r_u>\tc(t+u,X_u)\}}du+e^{-\int_{0}^{\tau}r_v dv}\tilde v(t+\tau,r_{\tau},X_{\tau})\right],
\end{equation}
which is analogous to the formula for $v$ in \eqref{eq:int-1}.

For an easier exposition of the arguments of proof we proceed in steps. In the first four steps we show the equality $\tilde c(t,x)=c(t,x)$ for all $(t,x)\in[0,T)\times\R_+$ such that $\tilde c(t,x)\in\cI$. Then, in the final step we use monotonicity and continuity of $\tilde c$ and $c$ to extend the equality to all $(t,x)\in[0,T)\times\R_+$.

{\em Step 1}. We first show that $\tilde v(t,r,x)=(K-x)^+$ for any $(t,r,x) \in \cO$ such that $r\ge \tc(t,x)$. Fix $(\hat t,\hat r,\hat x) \in \cO$ such that $\hat r>\tc(\hat t,\hat x)$ (the claim for $\hat r = \tc(\hat t,\hat x)$ follows by the continuity of $\tilde v$). Define a stopping time
\[
\tau_1:=\inf\{s\ge 0: r^{\hat r}_s\le \tc(\hat t+s,X^{\hat r,\hat x}_s)\}\wedge (T-\hat t).
\]
By the continuity of $s \mapsto \tc(\hat t+s, X_s)$ and  $s\mapsto r_s$, and the fact that $\underline r$ and $\overline r$ are unattainable by $(r_s)$, we have $\tc(\hat t + \tau_1, X_{\tau_1}) \in \cI$ on $\{\tau_1 < T-\hat t\}$. By assumption $\tilde v(t,\tc(t,x),x)=(K-x)^+$ and, consequently, $\tilde v\big(\hat t + \tau_1, \tc(\hat t + \tau_1, X_{\tau_1}), X_{\tau_1}\big) = (K-X_{\tau_1})^+$ since $\tilde v(T,r,x)=(K - x)^+$.
In combination with \eqref{eqn:dynkin_u}, this yields
\begin{align}\label{eqn:b1}
\tilde v(\hat t,\hat r,\hat x)
&=\E_{\hat r,\hat x}\left[\int_{0}^{\tau_1}e^{-\int_{0}^{u}r_v dv}Kr_udu+e^{-\int_{0}^{\tau_1}r_v dv}(K-X_{\tau_1})^+\right],
\end{align}
where we use that $r_u > c(\hat t+u, X_u)$ on $\{u < \tau_1\}$. Applying Tanaka's formula to $(r,x) \mapsto (K - x)^+$ and taking expectation, we get
\begin{align*}
(K - \hat{x})^+
&=\E_{\hat{r},\hat{x}}\left[\int_{0}^{\tau_1}e^{-\int_{0}^{u}r_v dv}Kr_u \mathds{1}_{\{X_u<K\}}du+e^{-\int_{0}^{\tau_1}r_v dv}(K - X_{\tau_1})^+ +\frac12 \int_0^{\tau_1}e^{-\int_{0}^{u}r_v dv}dL^{K}_u(X)  \right]\\
&=\E_{\hat{r},\hat{x}}\left[\int_{0}^{\tau_1}e^{-\int_{0}^{u}r_v dv}Kr_udu+e^{-\int_{0}^{\tau_1}r_v dv}(K - X_{\tau_1})^+\right],
\end{align*}
where $L^K(X)$ is the local time of the process $X$ at $K$. The local time $L^K(X)$ is null until $\tau_1$ since $r_u>\tc(t+u,X_u)\implies X_u<K$, recalling that $\tc(t,x)=\bar{r}$ when $x\ge K$. Compare the right-hand side of the above expression to \eqref{eqn:b1} to conclude that $\tilde v(\hat{t},\hat{r},\hat{x}) = (K - \hat{x})^+$.

{\em Step 2.} The next step is to show that $\tilde v\le v$ for $(t,r,x) \in \cO$. Since we have already proved $\tilde v(t,r,x)=(K-x)^+\le v(t,r,x)$ when $r\ge\tc(t,x)$, we take $(\hat{t},\hat{r},\hat{x}) \in \cO$ such that $\hat{r}<\tc(\hat{t},\hat{x})$. Define a stopping time
\[
\tau_2:=\inf\{s\ge0: r^{\hat{r}}_s\ge \tc(\hat{t}+s,X^{\hat{r},\hat{x}}_s)\}\wedge (T-\hat{t}).
\] 
Since $r_u < \tc(\hat t +u, X_u)$ on $\{ u < \tau_2\}$, we obtain from \eqref{eqn:dynkin_u} 
\begin{align*}
\tilde v(\hat{t},\hat{r},\hat{x})
&=
\E_{\hat{r},\hat{x}}\left[e^{-\int_{0}^{\tau_2}r_v dv}(K-X_{\tau_2})^+\right]
\le 
v(\hat{t},\hat{r},\hat{x}),
\end{align*}
where the first equality is by $\tilde v\big(\hat t + \tau_2, \tc(\hat t + \tau_2, X_{\tau_2}), X_{\tau_2}\big) = (K-X_{\tau_2})^+$ and the final inequality holds by the definition of $v$. 

{\em Step 3.} Now we show that $\tc(t,x)\le c(t,x)$ for any $(t,x)\in [0,T)\times(0,K)$ such that $\tilde c(t,x)\in\cI$ (it is immediate for $(t,x) \in [0, T) \times [K, \infty)$ as $\tilde c(t,x)=c(t,x)=\overline{r}$). Arguing by contradiction, assume that there exists $(\hat{t},\hat{x})\in [0,T)\times(0,K)$ such that $\cI\ni\tc(\hat t,\hat x)> c(\hat t,\hat x)$. Let $\hat{r}>
\tc(\hat{t},\hat{x})$, and define
\[
\tau_3:=\inf\{s\ge0: r^{\hat{r}}_s\le c(\hat{t}+s,X^{\hat{r},\hat{x}}_s)\}\wedge (T-\hat{t}).
\]
By \eqref{eq:int-1} and \eqref{eqn:dynkin_u}, we have
\begin{align*}
v(\hat{t},\hat{r},\hat{x})
&=
\E_{\hat{r},\hat{x}}\left[\int_{0}^{\tau_3}e^{-\int_{0}^{u}r_v dv}Kr_u \mathds{1}_{\{r_u>c(\hat t+u,X_u)\}}du+e^{-\int_{0}^{\tau_3}r_v dv}v(\hat t+\tau_3,r_{\tau_3},X_{\tau_3})\right],\\
\tilde v(\hat{t},\hat{r},\hat{x})
&=
\E_{\hat{r},\hat{x}}\left[\int_{0}^{\tau_3}e^{-\int_{0}^{u}r_v dv}Kr_u \mathds{1}_{\{r_u>\tc(\hat t+u,X_u)\}}du+e^{-\int_{0}^{\tau_3}r_v dv}\tilde v(\hat t+\tau_3,r_{\tau_3},X_{\tau_3})\right].
\end{align*}
Since $\tilde v(\hat{t},\hat{r},\hat{x})=(K - \hat{x})^+=v(\hat{t},\hat{r},\hat{x})$, $r_u > c(\hat t+u, X_u)$ on $\{u < \tau_3\}$, and $\tilde {v}\le v$, the above two equations imply that
\[
\E_{\hat{r},\hat{x}}\left[\int_{0}^{\tau_3}e^{-\int_{0}^{u}r_v dv}Kr_u \mathds{1}_{\{r_u>\tc(\hat t+u,X_u)\}}du\right]
\ge
\E_{\hat{r},\hat{x}}\left[\int_{0}^{\tau_3}e^{-\int_{0}^{u}r_v dv}Kr_udu\right].
\]
As the function $c$ is non-negative, $r_u \ge 0$ on $\{u < \tau_3\}$ and we conclude that
\begin{equation}\label{eqn:zer}
\E_{\hat{r},\hat{x}}\left[\int_{0}^{\tau_3}\mathds{1}_{\{r_u\le\tc(\hat t+u,X_u)\}}du\right]=0.
\end{equation}
The dynamics of $(r, X)$ is non-degenerate on $\cI \times \R_+$, so the density of $(r_u,X_u)$ has a full support (on $\cI \times \R_+$) for $u >0$ (this can be inferred by classical Gaussian bounds as those we use in \eqref{eq:Gbound} in Appendix). Hence, by the continuity of $\tc$ and $c$, for a sufficiently small $\eps > 0$,
\[
\P_{\hat{r},\hat{x}}\left(c(\hat t+u, X_u) < r_u < \tc(\hat{t}+u,X_u) \text{ for some } u \in (0, \eps)\right)>0.
\]
Paired with the continuity of trajectories of $(r,X)$, it contradicts \eqref{eqn:zer}.

{\em Step 4.} Next, we prove $\tc=c$ at all points such that $\tc \in\cI$. Arguing by contradiction, assume $\tc(\hat{t},\hat{x})< c(\hat{t},\hat{x})$ for some $(\hat t, \hat x) \in [0, T) \times (0, K)$ such that $\tc (\hat t,\hat x)\in\cI$. Let $\hat{r}\in(\tc(\hat{t},\hat{x}),c(\hat{t},\hat{x}))$ and define
\[
\tau_4:=\inf\{s\ge0: r^{\hat{r}}_s\ge c(\hat{t}+s,X^{\hat{r},\hat{x}}_s)\}\wedge (T-\hat{t}).
\]
By \eqref{eq:int-1} and \eqref{eqn:dynkin_u}, we have
\begin{align*}
v(\hat{t},\hat{r},\hat{x})
&=
\E_{\hat{r},\hat{x}}\left[ e^{-\int_{0}^{\tau_4}r_v dv}v(\hat t+\tau_4,r_{\tau_4},X_{\tau_4})\right],\\
\tilde v(\hat{t},\hat{r},\hat{x})
&=
\E_{\hat{r},\hat{x}}\left[\int_{0}^{\tau_4}e^{-\int_{0}^{u}r_v dv}Kr_u \mathds{1}_{\{r_u>\tc(\hat t+u,X_u)\}}du+e^{-\int_{0}^{\tau_4}r_v dv}\tilde v(\hat t+\tau_4,r_{\tau_4},X_{\tau_4})\right],
\end{align*}
where in the first expression we used that $\mathds{1}_{\{r_u>c(\hat t+u,X_u)\}} = 0$ on $\{u < \tau_4\}$.
Since $\tc(t,x)\le c(t,x)$ for $(t,x) \in [0, T) \times (0, K)$, we have $\tilde v(\hat t+\tau_4,r_{\tau_4},X_{\tau_4})=(K-X_{\tau_4})^+=v(\hat t+\tau_4,r_{\tau_4},X_{\tau_4})$ by step 1. Then recalling that $\tilde v \le v$ and comparing the two equations above give us
\[
\E_{\hat{r},\hat{x}}\left[\int_{0}^{\tau_4}e^{-\int_{0}^{u}r_v dv}Kr_u\mathds{1}_{\{r_u>\tc(\hat t+u,X_u)\}}du\right]\le0.
\]
This is a contradiction since by the continuity of $(r,X)$ and $\tc$ there is a random variable $\eta > 0$ such that
\[
r_u(\omega) >\tc(\hat t+u,X_u(\omega)) \quad \text{for all $u \in [0, \eta(\omega))$}.
\]

{\em Step 5.} Here we show that $\tc = c$ on $[0,T)\times\R_+$. Let $(t_n,x_n)$ be a sequence such that $\tilde c(t_n,x_n)\in\cI$ and 
$(t_n,x_n)\to(t_0,x_0)$ with $\tilde c(t_0,x_0)=\bar r$ (respectively $c(t_0,x_0)=\tilde c(t_0,x_0)=\underline r$, if $\underline r = 0$; note that $\tc \ge 0$, so $\tc > \underline{r}$ if $\underline{r} < 0$). Since $\tilde c(t_n,x_n)=c(t_n,x_n)$ for all $n$'s, by the four steps above, by continuity we also get $c(t_0,x_0)=\tilde c(t_0,x_0)=\bar r$ (respectively $c(t_0,x_0)=\tilde c(t_0,x_0)=\underline r$, if $\underline r = 0$). Then, by the monotonicity of both $c$ and $\tc$ we get $c(t,x)=\tc(t,x)$ for all $(t,x)\in[0,t_0]\times [x_0,\infty)$ (respectively $(t,x)\in[t_0,T]\times (0,x_0]$). This implies, in particular, that
\[
\{(t,x):\tc (t,x)\in\cI\}=\{(t,x):c (t,x)\in\cI\},
\] 
which concludes the proof.
\end{proof}

\begin{proof}[Proof of Corollary \ref{cor:int-eq}]
We can repeat the same arguments as in the proof of Proposition \ref{prop:int-eq-1}, always using $x>b(t,r)\iff r<c(t,x)$ to fall back into the exact set-up of steps 1--4 therein.
\end{proof}

\section{Hedging strategy}\label{sec:hedge}

We start this section with an auxiliary lemma whose assertions are used to show admissibility of the hedging strategy. Estimate \eqref{eqn:i1} is also used in the proof of Proposition \ref{prop:int-eq}.

\begin{lemma}
\label{lem:integr} For any compact set $\cK \subset \cI \times \R_+$, and $p \in [1,2]$, we have
\begin{align}
\label{eqn:i1}
&\sup_{(r,x) \in \cK} \sup_{t \in [0, T]} \E_{r,x} \Big[ \sup_{0 \le s \le T-t} e^{-\int_0^s r_u du} v(t+s, r_s, X_s) \Big]
< \infty,\\
\label{eqn:i3}
&\sup_{(r,x) \in \cK} \sup_{s \in [0, T]} \E_{r,x} \Big[e^{-p\int_0^s r_u du} \big|v_x(t+s, r_s, X_s)\big|^p X_s^p \Big] < \infty,\\
\label{eqn:i2}
&\sup_{(r,x) \in \cK} \sup_{s \in [0, T]} \E_{r,x} \Big[ e^{-2\int_0^s r_u du} \big(v_r(t+s, r_s, X_s)\big)^2\beta^2(r_s) \Big] < \infty.
\end{align}
\end{lemma}

\begin{proof}
From \eqref{eq:v} we obtain an upper bound for the function $v$:
\begin{equation}\label{eqn:ip1}
v(t,r,x) \le K \E_r \Big[\sup_{0 \le s \le T-t} e^{-\int_0^s r_u du} \Big].
\end{equation}
Using this bound, we have $P_{r,x}$-a.s.
\begin{align*}
e^{-\int_0^s r_u du} v(t+s, r_s, X_s) 
&\le e^{-\int_0^s r_u du} K \E_{r_s}\Big[\sup_{0 \le u \le T-t-s} e^{-\int_0^u r_v dv} \Big]\\
&= e^{-\int_0^s r_u du} K \E_{r}\Big[\sup_{s \le u \le T-t} e^{-\int_s^u r_v dv} \Big| \cF_s \Big]\\
&\le K \E_{r}\Big[\sup_{0 \le u \le T-t} e^{-\int_0^u r_v dv} \Big| \cF_s \Big]
\le K \E_{r}\Big[\sup_{0 \le u \le T} e^{-\int_0^u r_v dv} \Big| \cF_s \Big],
\end{align*}
where in the second equality we employ the Markov property of $r$. By Doob's maximal inequality applied to the martingale $Y_s = \E_{r}\Big[\sup_{0 \le u \le T} e^{-\int_0^u r_v dv} \Big| \cF_s \Big]$, we conclude
\begin{align*}
&\sup_{(r,x) \in \cK} \sup_{t \in [0, T]} \E_{r,x} \Big[ \sup_{0 \le s \le T-t} e^{-\int_0^s r_u du} v(t+s, r_s, X_s) \Big]\\
&\le
\sup_{(r,x) \in \cK} \sup_{t \in [0, T]} K\, \E_{r} \Big[ \sup_{0 \le s \le T-t} Y_s \Big]
\le
\sup_{(r,x) \in \cK} K\, \E_{r} \Big[ \sup_{0 \le s \le T} Y_s \Big]\\
&\le
\sup_{(r,x) \in \cK} 2K\, \big(\E_{r} [ Y_T^2 ] \big)^{1/2}
=
\sup_{(r,x) \in \cK} 2K \, \Big(\E_r \Big[\sup_{0 \le u \le T} e^{-2\int_0^u r_v dv}\Big]\Big)^{1/2} \le 2 K (C_1)^{1/2},
\end{align*}
where $C_1$ is the constant from \eqref{eq:integr}. This proves (i).

We now address \eqref{eqn:i3}. We have
\begin{align*}
e^{-p\int_{0}^{s}r_udu}\big|v_x(t+s,r_s,X_s)\big|^p X_s^p
= \big|v_x(t+s,r_s,X_s)\big|^p x^p e^{p\sigma B_s - \frac{p}{2}\sigma^2 s}
\le x^p e^{p\sigma B_s - \frac{p}{2}\sigma^2 s},
\end{align*}
where we use $-1 \le v_x \le 0$ in the last inequality, which follows from \eqref{eq:vx}. From here, \eqref{eqn:i3} is immediate.

It remains to prove \eqref{eqn:i2}. First we consider the case of Assumption \ref{ass:coef}(ii). From \eqref{eqn:v_r1}, \eqref{eqn:r_bnd} and \eqref{eqn:v_r2}, we deduce
\begin{equation}\label{eqn:vra}
\big(v_r(t,r,x)\big)^2 \le c_1\, \E_r \Big[ \sup_{0\le s \le T-t} e^{-2\int_{0}^{s}r_u du} \Big]
\end{equation}
for some constant $c_1> 0$ depending only on $T$ and the coefficients of \eqref{eq:r} (notice in particular that the expected value in the right-hand side above comes from the constant $C_1$ in \eqref{eqn:r_bnd}). Hence
\begin{align*}
e^{-2\int_0^s r_u du} \big(v_r(t+s, r_s, X_s)\big)^2 \beta^2(r_s)
&\le
e^{-2\int_0^s r_u du} c_1\, \E_{r_s} \Big[ \sup_{0\le u \le T-t-s} e^{-2\int_{0}^{u}r_v dv} \Big] \beta^2(r_s)\\
&\le
c_1\, \E_{r} \Big[ \sup_{0\le u \le T-t} e^{-2\int_{0}^{u}r_v dv} \Big| \cF_s \Big] \beta^2(r_s),
\end{align*}
where the last inequality is by the same argument as in the proof of \eqref{eqn:i1}. We take expectation of both sides and apply H\"{o}lder inequality with $q = p'/2$ ($p'>2$ is defined in Assumption \ref{ass:coef}) and $q' = q/(q-1)$
\begin{align*}
\E_r\Big[ e^{-2\int_0^s r_u du} \big(v_r(t+s, r_s, X_s)\big)^2 \beta^2(r_s) \Big]
&\le
c_1\, \Big( \E_r\Big[ \E_{r} \Big[ \sup_{0\le u \le T-t} e^{-2\int_{0}^{u}r_v dv} \Big| \cF_s \Big]^q \Big] \Big)^{1/q} 
\big( \E_r [\beta^{2q'}(r_s)] \big)^{1/q'}\\
&\le
c_1\, \Big( \E_r\Big[ \sup_{0\le u \le T-t} e^{-p'\int_{0}^{u}r_v dv}  \Big] \Big)^{1/q} 
\big( \E_r [\beta^{2q'}(r_s)] \big)^{1/q'}\\
&\le
c_1\, C_1^{1/q} \big( \E_r [\beta^{2q'}(r_s)] \big)^{1/q'},
\end{align*}
where the second inequality follows from Jensen's inequality and $C_1$ is the constant from \eqref{eq:integr}.
Let $L$ be the Lipschitz constant for $\beta$. Then, using triangle inequality for norms,
\begin{align*}
\Big(\E_r [(\beta(r_s))^{2q'}] \Big)^{1/q'} 
&\le 
\Big( \E_r \big[ \big|\beta(0) + L |r_s|\big|^{2q'} \big] \Big)^{1/q'}
=
\Big(\Big( \E_r \big[ \big|\beta(0) + L |r_s|\big|^{2q'} \big] \Big)^{1/2q'} \Big)^{2}\\
&\le
\Big(\beta(0) + L \big( \E_r \big[ |r_s|^{2q'} \big] \big)^{1/2q'} \Big)^{2}
\le
\Big(\beta(0) + L \big( C_2(1+|r|^{2q'} \big)^{1/2q'} \Big)^{2},
\end{align*}
where the last inequality follows from \eqref{eq:subg} and $2q'=p' \ge 2$. Combining the above estimates proves \eqref{eqn:i2}.

We address the case when $r$ follows the CIR dynamics. From the non-negativity of the process $r$ and from \eqref{eqn:vra} we obtain that $\big(v_r(t,r,x)\big)^2 \le c_1$ for any $(t,r,x) \in \cO$. Hence, we write
\[
\E_{r,x} \Big[ e^{-2\int_0^s r_u du} \big(v_r(t+s, r_s, X_s)\big)^2 \beta^2(r_s) \Big]
\le
c_1 \gamma^2 \E_{r} [|r_s|],
\]
where we used the explicit form of $\beta$. It remains to recall \eqref{eq:subg} to conclude \eqref{eqn:i2}.
\end{proof}

\begin{proof}[Proof of Proposition \ref{prop:hedge}]
The admissibility condition can be equivalently written as
\begin{equation}\label{eqn:h_3}
\int_0^{T} e^{-2\int_{0}^{s}r_udu}\big(\phi^{(1)}_s\sigma X_s\big)^2 ds +
\int_0^{T} e^{-2\int_{0}^{s}r_udu}\big(\phi^{(2)}_s\beta(r_s)P_r(s,r_s) \big)^2 ds < \infty, \qquad
\text{$\P_{r,x}$-a.s.}
\end{equation}
Estimates in Lemma \ref{lem:integr} imply
\[
\E_{r,x} \Big[\int_0^{T} e^{-2\int_{0}^{s}r_udu}\big(\phi^{(1)}_s\sigma X_s\big)^2 ds +
\int_0^{T} e^{-2\int_{0}^{s}r_udu}\big(\phi^{(2)}_s\beta(r_s)P_r(s,r_s) \big)^2 ds \Big] < \infty,
\]
which is a stronger condition than \eqref{eqn:h_3}. The fact that the portfolio replicates the option follows from the construction and \eqref{eqn:hedge}.
\end{proof}

\vspace{+15pt}
\noindent{\bf Acknowledgment}: C.~Cai gratefully acknowledges support by China Scholarship Council. T.~De Angelis gratefully acknowledges support by EPSRC grant EP/R021201/1.

\newpage

\begin{center}
\Large \sc Appendices
\end{center}

\appendix

\section{Convergence of stopping times}\label{app:cont}
In this appendix we provide a self-contained exposition around the continuity of hitting/entry times to certain domains in $\cO$. Similar results can be found in many textbooks, e.g., \cite[p. 32-40]{Dynkin2}, \cite[Chapter 1]{Blumenthal}.

Throughout this section, $\cDp$ denotes a subset of $\cO$ that is closed in $\cO$, i.e., $\overline \cDp \cap \cO = \cDp$. Introduce the {\em hitting} time to $\cDp$, denoted $\sigma_\cDp$, and the {\em entry} time to the interior of $\cDp$, denoted $\ss_\cDp$. That is, for $(t,r,x)\in [0, T] \times \cI \times \R_+$, 
we set $\P_{r,x}$-a.s.
\begin{equation}\label{eqn:sigma_cDp}
\begin{aligned}
&\sigma_\cDp:=\inf\{s>0\,:\,(t+s,r_s,X_s)\in\cDp\}\wedge(T-t),\\
&\ss_\cDp:=\inf\{s\ge 0\,:\,(t+s,r_s,X_s)\in\interior(\cDp)\}\wedge(T-t).
\end{aligned}
\end{equation}
The above definition differs from \eqref{eqn:sigmacD} since $\{T\} \times \cI \times \R_+ \subset \cD$ guaranteeing that $\sigma_\cD \le T-t$.
Both $\sigma_\cDp$ and $\ss_\cDp$ are stopping times with respect to the augmentation of the filtration generated by the Brownian motions $(B,W)$.
It is immediate to see that
\begin{align}\label{eq:ins}
\sigma_\cDp\le \ss_\cDp,\qquad\text{$\P$-a.s.}
\end{align}
We will often write $\sigma_\cDp(t,r,x)$ and $\ss_\cDp(t,r,x)$ to indicate the starting point of the process. 

Denote by $\cCp$ the complement of $\cDp$ in $\cO$: $\cCp := \cDp^c \cap \cO$ (which is an open set). When we write $\partial \cCp$ we mean the boundary of $\cCp$ in $[0, T] \times \cI \times \R_+$. We have two assumptions
\begin{assumption}[Regularity]
\label{ass:reg} 
For $(t_0,r_0,x_0)\in\partial\cCp$, we have
\[
\P_{t_0,r_0,x_0}(\sigma_\cDp>0)=\P_{t_0,r_0,x_0}(\ss_\cDp>0)=0.
\]
\end{assumption}
\begin{assumption}
\label{ass:cont}
For any sequence $(r_n,x_n)_{n\ge1}$ converging to $(r,x)\in\cI\times\R_+$ as $n\to\infty$, it holds that
\begin{equation}\label{eqn:cont}
\lim_{n\to +\infty}\sup_{0\le t\le T}\big( \left|r^{r_n}_t-r^{r}_t\right| + \left|X^{r_n,x_n}_t-X^{r,x}_t\right| \big)=0, \qquad\text{$\P$-a.s.}
\end{equation}
\end{assumption}
Note that Assumption \ref{ass:reg} necessitates that $\cDp = \overline{\interior(\cDp)} \cap \cO$. It also enables establishing a connection between $\sigma_\cDp$ and $\ss_\cDp$.

\begin{lemma}\label{lem:sigmaD} Under Assumption \ref{ass:reg},
$\P_{t,r,x}(\sigma_\cDp= \ss_\cDp)=1$ for all $(t,r,x)\in\cO$.
\end{lemma}
\begin{proof}
The equality is trivial for $(t,r,x) \in \interior(\cDp)$. Take $(t,r,x)$ in its complement, i.e., in $\overline \cCp \cap \cO$. Since \eqref{eq:ins} holds, we only need to show that $\P_{t,r,x}(\sigma_\cDp< \ss_\cDp)=0$. Let us argue by contradiction and assume that $\P_{t,r,x}(\sigma_\cDp< \ss_\cDp)>0$. There exists $\delta>0$ such that $\P_{t,r,x}(\ss_\cDp\ge \sigma_\cDp+\delta)>0$. Denoting by $\theta_t$ the shift operator, by the strong Markov property we get
\begin{align*}
\P_{t,r,x}(\ss_\cDp\ge \sigma_\cDp+\delta)
=&\,\E_{t,r,x}\left[\E_{t,r,x}\left(\mathds{1}_{\{\sigma_\cDp+\ss_\cDp\circ\theta_{\sigma_\cDp}\ge \sigma_\cDp+\delta\}}\Big|\cF_{\sigma_\cDp}\right)
\right]\\
=&\,\E_{t,r,x}\left[\P_{t+\sigma_\cDp,r_{\sigma_\cDp},X_{\sigma_\cDp}}\left(\ss_\cDp\ge\delta\right)\right]=0,
\end{align*}
where the last equality follows by observing that $(t+\sigma_\cDp,r_{\sigma_\cDp},X_{\sigma_\cDp})\in\partial\cCp$, $\P_{t,r,x}$-a.s.~and Assumption \ref{ass:reg}.
\end{proof}

\begin{lemma}\label{lem:usc}
Let $\cO\ni(t_n,r_n,x_n)\to(t,r,x)\in\cO$ as $n\to\infty$. Under Assumption \ref{ass:cont}, it holds $\P$-a.s.
\[
\limsup_{n\to\infty}\ss_\cDp(t_n,r_n,x_n)\le \ss_\cDp(t,r,x).
\]
\end{lemma}
\begin{proof}
For simplicity, we denote $\ss_n:=\ss_\cDp(t_n,r_n,x_n)$ and $\ss_\cDp:=\ss_\cDp(t,r,x)$. For $\P$-a.e.~$\omega \in \Omega$ we have by \eqref{eqn:cont}
\begin{equation}\label{eqn:usc1}
(t_n+s,r^{r_n}_{s},X^{r_n,x_n}_{s})(\omega) \to(t+s,r^{r}_{s},X^{r,x}_{s})(\omega), \qquad s \in [0, T-t].
\end{equation}
Fix $\omega\in\Omega$ in the set of $\P$-full measure for which the above holds. If $\ss_\cDp(\omega)=T-t$ then the result is obvious because $\ss_n(\omega)\le T-t_n$. Assume $\ss_\cDp(\omega)<T-t$. Take any $\delta < T-t$ such that $\ss_\cDp(\omega) < \delta$. By the continuity of paths and the openness of $\text{int}(\cDp)$, there is $\delta' \in (\ss_\cDp(\omega), \delta)$ such that $(t+\delta',r^{r}_{\delta'},X^{r,x}_{\delta'})(\omega)\in\text{int}(\cDp)$. From \eqref{eqn:usc1} and the openness of $\text{int}(\cDp)$, $(t_n+\delta',r^{r_n}_{\delta'},X^{r_n,x_n}_{\delta'})(\omega)\in\text{int}(\cDp)$ for all sufficiently large $n$, so $\limsup_{n\to\infty}\ss_n(\omega)\le \delta'$. As the above argument holds for any $\delta > \ss_{\cDp}(\omega)$ and for a.e.~$\omega \in \Omega$, we obtain the claim.
\end{proof}

\begin{lemma}\label{lem:rsigma} 
Let $(t_n,r_n,x_n)_{n\ge 1}$ be a sequence converging to $(t,r,x)\in\cO$ as $n\to \infty$. Then, under Assumptions \ref{ass:reg} and \ref{ass:cont},
\begin{align}\label{eq:lsc}
\liminf_{n\to\infty}\sigma_{\cDp}(t_n,r_n,x_n)\ge \sigma_\cDp(t,r,x),\qquad\text{$\P$-a.s.}
\end{align}
\end{lemma}
\begin{proof}
For $y,z\in\cO$ we denote by $\d(y,z)$ their Euclidean distance and by $\d(y,\partial\cCp)=\inf\{\d(y,z)\,,\,z\in\partial\cCp\}$. Denote $\sigma_n:=\sigma_\cDp(t_n,r_n,x_n)$ and $\sigma_\cDp:=\sigma_\cDp(t,r,x)$. To simplify notation we also set
\[
\zeta_s:=(t+s,r^{r}_s,X^{r,x}_s)\quad\text{and}\quad\zeta^n_s:=(t_n+s,r^{r_n}_s,X^{r_n,x_n}_s).
\]

Fix $\omega \in \Omega$ in a set of full $\P$-measure on which trajectories of $\zeta$ are continuous and the limit \eqref{eqn:cont} holds. If $\sigma_{\cDp}(\omega)=0$ then \eqref{eq:lsc} holds trivially. Otherwise, we must have $(t, r, x) \in \cCp$ by Assumption \eqref{ass:reg}, so $\d(\zeta_0(\omega), \partial \cCp) > 0$ since $\cCp$ is open. For $0 < \eps < \d(\zeta_0(\omega), \partial \cCp)$, define
\[
\delta_\eps = \inf\{ s \in [0, T-t]\,:\, \d(\zeta_s(\omega), \partial\cCp) \le \eps \}.
\]
Using the triangle inequality we get $\d(\zeta^n_s(\omega),\partial\cCp))+\d(\zeta^n_s(\omega),\zeta_s(\omega))> \eps$ for all $s\in[0,\delta_\eps]$. Thanks to Assumption \ref{ass:cont}, $\d(\zeta^n_s(\omega),\zeta_s(\omega))<\eps/2$ for all $s\in[0,\delta_\eps]$ and all sufficiently large $n$, so $\d(\zeta^n_s(\omega),\partial\cCp))\ge \eps/2$ for all $s\in[0,\delta_\eps]$ and all sufficiently large $n$. This gives 
\[
\liminf_{n\to\infty}\sigma_{n}(\omega)\ge \delta_\eps. 
\]
Using the continuity of $t \mapsto \zeta_t(\omega)$ and the fact that $\zeta_{\sigma_\cDp}(\omega) \in \partial \cCp$, we have $\lim_{\eps \to 0} \delta_\eps = \sigma_\cDp (\omega)$. As this holds for a.e. $\omega \in \Omega$, the proof of \eqref{eq:lsc} is complete.
\end{proof}

Lemma \ref{lem:sigmaD}, \ref{lem:usc} and \ref{lem:rsigma} imply
\begin{proposition}\label{prop:rsigma}
Let $(t_n,r_n,x_n)_{n\ge 1}$ be a sequence converging to $(t,r,x)\in\cO$ as $n\to \infty$. Then
\begin{equation}
\label{eq:limss}
\begin{aligned}
\lim_{n\to\infty}\ss_{\cDp}(t_n,r_n,x_n) &= \ss_\cDp(t,r,x),\qquad\text{$\P$-a.s.}\\
\lim_{n\to\infty}\sigma_{\cDp}(t_n,r_n,x_n) &= \sigma_\cDp(t,r,x),\qquad\text{$\P$-a.s.}
\end{aligned}
\end{equation}
\end{proposition}

\section{Proof of Proposition \ref{prop:reg}}\label{app:reg}

Fix $(t_0,r_0,x_0)\in\partial\cC$ and define $\cR:=[r_0,\overline r]\times[0,x_0]$, where we also recall that $\cI=(\underline r,\overline r)$. Since $t\mapsto c(t,x)$ is non-increasing, it is immediate to see that $[t_0,T]\times\cR\subseteq \cD$. Recalling the notation introduced in \eqref{eqn:sigma_cDp}, set
\[
\ss_\cR(r_0,x_0):=\inf\{s\ge0\,:\,(r^{r_0}_s,X^{r_0,x_0}_s)\in\text{int}(\cR)\}.
\]
We have $\ss_\cR(r_0,x_0)\ge\ss_\cD(t_0,r_0,x_0)$, $\P$-a.s., and $\ss_\cR(r_0,x_0)\ge\sigma_\cD(t_0,r_0,x_0)$ by the continuity of the process $(r, X)$. From now on we omit in the notation the dependence on $(t_0,r_0,x_0)$ since the initial point is fixed throughout the proof.

Take a compact ball $\cK\subset\cI \times \R_+$ centred at $(r_0,x_0)$. Let $\Sigma(r,x)$ denote the matrix of the diffusion coefficient for \eqref{eq:X}--\eqref{eq:r}, i.e.
\begin{align*}
\Sigma(r,x):=\frac{1}{2}
\left(
\begin{array}{cc}
\sigma^2x^2 & \rho\sigma x\beta(r)\\
\rho\sigma x\beta(r) & \beta^2(r)
\end{array}
\right).
\end{align*} 
Since the correlation coefficient $\rho\in(-1,1)$, there is $\gamma = \gamma_{\cK}>0$ such that
\begin{align}\label{eq:ue}
\frac{1}{\gamma}\|z\|^2\le\langle \Sigma(r,x)z,z\rangle \le\gamma\|z\|^2, \qquad z\in\R^2 \setminus\{0\},\ (r,x) \in \cK,
\end{align}
where $\langle\cdot,\cdot\rangle$ denotes the scalar product in $\R^2$ and $\|\cdot\|$ the corresponding norm.

Define a new process $(\widetilde r,\widetilde X)$ with the dynamics defined on $\R^2$
\begin{align}
\label{eq:tX} &d\, \widetilde X_t=\mu_\cK(\widetilde r_t, \widetilde X_t) dt+\sigma_\cK( \widetilde X_t ) dB_t,\qquad \widetilde X_0=x_0\\[+4pt]
\label{eq:tr} &d\, \widetilde r_t=\alpha_\cK(\widetilde r_t)dt+\beta_\cK(\widetilde r_t)dW_t,\qquad\qquad \widetilde r_0=r_0
\end{align}
such that the coefficients coincide with the coefficients of \eqref{eq:X}-\eqref{eq:r} on $\cK$, are Lipschitz continuous on $\R^2$ and satisfy the uniform ellipticity condition \eqref{eq:ue} with $\gamma$ on $\R^2$. Denoting $\tau_\cK:=\inf\{t\ge0\,:\,(r_t,X_t)\notin\text{int}(\cK)\}$ and $\widetilde \tau_\cK:=\inf\{t\ge0\,:\,(\widetilde r_t,\widetilde X_t)\notin\text{int}(\cK)\}$, by the uniqueness of solutions for SDEs we get indistinguishable stopped paths:
\[
(r_{t\wedge\tau_\cK},X_{t\wedge\tau_\cK})_{t\ge0}=(\widetilde r_{t\wedge\widetilde\tau_\cK},\widetilde X_{t\wedge\widetilde\tau_\cK})_{t\ge0}\quad\text{$\P_{r_0,x_0}$-a.s.}
\]

The uniform ellipticity condition \eqref{eq:ue} on $\R^2$ implies that the process $(\widetilde r,\widetilde X)$ admits a transition density $\widetilde p(t,(r,x),(r',x'))$ which satisfies the following Gaussian bound (see, e.g., \cite{aronson1967bounds,fabes1989new}): there exists $m>0$ and $\Lambda>0$ such that
\begin{equation}\label{eq:Gbound}
\widetilde p(t,(r,x),(r',x')) \ge m\,t^{-1}\exp\left(-\Lambda\frac{(r'-r)^2+(x'-x)^2}{t}\right).
\end{equation}

Let $\cR''$ be a closed cone with vertex $(r_0, x_0)$ and non-empty interior contained in $(r_0, \infty) \times (-\infty, x_0) \cup \{(r_0, x_0)\}$. Put $\cR' = \cR'' \cap (\cI \times \R_+)$. Denote by $\ss_{\cR}'$ the entry time of $(r,X)$ to $\interior(\cR')$ and by $\ss_{\cR}''$ the entry time of $(\widetilde r,\widetilde X)$ to $\interior(\cR'')$. The next estimate relies on analogous results for multi-dimensional Brownian motion \citep[Thm.~4.2.9]{karatzas1998brownian}; in particular the second inequality below follows from \eqref{eq:Gbound}:
\begin{align*}
\P_{r_0,x_0}(\ss_\cR''\le t)& \ge \P_{r_0, x_0} ((\widetilde r_t, \widetilde X_t) \in \cR'') 
\ge \frac{m}{t}\int_{\cR''}\exp\left(-\Lambda\frac{(r-r_0)^2+(x-x_0)^2}{t}\right)dr\,dx.
\end{align*} 
We change variables to $y:=(r-r_0)/\sqrt{t}$ and $z:=(x-x_0)/\sqrt{t}$ and use that this transformation maps $\cR''$ into $\cR''_0 = \cR'' - (r_0, x_0)$ to obtain 
\begin{align}\label{eq:reg2}
\P_{r_0,x_0}(\ss_\cR''\le t)\ge 2\pi m \int_{\cR''_0}\frac{1}{2\pi}e^{-\Lambda(y^2+z^2)}dy\,dz=:q>0.
\end{align} 
For any $t>0$, since $\cR'\subset \cR$ we have
\begin{align*}
\P_{r_0,x_0}(\ss_\cR\le t)\ge&\, \P_{r_0,x_0}(\ss'_\cR\le t)\ge\P_{r_0,x_0}(\ss_\cR'\le t,\tau_\cK>t)\nonumber\\
=&\,\P_{r_0,x_0}(\ss_\cR''\le t,\widetilde\tau_\cK>t)\ge\P_{r_0,x_0}(\ss_\cR''\le t)-\P_{r_0,x_0}(\widetilde\tau_\cK\le t)
\ge q - \P_{r_0,x_0}(\widetilde\tau_\cK\le t),
\end{align*}
where the last inequality is by \eqref{eq:reg2}. As $t\downarrow 0$, we have $\P_{r_0,x_0}(\widetilde\tau_\cK\le t)\to0$ and $\P_{r_0,x_0}(\ss_\cR\le t)\to \P_{r_0,x_0}(\ss_\cR=0)$, which implies that $\P_{r_0,x_0}(\ss_\cR=0)\ge q>0$. By the Blumenthal $0-1$ law \citep[Thm.~2.7.17]{karatzas1998brownian} we obtain $\P_{r_0,x_0}(\ss_\cR=0)=1$. Recalling that $\P_{t_0, r_0, x_0} (\ss_\cR \ge\ss_\cD) = 1$, and $\P_{t_0, r_0, x_0} (\ss_\cD\ge\sigma_\cD) = 1$, we conclude $\P_{t_0, r_0, x_0} (\ss_\cD = 0) = \P_{t_0, r_0, x_0} (\sigma_\cD=0) = 1$.\hfill$\square$\\

\begin{remark}\label{rem:lambda}
It is worth noticing that the arguments above show the existence of the transition density of the process $(\widetilde r,\widetilde X)$ for any compact set $\cK \subset \cI \times \R_+$ such that $\cK = \overline{\interior(\cK)}$. This implies that for each $t\in[0,T]$ also the law of $(r_t,X_t)$ is absolutely continuous with respect to the Lebesgue measure on $\cI\times\R_+$, when the boundary of $\cI\times\R_+$ is unattainable by $(r_t,X_t)$. Indeed, let $N\subset\cI\times\R_+$ be such that $\lambda(N)=0$, with $\lambda$ denoting the Lebesgue measure on $\R^2$. Let $\cK\subset\cI\times\R_+$ be a compact set such that $\cK = \overline{\interior(\cK)}$. Then by the same construction as above
\begin{equation}\label{eq:absc}
\begin{aligned}
\P_{r_0,x_0}\big((r_t,X_t)\in N\big)
&=
\P_{r_0,x_0}\big((r_t,X_t)\in N,t\le \tau_\cK\big)+\P_{r_0,x_0}\big((r_t,X_t)\in N,t> \tau_\cK\big)\\
&=
\P_{r_0,x_0}\big((\widetilde r_t,\widetilde X_t)\in N,t\le \widetilde\tau_\cK\big)+\P_{r_0,x_0}\big((r_t,X_t)\in N,t> \tau_\cK\big)\\
&\le
\P_{r_0,x_0}\big((\widetilde r_t,\widetilde X_t)\in N\big)+\P_{r_0,x_0}\big(\tau_\cK< t\big)
=\P_{r_0,x_0}\big(\tau_\cK< t\big),
\end{aligned}
\end{equation}
where the final equality uses that the transition law of $(\widetilde r,\widetilde X)$ is absolutely continuous with respect to $\lambda$. Now, letting $\cK\uparrow \cI\times\R_+$, using that $0$ and $+\infty$ are not attainable by $X$ and $\underline r$ and $\overline r$ are not attainable by $r$, we can make $\P_{r_0,x_0}\big(\tau_\cK < t\big)$ arbitrarily small, which proves the claim.
\end{remark}

\section{Generalisation of It\^o's formula}\label{app:Ito}
Here we state a particular case of \cite[Theorem 2.1]{cai2021b} using notations appropriate for our set-up. We set $\bm x=(x_1,x_2)\in \R^2$. Let $B=(B^j)_{j=1,2}$ be a two-dimensional standard Brownian motion and denote $\bm{X}:=(X^{1},X^{2})$ the unique strong solution of
\begin{equation*}
\d X^{i}_t=\alpha^{i}(\bm{X}_{t})\d t+\sum_{j=1}^{2}\sigma^{ij}(\bm{X}_{t})\d B^{j}_t, \quad X^{i}_0=x_i,\quad i=1,2.
\end{equation*}
Let
$\beta^{ij}(\bm{x}):=\sum_{k=1}^{2}\sigma^{ik}(\bm{x})\sigma^{kj}(\bm{x})$
and denote $\cE:=\{(t,\bm x)\in[0,T]\times\R^2: x_1< g(t,x_2)\}$
for some function $g:[0,T]\times \R$. Here $\cE^c:=([0,T]\times\R^2)\setminus \cE$. Then the theorem reads as follows:
\begin{theorem}
Assume the following:
\begin{enumerate}
\item
The coefficients $\beta^{ij}$ are locally Lipschitz and $\P((t,\bm{X}_{t})\in\partial \cE)=0$ for a.e.\ $t\ge 0$; 
\item
A function $U:\R_+\times\R^2\to \R$ is such that $U\in C^1([0,T]\times\R^2)$ with $U\in C^{1,2}(\cE)\cap C^{1,2}(\cE^c)$. Moreover, for any compact subset $K\subset [0,T]\times\R^2$ the function 
\begin{align*}
L(t,\bm x):=\sum_{i,j=1}^2 \beta^{ij}(\bm{x})U_{x_i x_j}(t,\bm{x})
\end{align*} 
is bounded for $(t,\bm{x})\in K\setminus \partial\cE$. That is, there exists $c_K$ such that 
\item
The mappings $x_2\mapsto g(t,x_2)$ and $t\mapsto g(t,x_2)$ are monotonic. 
\end{enumerate}

Then, we have the change of variable formula: 
\begin{align*}
&U(t,\bm{X}_{t})=U(0,\bm{x})\\
&\quad+\int_0^{t} \Big[\Big(U_t+\sum_{i=1}^m \alpha^i U_{x_i}\Big)(u,\bm{X}_{u})+\tfrac{1}{2}\sum_{i,j=1}^m\mathds{1}_{\{(u,\bm X_{u})\notin\partial\cE\}}\big(\beta^{ij}U_{x_i x_j}\big)(u,\bm{X}_{u})\Big]\d u \\
&\quad+\sum_{i,j=1}^{m}\int_{0}^{t}U_{x_i}(u,\bm{X}_{u})\sigma^{ij}(\bm{X}_{u}) \d B^{j}_u,\qquad\text{for $t\in [0,\infty)$, $\P$-a.s.}
\end{align*}
\end{theorem}

\section{Derivation of \eqref{eq:numer-2}-\eqref{eq:numer-3}}\label{app:numerics}

We start by introducing the following notation (we suppress dependence on $r$ and $x$ for the sake of simplicity):
\allowdisplaybreaks
\begin{align}
\label{eqn:explicit_pricing}
d_1(t,T):&=\frac{\log\left(K\frac{P(t,T)}{x}\right)+\frac{\gamma_1(t,T)}{2}}{\sqrt{\gamma_1(t,T)}},\qquad d_2(t,T):=d_1(t,T)-\sqrt{\gamma_1(t,T)},\\
q(t,u):&=e^{-(u-t)\kappa}r+\theta(1-e^{-(u-t)\kappa})-\frac{\beta^2}{2}g(\kappa,u-t)^2,\notag\\
\phi(t,u,y; b):&=\frac{\log\left(\frac{P(t,u)}{x}b\big(u,q(t,u)+y \sqrt{\gamma_2(t,u)}\big)\right)+\frac{\gamma_1(t,u)}{2}-\sqrt{\gamma_1(t,u)}\tilde{\rho}(t,u)y}{\sqrt{(1-\tilde{\rho}(t,u)^2)\gamma_1(t,u)}},\notag\\
\mu(t,u):&=rg(\kappa,u-t)+\theta\big(u-t-g(\kappa,u-t)\big),\notag\\
\gamma_1(t,u):&=(u-t)\sigma^2+\frac{2\rho\sigma\beta}{\kappa}\big(u-t-g(\kappa,u-t)\big)+\frac{\beta^2}{\kappa}\big(u-t-2g(\kappa,u-t)+g(2\kappa,u-t)\big),\notag\\
\gamma_2(t,u):&=\beta^2g(2\kappa,u-t),\notag\\
\tilde{\rho}(t,u):&=\frac{\rho\sigma\beta g(\kappa,u-t)+\frac{\beta^2}{2}g(\kappa,u-t)^2}{\sqrt{\gamma_1(t,u)\gamma_2(t,u)}},\notag\\
g(a,u):&=\frac{1-e^{-au}}{a}.\notag
\end{align}

\begin{lemma}\label{lem:aux}
For a measurable bounded function $\varphi:(0, \infty) \times \R \to \R$ and $s\ge t$, the function
\[
u(t, s, r, x) = \E_{r,x} \left[e^{-\int_0^{s-t}r_u du}\varphi(X_{s-t},r_{s-t})\right]
\]
has an explicit representation
\begin{multline}
\label{eq:numercom-6}
u(t, s, r, x)
=e^{-\mu(t,s)+\frac{1}{2}\beta^2\int_{t}^{s}g(\kappa,s-u)^2du}\int_{\R^2}\varphi\Big(x e^{L(t,s)+\sqrt{\gamma_1(t,s))}z},\ q(t,s)+y \sqrt{\gamma_2(t,s)}\Big)\\
e^{-\frac{1}{2(1-\tilde{\rho}(t,s)^2)}\left(z^2+ y^2-2\tilde{\rho}(t,s)zy\right)}\frac{1}{2\pi\sqrt{1-\tilde{\rho}(t,s)^2}}dzdy,
\end{multline}
where
\[
L(t,s)=\mu(t,s)-\frac{\sigma^2}{2}(s-t)-\int_t^{s}\left(\beta^2g(\kappa,s-u)^2+\rho\sigma
g(\kappa,s-u)\right)du.
\]
\end{lemma}

We defer the proof until the end of this section and apply the above lemma to derive formulae for $P(t, T)$, $v_e$ and $v_p$. Taking $\varphi\equiv 1$ in \eqref{eq:numercom-6}, we have
\begin{equation}
\label{eq:numercom-7}
P(t, s) 
= \E_r\left[e^{-\int_0^{s-t}r_sds}\right]
= e^{-\mu(t,s)+\frac{1}{2}\beta^2\int_{t}^{s}g(\kappa,s-u)^2du}
= e^{-\mu(t,s)+\frac{\beta^2}{2\kappa^2}\left(s-t-2g(\kappa,s-t)+g(2\kappa,s-t)\right)}.
\end{equation}
This also implies that
\begin{equation*}
\begin{aligned}
e^{L(t,s)}
&=e^{\mu(t,s)-\frac{1}{2}\int_t^{s}\beta^2g(\kappa,s-u)^2du}\ e^{-\frac{\sigma^2}{2}(s-t)-\frac{1}{2}\int_t^{s}\beta^2g(\kappa,s-u)^2du-\int_t^{s}\rho\sigma g(\kappa,s-u)du}
=\frac{1}{P(t,s)} e^{-\frac{\gamma_1(t,s)}{2}}.
\end{aligned}
\end{equation*}
Insert this and \eqref{eq:numercom-7} into \eqref{eq:numercom-6} and let
\[
z=\tilde{\rho}(t,s)y+\sqrt{1-\tilde{\rho}(t,s)^2}\hat{z}.
\]
This transforms the integral in \eqref{eq:numercom-6} into an integral of two independent Gaussian variables
\begin{multline}
\label{eq:numercom-8}
u(t,s,r,x)
=P(t,s)\int_{\R^2}\varphi\left(\frac{x}{P(t,s)} e^{-\frac{\gamma_1(t,s)}{2}+\sqrt{\gamma_1(t,s)}\left(\tilde{\rho}(t,s)y+\sqrt{1-\tilde{\rho}(t,s)^2}\hat{z}\right)},\ q(t,s)+\sqrt{\gamma_2(t,s)}y\right)\\
e^{-\frac{1}{2}({\hat z}^2+ y^2)}\frac{1}{2\pi}d\hat{z}dy.
\end{multline}
Now, letting $\varphi(x, r)=(K-x)^+$ and $\varphi(x,r)=Kr\mathds{1}_{\{x<b(s,r)\}}$, we obtain \eqref{eq:numer-2}-\eqref{eq:numer-3}.

\begin{proof}[Proof of Lemma \ref{lem:aux}]
The proof follows the lines of similar computations in the literature, see \cite{baxter1996financial} and \cite{detemple2002valuation}. Using the explicit expression of $r$ in \eqref{eq:numer-R-ex} and stochastic Fubini's theorem \citep[Theorem IV.64]{Protter}, we compute
\begin{align}\label{eq:numercom-1}
\int_{t}^s r_udu&=g(\kappa,s-t)r_t+\theta\left(s-t-g(\kappa,s-t)\right)+\beta\int_{t}^{s}g(\kappa,s-u)dW_u\\ 
&=\mu(t,s)+\beta\int_{t}^{s}g(\kappa,s-u)dW_u.\nonumber 
\end{align}
Define $Z_t$ implicitly by
\[
dB_t=\rho dW_t+\sqrt{1-\rho^2}dZ_t.
\]
Then $Z$ is a Brownian motion that is independent of $W$. Using the explicit expression of $X$ and \eqref{eq:numercom-1}, we write $u$ as
\begin{equation}
\label{eq:numercom-2}
\begin{aligned}
u(t,s,r,x)
&=\E_{r}\bigg[e^{-\mu(t,s)-\beta\int_{0}^{s-t}g(\kappa,s-t-u)dW_u}\varphi\Big(x \exp\Big\{\mu(t,s)-\frac{\sigma^2}{2}(s-t)\\
&\hspace{50pt}+\beta\int_0^{s-t}g(\kappa,s-t-u)dW_u+\int_{0}^{s-t}\sigma(\rho dW_u+\sqrt{1-\rho^2}dZ_u)\Big\}, r_{s-t}\Big)\bigg].
\end{aligned}
\end{equation}
Define a new measure $\widetilde{\P}$ by the Radon-Nikodym density
\[
\frac{d\widetilde{\P}}{d\P}:=e^{-\beta\int_0^{s-t}g(\kappa,s-t-u)dW_u-\frac{1}{2}\int_0^{s-t}\beta^2g(\kappa,s-t-u)^2du}.
\]
The process $\widetilde{W}$ given by
\[
d\widetilde{W}_t=d W_t+\beta g(\kappa,s-t)dt
\]
is a Brownian motion under $\widetilde{\P}$. We write the explicit formula \eqref{eq:numercom-1} for $r$ in terms of $\widetilde{W}$:
\[
r_s=r_te^{-(s-t)\kappa}+\theta(1-e^{-(s-t)\kappa})-\frac{\beta^2}{2}g(\kappa,s-t)^2+\beta\int_{t}^{s}e^{-(s-u)\kappa}d\widetilde{W}_u.
\]
Denoting by $\widetilde{\E}$ the expectation under $\widetilde{P}$, we obtain from \eqref{eq:numercom-2}
\begin{equation}
\label{eq:numercom-4}
\begin{aligned}
u(t,s,r,x)
&=e^{-\mu(t,s)+\frac{1}{2}\beta^2\int_{t}^{s}g(\kappa,s-u)^2du}\widetilde{\E}\Big[\varphi\Big(x e^{L(t,s)+A(t,s)},\ q(t,s)+Y(t,s)\Big)\Big],
\end{aligned}
\end{equation}
where
\begin{align*}
A(t,s):&=\int_{0}^{s-t}(\rho\sigma+\beta g(\kappa,s-t-u))d\widetilde{W}_u+\int_0^{s-t}\sigma\sqrt{1-\rho^2}dZ_u,\\
Y(t,s):&=\beta\int_{0}^{s-t}e^{-(s-t-u)\kappa}d\widetilde{W}_u.
\end{align*}
For each fixed $t < s$, the random vector $(A,Y)$ is multivariate Gaussian under $\widetilde \P$ with zero mean, and variance and covariance given by
\begin{align*}
&Var_{\widetilde \P}(A(t,s))=\gamma_1(t,s),\qquad Var_{\widetilde \P}(Y(t,s))=\gamma_2(t,s),\\
&Cov_{\widetilde \P}(A(t,s),Y(t,T))=\tilde{\rho}(t,s)\sqrt{\gamma_1(t,s)\gamma_2(t,s)}.
\end{align*}
Hence, we have an explicit integral representation of \eqref{eq:numercom-4}
\begin{equation}
\label{eq:numercom-5}
\begin{aligned}
u(t,s,r,x)=
&=e^{-\mu(t,s)+\frac{1}{2}\beta^2\int_{t}^{s}g(\kappa,s-u)^2du}\int_{\R^2}\varphi\Big(x e^{L(t,s)+z},\quad q(t,s)+y\Big)\\
&\hspace{30pt} \times e^{-\frac{1}{2(1-\tilde{\rho}(t,s)^2)}\left(\frac{z^2}{\gamma_1(t,s)}+ \frac{y^2}{\gamma_2(t,s)}-\frac{2\tilde{\rho}(t,s)zy}{\sqrt{\gamma_1(t,s)\gamma_2(t,s)}}\right)}\frac{1}{2\pi\sqrt{\gamma_1(t,s)\gamma_2(t,s)}\sqrt{1-\tilde{\rho}(t,s)^2}}dzdy.
\end{aligned}
\end{equation}
A change of variable yields \eqref{eq:numercom-6}.
\end{proof}

\section*{Acknowledgements}
C. Cai gratefully acknowledges support by China Scholarship Council. T. De Angelis gratefully acknowledges support by EPSRC Grant EP/R021201/1.
All authors would like to thank Sheng Wang for drawing our attention to gaps in the proofs of Propositions 3.3 and 3.11.

\bibliographystyle{apalike} 
\bibliography{reference}

\end{document}